\documentclass[11pt, letterpaper]{article}

\usepackage{soul}
\usepackage{amsthm}
\usepackage{amssymb}  
\usepackage{amsfonts}
\usepackage{mathrsfs}
\usepackage{mathtools}
\usepackage[bbgreekl]{mathbbol}
\usepackage[noadjust]{cite}

\usepackage{booktabs} 
\usepackage{colortbl}
\usepackage{multirow} 
\usepackage{multicol}

\usepackage{comment}
\usepackage{subfigure}

\DeclareSymbolFontAlphabet{\mathbb}{AMSb}
\DeclareSymbolFontAlphabet{\mathbbl}{bbold}

\newcommand{\bee}{\begin{enumerate}}
\newcommand{\eee}{\end{enumerate}}
\newcommand{\bi}{\begin{itemize}}
\newcommand{\ei}{\end{itemize}}
\newcommand{\bq}{\begin{eqnarray}}
\newcommand{\eq}{\end{eqnarray}}
\newcommand{\bqn}{\begin{eqnarray*}}
\newcommand{\eqn}{\end{eqnarray*}}

\DeclareMathOperator{\diag}{diag}

\usepackage{siunitx}
\AtBeginDocument{\DeclareSIUnit{\MWh}{MWh}}
\AtBeginDocument{\DeclareSIUnit{\kWh}{kWh}}
\usepackage{amsmath}
\usepackage{mathtools}
\usepackage{isomath}
\let\underbrace\LaTeXunderbrace

\theoremstyle{definition}

\newtheorem{remark}{Remark}
\newtheorem{definition}{Definition}
\newtheorem{theorem}{Theorem}

\newtheorem{lemma}[theorem]{Lemma}

\newtheorem{example}{Example}
\usepackage{lipsum}

\newcommand{\slow}[1]   
{ \textcolor{red}{(SL: #1)}} 

\usepackage{etoolbox}
\newtoggle{long}
 \toggletrue{long}

\newtheorem{CustomerAssumption}{Assumption}
\newenvironment{MyAssumption}[1]
  {\renewcommand\theCustomerAssumption{#1}\CustomerAssumption}
  {\endCustomerAssumption}

\begin{document}

\title{Reverse Kron reduction of Multi-phase Radial Network}

\author{Steven H. Low
\thanks{This research was supported by NSF through grants ECCS 1931662.
We thank Stanojev Ognjen of ETH, Lucien Werner and Yiheng Xie of Caltech for helpful discussions.
}
\\
CMS, EE, Caltech, Pasadena CA
\\
slow@caltech.edu
}

\maketitle

\begin{abstract}
We consider the problem of identifying the admittance matrix of a three-phase radial network from voltage and current measurements at a subset of nodes.  These measurements are used to estimate a virtual network represented by the Kron reduction (Schur complement) of the full admittance matrix.  We focus on recovering exactly the full admittance matrix from its Kron reduction, i.e., computing the inverse of Schur complement.  The key idea is to decompose Kron reduction into a sequence of iterations that maintains an invariance structure, and exploit this structure to reverse each step of the iterative Kron reduction. 
\end{abstract}

\newpage
\tableofcontents

\newpage
Main changes:
\bi
\item December 18, 2023: 
	\bi
	\item Rewrote Section \ref{sec:graphBarY; subsec:graphbarY}, mainly
	Theorem \ref{sec:graphBarY; subsec:graphbarY; thm:G(bar Y)}.  Revise 
	Sections \ref{sec:algorithm; subsec:IdentifyY_1111} and 
	\ref{sec:algorithm; subsec:DecomposeMaxCliques} accordingly.
	 Added Sections \ref{sec:algorithm; subsec:ReverseStep1} and 
		\ref{sec:algorithm; subsec:summary}.
	\item Rewrote Section \ref{sec:IterativeKR} to focus on growing and shrinking maximal clique by reversible one-step Kron reduction.
	\item Revise Section \ref{sec:Assumption3} accordingly.
	\ei

\item March 18, 2024: 
	\bi
	\item Added Section \ref{subsec:RelatedWork} on prior work unbalanced three-phase network ID
		and reorganized Section \ref{sec:intro; subsec:summary}.
	\item Fixed notational inconsistencies of $C^{l+1}_{11}$ and $P^l$ in 
		Section \ref{sec:IDMC; subec:reverseKR} pointed out by Yiheng Xie.
	\ei
\item March 19, 2024: Simplify the description of Steps 2 and 4 
	in Algorithm 1 of Sections \ref{sec:algorithm; subsec:DecomposeMaxCliques} and
	Section \ref{sec:algorithm; subsec:CombineMaxCliques} respectively to extract and combine
	maximal cliques (as v7).
\ei
\newpage

\pagestyle{plain}
\pagestyle{myheadings}
\markboth
{Reverse Kron reduction of Multi-phase Radial Network}
{Reverse Kron reduction of Multi-phase Radial Network}


\section{Introduction}
\label{sec:intro}

\subsection{Motivation and summary}
\label{sec:intro; subsec:summary}

Modeling distribution networks below substations is increasingly important as distributed energy resources proliferate on these systems.  Today these networks are sparsely monitored at best, with few
$\mu$PMUs (micro-Phasor Measurement Units) beyond the SCADA (Supervisory Control and Data Acquisition) system at substations and smart meters at utility customers.  
As a result, the utility company often does not have an accurate model of its network as it 
evolves either due to faults, repairs or upgrades.
This limits their ability to analyze power flows and optimize their planning and operations.  
This has motivated a large number of papers to identify the topology, line admittances, or switch status
of distribution grids. A recent tutorial \cite{DekaKekatosCavraro-2024} explains different approaches 
in the literature and contains an extensive list of references.  Various methods have been proposed
for these identification problems using AMI 
data (voltage magnitudes, real and reactive power injections) or PMU data (voltage and current phasors,
real and reactive line flows), measured at all or a subset of nodes, with only passive measurement
or also active probing, for single-phase or unbalanced three-phase 
networks, in radial or mesh topologies (see Table 1 of \cite{DekaKekatosCavraro-2024}).
Most of the literature focuses on identification problems for single-phase networks or 
assumes measurements are available at every node in the network.  
This paper studies the identification of the admittance 
matrix (topology and line admittances) for unbalanced three-phase radial networks from voltage 
and current measurements at a subset of nodes.  This is more realistic for distribution systems.

\paragraph{Summary.}
The phasors of the nodal voltages $V$ and current injections $I$ are related linearly, $I = YV$,
by the complex symmetric admittance matrix $Y$.  When only a subset of the voltage and 
current phasors $(V_1, I_1)$ are measured, they satisfy $I_1 = \bar Y V_1$ where $\bar Y$ is 
called the Kron reduced admittance matrix.   It describes a virtual network topology consisting of only 
measured nodes and the ``effective'' line admittances connecting these nodes.
To identify the full matrix $Y$ from the partial measurement $(V_1, I_1)$, we first estimate 
the Kron reduced admittance matrix $\bar Y$ from $(V_1, I_1)$.  
This corresponds to identification from full measurements and several methods in the literature 
can be applied, e.g., \cite{DekaKekatosCavraro-2024}, 
\cite{LiWengLiaoKeelBrown-2021,YuWengRajagopal-2019, Ardakanian2019}.
Then we identify the unique $Y$ given $\bar Y$, the focus of this paper.

The mathematical problem is as follows.  Consider any complex symmetric 
matrix $Y =: \begin{bmatrix} Y_{11} & Y_{12} \\ Y_{21} & Y_{22} \end{bmatrix}$ with a nonsingular
$Y_{22}$ where $Y_{11}$ describes the connectivity between measured nodes, $Y_{22}$ describes
the connectivity between hidden nodes, and $Y_{12}$ describes the connectivity between measured
and hidden nodes.  The Kron reduced admittance matrix $\bar Y$ is the Schur complement of $Y_{22}$ 
of $Y$.  This defines a mapping from $Y$ to $\bar Y$ given by
\begin{align*}
\bar Y & \ := \ f(Y) \ := \ Y_{11} - Y_{12}Y_{22}^{-1} Y_{21}
\end{align*}
When does the inverse $f^{-1}(\bar Y)$ exist and how to compute it when it does?  In this paper we 
show that $f^{-1}(\bar Y)$ exists when the graph underlying $Y$ is a tree (and when some other conditions hold),
by describing an explicit construction of $Y$ from $(V_1, I_1)$. 

Our construction method extends the method in 
\cite{YuanLow2023} for single-phase radial networks to an unbalanced three-phase setting.
For a single-phase radial network, the series line admittances are always nonzero.  This allows certain 
structural properties important for the identification of $Y$ to be preserved under Kron reduction.  In a 
multi-phase network however line admittances are $3\times 3$ matrices for a three-wire model.  It is 
often unclear when Kron reduction exists or whether these structural properties are still preserved under 
Kron reduction.  This paper follows the same idea in \cite{YuanLow2023} and develops two new results 
to overcome the difficulties in the three-phase setting.  Specifically the construction method consists 
of three main algorithms:
\bee
\item Algorithm 1: Reduces the overall identification problem into the special case of identifying
	the admittance matrix of a single maximal clique consisting of measured leaf nodes
	connected by a tree of hidden nodes (Section \ref{sec:algorithm}).  This procedure is
	the same for both single-phase and unbalanced three-phase networks \cite{YuanLow2023}.
	
\item Algorithm 2: Identifies a maximal clique in an unbalanced three-phase setting 
	(Section \ref{sec:IterativeKR}).  Consider iterative Kron reduction where a single hidden
	node is reduced in each iteration.  The novel idea is to characterize an invariant structure 
	that is preserved in one-step Kron reduction and use it to derive its one-step inverse, 
	yielding an iterative reverse Kron reduction.
	
\item Algorithm 3: Identifies a new hidden node to be added in each iteration of Algorithm 2 
	to the set of identified nodes (Section \ref{sec:Assumption3}).  This new hidden node 
	is adjacent to a 	subset of sibling nodes that have been identified in previous iterations.  
	We generalize a sibling grouping property of \cite{YuanLow2023} that underlies Algorithm 3
	from single-phase to 	three-phase setting under a (restrictive) uniform line assumption.
\eee

\paragraph{Organization and notation.}

We define precisely in Section \ref{sec:NID} the network identification problem addressed
in this paper.  We explain in Section \ref{sec:graphBarY} the graph structures of the admittance
matrix $Y$, to be determined, and its Kron reduction $\bar Y$, which is assumed known.  
We present in Section \ref{sec:algorithm} our overall three-phase identification algorithm
 motivated by the graph structures (Algorithm 1).  This reduces the network identification
 problem into the special case of identifying a single maximal clique.
To identify a single maximal clique, we characterize in Section \ref{sec:IterativeKR}
an invariant structure that is preserved under iterative Kron reduction and use it to derive
a reverse iterative Kron reduction (Algorithm 2).  
Each iteration of Algorithm 2 adds a new hidden node to the set of nodes that have been
identified.  We show in Section \ref{sec:Assumption3} how such a node can be identified 
(Algorithm 3) when the distribution lines are uniform.

We write a column vector in $\mathbb C^n$ either as $x = (x_1, \dots, x_n)$ or 
$x = \begin{bmatrix} x_1 \\ \vdots \\ x_n \end{bmatrix}$.
We use $\textbf 1_k$ to denote the column vector of $k$ 1s and $\textbf I_k$ to denote the identity
matrix of size $k$.  Then $\textbf 1_j \otimes \textbf I_k$ is the $jk \times k$ matrix of $j$ identity
matrices each of size $k$ stacked vertically.  When the dimension is clear from the context we often 
write $\textbf 1$ and $\textbf I$ for $\textbf 1_k$ and $\textbf I_k$ respectively.
If $A_1, \cdots, A_j$ are $j$ matrices each of $k\times k$, then $\diag\left( A_1, \cdots, A_j \right)$
is the $jk\times jk$ block-diagonal matrix with $A_1, \cdots, A_j$ as its diagonal blocks.
Given an admittance matrix $A$, $G(A)$ denotes its underlying graph.  We will often use $A$
to refer to either the admittance matrix, its underlying graph $G(A)$, or the set of nodes in $G(A)$, 
when the meaning should be clear from the context.

\subsection{Prior work}
\label{subsec:RelatedWork}

As mentioned above there is a large literature
on the identification of the topology, line admittances, or switch status of distribution grids for single-phase
networks. We leave the discussion of this literature to the comprehensive tutorial 
\cite{DekaKekatosCavraro-2024} and summarize below only papers on identification
problems for unbalanced three-phase networks.
There are papers on topology and line parameter identification using active probing.  
For example, \cite{CavraroKekatos-2019} identifies the admittance matrix
of a single-phase radial network by perturbing inverter injections and measuring voltage
phasors.  It proves that if injections are perturb at all leaf nodes and voltage responses are
measured at all nodes, then the topology is identifiable by solving a nonconvex problem.
Our summary covers only methods that use passive measurements.
There is also a large literature on line parameter estimation given network topology for
either single-phase or three-phase networks.  
For a recent example, see \cite{GuptaSossanLeBoudecPaolone-2021} that studies
this problem for a three-phase network (radial or mesh) 
from voltage and current phasor measurements at all nodes and proposes a method to
pre-process these measurements in order to improve the performance of Total Least Square 
solution.  Our summary covers only methods that involve topology estimation.

We now summarize several papers 
on the identification of unbalanced three-phase networks
\cite{DekaChertkovBackhaus-2020, BariyaDekaVonMeier-2021, GandluruPoudelDubey-2020, LiaoWengLiuZhaoTanRajagopal-2019, LiWengLiaoKeelBrown-2021}
\cite{YuWengRajagopal-2019, Ardakanian2019, VaninGethDhulstVanHertem-2023, SoltaniMaKhorsandVittal-2023},
with emphasis on papers that explicitly exploit the radial structure of the network 
\cite{DekaChertkovBackhaus-2020, BariyaDekaVonMeier-2021, GandluruPoudelDubey-2020, LiaoWengLiuZhaoTanRajagopal-2019, LiWengLiaoKeelBrown-2021}.

\noindent
\emph{Operational radial networks}.
Consider a given mesh network in which tie switches and sectionalizing switches are configured so 
that the operational network at any time consists of a forest of nonoverlapping trees that span all nodes.  
Often the switch status may not be known accurately due to frequent reconfigurations or manual changes in
distribution systems.  The problem of estimating the switch status and hence the operational topology is
called topology detection in \cite{DekaKekatosCavraro-2024}.  This problem usually assumes the topology,
line admittances, and switch locations are known. 

Given voltage magnitude measurements, 
\cite{DekaChertkovBackhaus-2020} estimates the operational network by extending
the statistical method of \cite{Deka2018-TCNS} from a single-phase to an unbalanced
three-phase setting.  The key idea is to use linear approximation $v = H_1 p + H_2 q$ 
to relate the voltage magnitudes $v$ to real and reactive power injections $(p,q)$.  
For single-phase networks in \cite{Deka2018-TCNS} this is obtained by linearizing
the single-phase polar form of AC power equation around the flat voltage profile ($|V_i|=1$ and $\theta_i = \theta_j$
for all $i,j$), or equivalently, linearizing the equation $|V_i|^2 - |V_j|^2 = 2(r_{ij} P_{ij} + x_{ij}Q_{ij})$ in
the single-phase DistFlow model of \cite{Baran1989a, Baran1989b}.  The resulting system matrices $H_1, H_2$
are Laplacian matrices of network topology and line conductances and reactances.  This linear 
model $v = H_1 p + H_2 q$ then relates the covariance matrix $\Omega_v := E(v - E(v))(v - E(v))^{\sf T}$ 
and the variances $C(i,j) := E[(v_i-v_j) - E(v_i-v_j)]^2$, for all nodes $i, j$,
to the covariance matrices of $(p, q)$.  
Using this, \cite{Deka2018-TCNS} proves several interesting properties of $\Omega_v$ and $C(i,j)$ and use these
properties to identify a node's unique parent and thus iteratively the operational topology.   
Specifically, if powers at different nodes are nonnegatively correlated, 
\cite{Deka2018-TCNS} proves that $\Omega_v(i,i)$ increases as $i$ moves farther away from the root 
(substation).  Under this assumption and the assumption that the (true) covariance matrices of $(p, q)$ are known,
Algorithm 1 of \cite{Deka2018-TCNS} computes the empirical $\hat\Omega_v$ from given voltage magnitude 
measurements, and uses properties of $\Omega_v$ to identify successively each node's unique parent, 
starting from leaf nodes, and thus each operational tree.
If powers at different nodes are uncorrelated instead, \cite{Deka2018-TCNS} 
proves that for each node $i$, the minimization of $C(i,j)$ over $j$ that is not a descendent of $i$ is attained
uniquely at $i$'s parent.  Under this assumption, Algorithm 2 of \cite{Deka2018-TCNS} identifies each 
operational tree in a similar manner without needing the covariances of $(p,q)$.  Algorithm 3 of \cite{Deka2018-TCNS} 
extends Algorithm 1 to the case where voltage magnitudes are measured only at a given subset of nodes.
For three-phase networks, \cite{DekaChertkovBackhaus-2020} derives a linear approximation (equivalent
to that in \cite{Gan-2014-UnbalancedOPF-PSCC}) $v = H_1p + H_2 q$ that relates three-phase voltage 
magnitudes $v$ and real and reactive powers $(p,q)$.  Instead of voltage magnitudes, it uses measurements 
of voltage phasors at all nodes to identify the operational network.  Instead of covariance properties, it proves 
conditional independence properties of voltage phasors and use it to iteratively identify connectivity among
pairs of nodes.  Since it does not use covariance matrices of voltages, but only their conditional independence,
the algorithm does not require line admittances or power injection statistics.
In contrast, \cite{BariyaDekaVonMeier-2021} extends the covariance properties of \cite{Deka2018-TCNS} from
a single-phase setting to an unbalanced three-phase setting and use them to iteratively identify a node's parent
in a similar way.  
Specifically, instead of using linearized power flow model, \cite{BariyaDekaVonMeier-2021} uses the three-phase
(reduced) impedance matrix $Z$ to relate current phasor $I$ to voltage phasors $V$, i.e., $V = ZI$.  Under the 
assumption 
that current injections are uncorrelated across nodes and phases and that per-phase current phasors all have equal
variance for all nodes and all phases, \cite{BariyaDekaVonMeier-2021} then extends the result of \cite{Deka2018-TCNS} 
to show that, given $i$, $C(i,j)$ is minimized at a node $b$ that is either its parent or its child.
This set of algorithms bear similarity to the recursive grouping algorithm of \cite{ChoiWillsky2011} where the
information distance $d_{ij}$ in \cite{ChoiWillsky2011} is sometimes replaced by ``effective impedance'' between
observed nodes $i$ and $j$.

The three-phase topology detection problem is formulated in \cite{GandluruPoudelDubey-2020} as a
mixed integer linear program. It assumes the three-phase real and reactive line flow measurements
$(\hat P, \hat Q)$ and power injection measurements $(\hat p, \hat q)$ are available. It uses the 
three-phase linear DistFlow model of \cite{Gan-2014-UnbalancedOPF-PSCC}
which relates injections $(p,q)$ directly to line flows $(P, Q)$ (voltages are not needed because 
line losses $|z_{ij}|^2 \ell_{ij}$ are ignored).  The optimization is over both the given set of operational 
topologies and three-phase real and reactive powers $(p, q, P, Q)$ and the cost is the $l_1$ norm of 
the error $\|(p, q, P, Q) - (\hat p, \hat q, \hat P, \hat Q)\|_1$ normalized per-phase by empirical variances,
subject to linear three-phase power flow equations that relate $(p,q)$ to $(P, Q)$.

\noindent
\emph{Radial topology and phase identification}. 
Given a set of network nodes, topology and phase identification is the problem of identifying the lines connecting
these nodes as well as the phase labels at each bus.  Unlike a single-phase network, phase labels may 
not always be available or accurate in a three-phase network and phase identification is often an integral
part of a topology identification problem. 
Reference \cite{LiaoWengLiuZhaoTanRajagopal-2019} identifies the radial network topology and bus phase labels
by extending the graphical-model method of
\cite{Weng2017} from a single-phase to an unbalanced three-phase setting with unknown phase 
identities. A key assumption of \cite{Weng2017} is that the injection current phasors at different nodes are 
statistically independent.  This leads to the important conclusion that the joint distribution 
$p(V) := p(V_1, \dots, V_n)$
of the voltage phasors on a radial network is what is called in \cite{chow1968approximating} a probability distribution of 
\emph{first-order tree dependence},\footnote{Strictly speaking, this assumes $V_j$ takes discrete values, but 
the formulation can be extended to continuous-valued random variables as well.}  i.e., 
\begin{align*}
p(V) \ := \ p(V_1, \dots, V_n) & \ = \ \prod_j p\left(V_j | V_{i(j)} \right)
\end{align*}
where $i(j)$ denotes the unique parent $j$ of $i$ in the radial network (i.e., $i(j)$ is on the unique path from the root 
to $j$).\footnote{If $V_j$ are independent then $p(V) = \prod_j p(V_j)$.  A first-order tree dependence
can hence be interpreted as the minimum amount of dependence among $V_1, \dots, V_n$.}
To see this, let ${\sf T}_j$ denote the subtree of descendants of $j$ rooted at node $j$, including $j$.  
Since we assume shunt admittances are zero, Kirchhoff's current law implies 
\begin{align*}
V_j & \ = \ V_{i(j)} \, + \, z_{j i(j)} \sum_{k\in{\sf T}_j} I_k
\end{align*}
where $z_{jk}$ is the series line impedance of line $(j,k)$.
Therefore, given $V_{i(j)}$ at its parent, $V_j$ depends only on current injections $I_k$ in 
${\sf T}_j$.  Since current injections are independent, $V_j$ is independent of $V_{j'}$ that is 
not a descendant in ${\sf T}_j$, given $V_{i(j)}$.
Order the nodes in the breadth-first-search order starting from the root node 1 so that $j$'s parent
$i(j) \in\{1, \dots, j-1\}$ for all nodes $j$.
Then $p(V_j | V_{j-1}, \dots, V_1) = p\left( V_j | V_{i(j)} \right)$ and hence
\begin{align*}
p(V) & \ = \ p(V_1) p(V_2|V_1) \cdots p\left(V_n | V_{n-1}, \dots, V_1\right)
\ = \ \prod_j p\left(V_j | V_{i(j)} \right)
\end{align*}
i.e., the distribution $p(V)$ is of first-order tree dependence.  Such a distribution is represented in
\cite{chow1968approximating} by a tree, called the dependence tree of $p(V)$, with $n$ nodes 
where two nodes $i$ and $j$ are adjacent if and only if $i = i(j)$ is the unique parent of $j$. 
Consider the complete graph with $n$ nodes and label each edge
$(j,k)$ by the mutual information $I(V_j, V_k) := E_{p(V_j, V_k)}\left( \log\frac{p(V_j, V_k)}{p(V_j)p(V_k)}\right)$.
Then \cite{chow1968approximating} shows that a spanning tree of the complete graph is the dependence 
tree of $p(V)$ if and only if it is a maximum-weight spanning tree.  This means that
the radial network is one of the maximum-weight spanning trees.  When the weights are distinct, the 
minimum-weight spanning tree is unique and hence an identification algorithm based on minimum-weight 
spanning tree will produce the true network topology.  Otherwise, it may identify a different dependence tree 
than the radial grid that has the same distribution $p(V)$.  These results motivate the identification algorithm of 
\cite{Weng2017} for single-phase radial networks and its extension in 
\cite{LiaoWengLiuZhaoTanRajagopal-2019} for unbalanced three-phase radial networks.
Specifically, for single-phase networks, given voltage phasor measurements $\hat V_j(t)$ at each node $j$ 
at each time $t$, \cite{Weng2017} computes the empirical mutual information $I(\hat V_j, \hat V_k)$ and uses
them as edge weights of the complete graph with $n$ nodes.  Then it computes the maximum-weight
spanning true using the standard Kruskal's algorithm.
For three-phase networks, \cite{LiaoWengLiuZhaoTanRajagopal-2019} defines the vector
$\Delta V_j(t) := V_j(t) - V_j(t-1)\in\mathbb C^3$ of incremental voltage changes across successive measurements
and computes the empirical mutual information $I(\Delta \hat V_j, \Delta \hat V_k)$ of the incremental changes.  It
applies the same algorithm to the complete graph with edge weights $I(\Delta \hat V_j, \Delta \hat V_k)$.

\noindent
\emph{Radial admittance matrix identification}. 
The problem of admittance matrix (or impedance matrix) identification is the problem of determining 
topology and line admittances (or impedances) from measurements.  This problem is studied in 
\cite{LiWengLiaoKeelBrown-2021} for both single-phase and unbalanced three-phase radial networks  
when measurements are available only at a subset of the nodes.  
For single-phase radial networks, it uses the linear relation $V = ZI$ and supposes voltage phasors $V_M$ and current
phasors $I_M$ are measured at a subset $M$ of the buses. It identifies the impedance matrix $Z$ in
two steps.  First it estimates the Kron reduced impedance matrix $\bar Z$ as an optimal solution of 
$\min_Z \| V_M - Z I_M\|_F$.  Then it makes use of the fact that, for a radial network, $Z_{ij}$ is the 
sum of the line impedances on the common segment of the path from the reference bus to node $i$ 
and the path from the reference bus to $j$.  This allows the use of ``effective impedance'' as the 
information state $d_{ij}$ in \cite{ChoiWillsky2011} between two measured nodes $i$ and $j$, and 
hence the use of the recursive grouping algorithm of \cite[Section 4]{ChoiWillsky2011} to both determine 
the tree topology of the hidden nodes and their line impedances, yielding the impedance matrix $Z$.\footnote{The
same method is used in \cite{PengwahFangRazzaghiAndrew-2022} to identify a radial topology from
AMI data without phase angles.}
For three-phase radial networks, \cite{LiWengLiaoKeelBrown-2021} uses the (inverse of the) 
linear approximation of three-phase power flow model of \cite{WangZhangLiYangKang-2017} 
that relates voltage phasors $V$ and
complex powers.  It claims without proof that the same method for single-phase grids applies directly 
to this model.

We study the same problem in this paper but determines the admittance matrix $Y$ instead
of its inverse $Z$.  We use the three-phase relation $I = YV$ and develops a different method
to identify $Y$ from its Kron reduction $\bar Y$.

\noindent
\emph{Mesh networks}. 
Besides radial networks, identification problems are studied in
\cite{YuWengRajagopal-2019, Ardakanian2019, VaninGethDhulstVanHertem-2023, SoltaniMaKhorsandVittal-2023}
for unbalanced three-phase mesh networks.  All these methods apply to both single or three-phase systems, radial or mesh. 
An AC power flow equation $s = f(V; Y_u, u)$ is used in \cite{YuWengRajagopal-2019} where $s$ 
are the complex power injections, $V$ are the voltage phasors, $u$ is the system state that takes one of
$K$ values, representing, e.g., different topologies, and $y_u$ is the vector of line admittances of the 
network when the system is in state $u$.  Hence $y_u$ contains both topology and line parameters in each state $u$.
By defining certain vector $x := x(s, V)$ and matrix $X := X(s, V)$ as functions of $(s, V)$, \cite{YuWengRajagopal-2019}
treats the power flow equation as a linear mapping from $X$ to $x$ of the form $x = Xy_u$ for each $u$.
Given $(s,V)$, $(x, X)$ will be known quantities.  The paper assumes an error model where the observed
quantities $(\hat x, \hat X)$ are given by $\hat x = x + \epsilon_x$ and $\hat X = X + \epsilon_X$ with additive
Gaussian noise $(\epsilon_x, \epsilon_X)$.
To identify the topology and line admittances, it solves a maximum likelihood estimation: given historical
data $(\hat s_t, \hat V_t, t = 1,\dots, T)$, maximize the $\log$ likelihood 
$L(\hat x(\hat s_t, \hat V_t), \hat X(\hat s_t, \hat V_t), \forall t \, | \, x, X; y_u, u)$ over the variables $(x, X; y_u, u)$, 
subject to $x = X y_u$.
The model $I = YV$ is used in \cite{Ardakanian2019} where $(I, V)$ are the three-phase current injection
phasors and voltage phasors and $Y$ is the three-phase admittance matrix.  Given current and voltage 
measurements $(\hat I^T, \hat V^T)$ over $T$ sampling periods, \cite{Ardakanian2019} estimates $Y$ by minimizing
the Frobenius error $\| Y\hat V^T - \hat I^T \|_F$ subject to symmetry and sparsity requirements on $Y$.
Since the problem is nonconvex due to the sparsity requirement, it solves instead its convex relaxation
(lasso) and proposes several solution techniques.
The approaches of \cite{VaninGethDhulstVanHertem-2023, SoltaniMaKhorsandVittal-2023} are similar.
Network topology is assumed known in \cite{VaninGethDhulstVanHertem-2023} and the goal is the joint estimation
of the line impedances $Z$ and state $(I, V)$ from only AMI measurements at customers (voltage magnitudes, real and reactive
power injections).  Network topology and line impedances are assumed known in 
\cite{SoltaniMaKhorsandVittal-2023} but the switch status $u$ on lines with controllable switches, and hence the
operational topology, are unknown.  The goal is to the joint estimation of switch status $u$ and state $(I, V)$ from
a mix of AMI and PMU data.  Let $y$ denote the quantities for which measurements are available and 
$x$ denote the state $(I, V)$ and line impedances $Z$ or switch status $u$. Let $y = f(x)$.  
The joint estimation problems in \cite{VaninGethDhulstVanHertem-2023, SoltaniMaKhorsandVittal-2023} 
are formulated as an optimal power flow problem of the form: given measurements $\hat y$,
$\min_{x} c(\hat y, f(x))$ subject to current and power flow equations $g(x) = 0$ and operational and other constraints
$h(x)\leq 0$.  Here $c$ denotes the cost function.

\section{Network identification problem}
\label{sec:NID}

In this section we describe our model of a multiphase radial network and define the network identification
problem.  To simplify notation, we assume all nodes and lines are three-phased,
though all results describe in this report extend with straightforward modifications to the case where some 
of the nodes or lines are single-phased or two-phased.

\subsection{Admittance matrix $Y$}

Consider a network ${G}\coloneqq({N},{E})$ with $n$ nodes where
$ N := \{1, \dots, N\}$ is the set of nodes and $ E\subseteq  N\times  N$ is the set of lines.
We use $N$ to denote both the set and the number of nodes; the meaning should be
clear from the context.
We assume the network $ G$ is three-phased,
where each line is characterized by a $3\times 3$ series admittance matrix 
${y}_{jk}\in\mathbb C^{3\times 3}$ and the admittance matrix $Y$ of $G$ is a 
$3N\times 3N$ matrix.  We assume shunt admittances are zero.
We will refer to the rows (columns) $3j-2$, $3j-1$, $3j$ that are associated with node $j$
as the $j$th row block ($j$th column block).  We use $Y[j,k] \in\mathbb C^{3\times 3}$ to denote
the $3\times 3$ submatrix consisting of the $j$th row block and the $k$th column block.  
Then the admittance matrix $Y$ is defined by 
$Y[j,k] = - y_{jk}\neq 0$ if $(j,k)\in E$, $Y[j,j] = \sum_{k:(j,k)\in E}\, y_{jk}$, and $Y[j,k] = 0$ otherwise.

\begin{definition}[Block symmetry and block row sum]
\label{ch:mun.1; sec:bim; subsec:VI; def:BlockSymmetry}
Given a matrix $A\in\mathbb C^{3n\times 3n}$, partition it into $n\times n$ blocks of $3\times 3$ 
submatrices.  Denote by $A_{jk} = A[j,k] \in\mathbb C^{3\times 3}$ its $jk$th submatrix.  
\begin{enumerate}
\item $A$ is called \emph{block symmetric} if $A_{jk} = A_{kj}$ for all $j, k = 1, \dots, n$.
\item $A$ is said to have zero \emph{row-block sums} if 
	$\sum_k A_{jk} = 0$ for all $j=1, \dots, n$.
\end{enumerate}
\end{definition}
A matrix can be symmetric but not block symmetric, and vice versa.  
A symmetric matrix $A$ is block symmetric if, in addition, all its off-diagonal blocks are themselves
symmetric, i.e., $A_{jk}^{\sf T} = A_{jk}$, for all $j\neq k$.
A block symmetric $A$ is symmetric if, in addition, all blocks $A_{jk}$, including the
diagonal blocks, are symmetric.  
If a matrix has zero row-block sums, then all its row 
sums are zero, but the converse may not hold.  

We make the following assumption on the network $G$ and its line admittance matrices $y_{jk}$.
\begin{MyAssumption}{1}[Line admittances]
\label{Assumption:y_jk}
We assume the network $ G$ is radial (i.e., with tree topology) and connected.
For all lines $(j,k)\in E$, we assume that shunt admittances are zero and the series admittance matrices $y_{jk}$ satisfy:
\begin{enumerate}
\item $y_{jk} = y_{kj}\in\mathbb C^{3\times 3}$ so that $Y$ is block symmetric.
\item $y_{jk}$ are symmetric so that $Y$ is also symmetric, i.e., $Y^{\sf T} = Y$.
\item Re$\left(y_{jk} \right)\succ 0$. 
\end{enumerate}
\qed
\end{MyAssumption}
With Assumption \ref{Assumption:y_jk}.1 the line admittance $y_{jk}$ can model transmission or distribution lines or three-phase transformers in $YY$ or $\Delta\Delta$ configuration, but not transformers in $\Delta Y$ or $Y\Delta$ configuration where the series admittances $y_{jk}$ and $y_{kj}$ may be different.
Assumption \ref{Assumption:y_jk}.3 can be replaced by: for all lines $(j,k)\in E$, Im$(y_{jk}^s)\prec 0$.
Assumption 1.3 is typically satisfied by transmission and distribution lines as well as three-phase transformers in the $YY$ configuration, but not 
the $\Delta\Delta$ configuration whose equivalent series admittance matrix is singular due to conversion matrices \cite{Low2022}.

Let $A \subsetneq  N$ and $Y_A$ be the $3|A|\times 3|A|$ principal submatrix of $Y$ consisting of row
 and column blocks $Y[j,k]$ with $j, k\in A$.  Since our admittance matrix $Y$ has zero row-block sums, it is not invertible. When $A$ is a \emph{strict} subset of $ N$, $Y_A$ is invertible under 
 Assumption \ref{Assumption:y_jk}. The Schur complement $Y/Y_A$ of $Y$ is not invertible
since $Y/Y_A$ also has zero row-block sums.  
Suppose $Y_A := \begin{bmatrix} B_{11} & B_{12} \\ B_{12}^{\sf T} & B_{22} \end{bmatrix}$ with an invertible
$B_{22}$. Then the Schur complement $Y_A/B_{22}$ of $B_{22}$ of $Y_A$ is 
$Y_A/B_{22} := B_{11} - B_{12}B_{22}^{-1} B_{12}^{\sf T}$ (we will study Schur complement in more
detail in Section \ref{sec:IterativeKR}).   These and other properties, from \cite{Low2022}, 
are summarized in the next lemma which is fundamental to our identification method.
\begin{lemma}[\cite{Low2022}]
\label{lemma:InvertibleSubmatrix}
Suppose Assumption \ref{Assumption:y_jk} holds.   
\begin{enumerate}
\item For any line $(j,k)\in E$, $y_{jk}^{-1}$ exists and is complex symmetric. 
    Moreover Re$\left( y_{jk}^{-1} \right)\succ 0$.
\item For any strict subset $A \subsetneq  N$, $Y_{A}^{-1}$ exists and 
    is complex symmetric.  Moreover both Re$(Y_{A})\succ 0$ and 
    Re$\left( Y_A^{-1} \right)\succ 0$.
\item $\left( Y_A/B_{22}\right)^{-1}$ exists and is complex symmetric.  Moreover 
    both Re$(Y_{A}/B_{22})\succ 0$ and Re$\left(( Y_A/B_{22})^{-1} \right)\succ 0$.
\end{enumerate}
\end{lemma}

\begin{remark}[Assumption \ref{Assumption:y_jk}]
As made precise in Lemma \ref{lemma:InvertibleSubmatrix}, the importance of 
Assumption \ref{Assumption:y_jk} is that it allows us to take inverse of any principal
submatrix of the admittance matrix $Y$, justifying arbitrary and successive Kron reductions.  It also allows
certain structural properties to be preserved under Kron reduction, which underlies our results in 
Sections \ref{sec:IterativeKR} and \ref{sec:Assumption3}.
\qed
\end{remark}

\subsection{Hidden nodes and Kron reduction $\Bar Y$}

We assume there are two types of nodes, i.e., $N =: M \cup H$. Nodes $j\in M$ are called \emph{measured nodes} whose three-phase nodal
voltage and current injection phasors $(V_j, I_j)\in\mathbb C^6$ are measured.  Nodes $j\in H$ are called \emph{hidden nodes} whose nodal 
voltages and currents are not measured.
We abuse notation and use $M$ and $H$ to denote both the sets and the numbers of measured and hidden nodes, so $N = M+H$.
The nodes are labeled such that the first $M$ nodes are measured and the last $H$ buses are hidden. We partition the admittance matrix into four sub-matrices accordingly:
\begin{equation}
\label{eq:Y.1}
    Y =:  \begin{bmatrix} Y_{11} & Y_{12} \\ Y_{21} & Y_{22} \end{bmatrix}
 \end{equation}
where $Y_{11} \in \mathbb C^{3M\times 3M}$ specifies the connectivity between measured nodes, $Y_{22} \in\mathbb C^{3H\times 3H}$ specifies the connectivity between hidden nodes, $Y_{12}$ specifies the connectivity between measured and hidden nodes, and $Y_{21} = Y_{12}^{\sf T}$ under Assumption \ref{Assumption:y_jk}.   Partition the voltage and current phasors accordingly: $(V_1, I_1)$ correspond to the voltages and currents of measured nodes and
$(V_2, I_2)$ those of hidden nodes.  If the current injections $I_2$ at hidden nodes are zero, then the network model is:
\begin{equation*}
    \begin{bmatrix}
    I_1 \\ 0
    \end{bmatrix}
    = 
    \begin{bmatrix}
    Y_{11} & Y_{12} \\
    Y_{21} & Y_{22}
    \end{bmatrix}
    \begin{bmatrix}
    V_1 \\ V_2
    \end{bmatrix}
\end{equation*}
Under Assumption \ref{Assumption:y_jk}, Lemma \ref{lemma:InvertibleSubmatrix}
implies that $Y_{22}$ is nonsingular and hence we can eliminate $V_2$ by computing the Schur complement $Y/Y_{22}$ of $Y_{22}$
of $Y$.  We denote the Schur complement by $\bar Y$:
\begin{subequations}
\begin{align}
\label{eq:BarY}
    \Bar{Y} & \ := \ Y/Y_{22} \ := \ Y_{11} - Y_{12}Y_{22}^{-1}Y_{12}^\top.
\end{align}
and call $\bar Y$ a \emph{Kron-reduced admittance matrix} or a \emph{Kron reduction of $Y$}
because $\bar Y$ relates the voltages and currents at the measured nodes:
\begin{align}\label{eq:observable_model_part}
    I_1 = \Bar{Y} V_1,
\end{align}
The matrix $\bar Y$ is the admittance matrix of a virtual network consisting of only
measured nodes in which two measured nodes are adjacent in $G(\bar Y)$ if and only if there 
is a path consisting of
only hidden nodes that connect them in the original graph $G(Y)$.  Its off-diagonal element $\bar Y[j,k]$ 
represents the ``effective'' admittance
between two nodes $j, k$ in the virtual network.
\end{subequations}
The identification of $\bar Y$ from the measurements $(V_1, I_1)$ is the same as identification from 
full measurements and several methods in the literature can be used, e.g., \cite{DekaKekatosCavraro-2024}, 
\cite{LiWengLiaoKeelBrown-2021,YuWengRajagopal-2019, Ardakanian2019};
see discussions in Section  \ref{subsec:RelatedWork}.
We assume this has been done and focus on identifying $Y$ from a given Kron reduction 
$\bar Y$. 

The relation \eqref{eq:BarY} specifies a mapping from the original admittance matrix $Y$ to its Kron reduction $\Bar Y$.  
In general given a Kron-reduced admittane matrix $\bar Y$, it is not possible to recover the original $Y$ uniquely.
When the network is radial, constructing the inverse mapping from $\bar Y$ to $Y$  
turns out to be possible when the hidden nodes satisfy the following assumption.
\begin{MyAssumption}{2}[Hidden nodes]
\label{Assumption:HiddenNodes}
\begin{enumerate}
\item Hidden nodes have zero injections $I_i=0\in\mathbb C^3,\,\forall i\in{H}$.
\item Every hidden node has node degree at least 3.
\end{enumerate}
\qed
\end{MyAssumption}
\begin{remark}[Assumption \ref{Assumption:HiddenNodes}]
If hidden node has a node degree less than 3, it cannot be identified, i.e. it cannot be distinguished
from its neighbors.  Assumption \ref{Assumption:HiddenNodes} implies that all (degree-1) leaf nodes
must be measured nodes, i.e., all hidden nodes must be ``surrounded by'' measured nodes.

Assumption \ref{Assumption:HiddenNodes} is used in Section \ref{sec:IDMC; subec:reverseKR} for 
reversing Kron reduction (in computing $(y_l, \hat y_l)$ from \eqref{Step3; eq:1}) and in 
Section \ref{sec:Assumption3} for identifying sibling nodes and their common parent
(Theorem \ref{thm:SameHiddenNode}).
\qed
\end{remark}


\subsection{Admittance matrix identification}
\label{sec:NID; subsec:problem}

\paragraph{Network identification problem:}
Given a Kron-reduced admittance matrix $\bar Y$, our goal is to construct the original admittance matrix 
$Y$ that satisfies \eqref{eq:BarY} under Assumptions \ref{Assumption:y_jk} and 
\ref{Assumption:HiddenNodes} (as well as additional assumptions to be introduced later). 
\begin{remark}[Identification of $\bar Y$]
\bee
\item 
We assume that the Kron-reduced admittance matrix $\bar Y$
has been identified from the voltage and current phasors $(V_1, I_1)$ at the measured nodes, e.g.,
through regression using \eqref{eq:observable_model_part}.  
\item 
Identification of $\bar Y$ yields the set $M$ of measured nodes, but not the number nor the identity of 
hidden nodes in $H$, nor the set $E$ of edges.
\eee
\qed
\end{remark}

In the rest of this report, we describe a method to solve the network identification problem.

\section{Graph structures of $Y$ and $\Bar Y$}
\label{sec:graphBarY}

In this section we describe the graph structures of $Y$ and its Kron reduction $\bar Y$ that
motivate our identification strategy.

\subsection{Graph $G(Y)$}
\label{sec:graphBarY; subsec:graphY}

Given the admittance matrix $Y\in\mathbb C^{3N\times 3N}$ of a three-phase network, the 
\emph{graph underlying $Y$} is the graph $ G(Y) := ( M\cup  H,  E_1 \cup  E_2)$ where 
$ M$ and $ H$ are disjoint sets of measured and hidden nodes respectively, and
there is an edge $(j,k) \in  E_1 \cup  E_2$ if and only if the submatrix
$Y[j,k]\in\mathbb C^{3\times 3}$ is nonzero.  Here $ E_1$ denotes the
set of edges between only measured nodes in $ M$, and $ E_2$ denotes the set of edges each of which is incident on \emph{at least one}
 hidden node in $ H$.


We partitions the set $ M$ of measured nodes into \emph{internal measured nodes} that are not adjacent to any hidden nodes and \emph{boundary measured nodes} that are adjacent to some hidden nodes.   
Similarly partition the set $ H$ of hidden nodes into \emph{internal hidden nodes} that are not adjacent to any measured nodes and \emph{boundary hidden nodes} that are adjacent to some measured nodes. 
If we knew the identity of internal and boundary measured and hidden nodes, then we can 
decompose the submatrices $Y_{11}$, $Y_{22}$ and $Y_{12}$ of the full admittance matrix 
$Y$ in \eqref{eq:Y.1} into submatrices accordingly
($Y$ is symmetric under Assumption \ref{Assumption:y_jk} and we hence sometimes only specify
its upper submatrix):
\begin{align}
\label{eq:Y.2}
    Y & \ =: \ \left[ \begin{array}{c | c}
    Y_{11} & Y_{12} \\ \hline Y_{12}^{\sf T} & Y_{22}
    \end{array} \right]   
    \ =: \ \left[ \begin{array}{c c | c c}
    Y_{11,11} & Y_{11, 12} & 0 & 0 
    \\ 
      & Y_{11, 22} & Y_{12, 21} & 0 
    \\ \hline
      & & Y_{22, 11} & Y_{22, 12} 
    \\
      & & & Y_{22, 22} 
    \end{array} \right]
\end{align}
By definition, $ E_1$ is the set of edges in $Y_{11}$ between measured nodes only,
and $ E_2$ is the set of edges in $Y_{12}$ between boundary measured and hidden nodes and 
in $Y_{22}$ between hidden nodes.
Moreover $ E_1$ is partitioned into edges between internal measured nodes (in $Y_{11,11}$), 
between boundary measured nodes (in $Y_{11,22}$), and between internal and boundary measured nodes
(in $Y_{11,12}$).   See Figure \ref{fig:Y1} (the example does not satisfy 
Assumption \ref{Assumption:HiddenNodes}).
	\begin{figure}[htbp]
	\centering
	\includegraphics[width=0.45\textwidth] {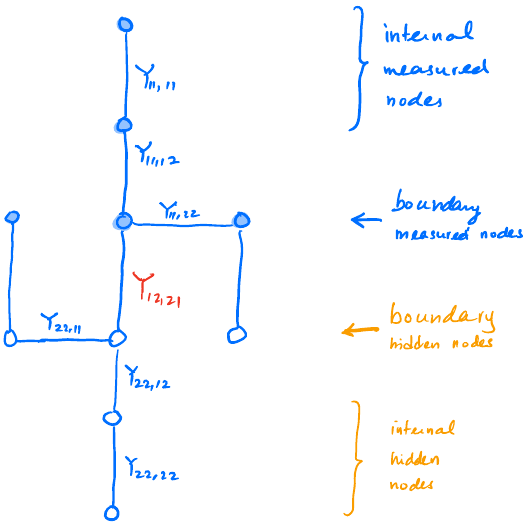}
	\caption{Decomposition of admittance matrix $Y$ according to edges.  	
	Shaded nodes are measured nodes and unshaded nodes are hidden nodes.
}
	\label{fig:Y1}
	\end{figure}
Similarly $ E_2$ is partitioned into edges between boundary measured and hidden nodes (in $Y_{12,21}$),
between boundary hidden nodes (in $Y_{22,11}$), between internal hidden nodes (in $Y_{22,22}$), and
between internal and boundary hidden nodes (in $Y_{22,12}$).  These partitions imply, in particular, that 
$Y_{12}$ has the structure in \eqref{eq:Y.2} with three zero submatrices and a nonzero $Y_{12, 21}$.
This structure is important for identification of $Y$ from $\bar Y$ (see Step 1 in
Section \ref{sec:algorithm; subsec:IdentifyY_1111}).

Of course, given only $\bar Y$, we do not know the structure in \eqref{eq:Y.2}.  We can however
differentiate between internal and boundary measured nodes from the graph structure of $\bar Y$ 
as we next explain.

\subsection{Kron-reduced graph $ G(\Bar Y)$}
\label{sec:graphBarY; subsec:graphbarY}

In this subsection we show that the graph $G(\bar Y)$ underlying the Kron reduced admittance matrix 
$\bar Y$ consists of maximal cliques of size at least 3 connected by subtrees of $G(Y)$.  This allows
us to identify internal and boundary measured nodes.

To describe this precisely, consider an arbitrary connected graph $\hat G$.  A \emph{clique} is a fully 
connected subgraph $C(\hat G)$ of $\hat G$ in which every pair of nodes are adjacent.  Its size is the 
number of nodes in $C(\hat G)$.  A \emph{maximal clique}
$C(\hat G)$ is a clique such that, for every node $j$ in $\hat G\setminus C(\hat G)$, there is a node in 
$C(\hat G)$ that is not adjacent to $j$ in $\hat G$.  
Similarly a \emph{subtree} $T(\hat G)$ of $\hat G$ is a subgraph of $\hat G$ that is a tree. 
A \emph{maximal subtree} of $\hat G$ is a subtree $T(\hat G)$ such that adding another
node in $\hat G\setminus T(\hat G)$ to $\hat G$ results in a subgraph that is not a tree.  
Since $\hat G$ is connected, there is a path in $\hat G$ between any two nodes.  We extend this notion 
to subgraphs and say that two subgraphs $\hat G_1$ and $\hat G_2$ of $\hat G$ are 
\emph{connected through another subgraph $\hat G_3$} if their nodes are connected in $\hat G$ by paths
consisting only of nodes in $\hat G_1, \hat G_2, \hat G_3$.
We say that $\hat G_1$ and $\hat G_2$ are \emph{adjacent} if there is a node $j_1$ in $\hat G_1$ 
and a distinct node $j_2$ in $\hat G_2$ that are adjacent, i.e., $(j_1, j_2)$ is a line in $\hat G$.  
We say that they \emph{overlap} with each
other if there is a node that is in both subgraphs.  

\begin{theorem}[Boundary and internal measured nodes]
\label{sec:graphBarY; subsec:graphbarY; thm:G(bar Y)}
Suppose Assumption \ref{Assumption:y_jk} holds.  Given the Kron-reduced admittance matrix $\bar Y$, 
its underlying graph $ G(\bar Y)$ consists of only maximal cliques and maximal subtrees, i.e., the set $M$ 
of measured nodes of $G(\bar Y)$ can be expressed as 
\begin{align*}
M & \ =:\ \underbrace{ \left( C_1 \cup \cdots \cup C_m \right) }_{M_\text{bnd}} \ \bigcup \
	\left( T_1 \cup \cdots \cup T_n \right) 
\end{align*}
for some $m$ maximal cliques $C_i$ and $n$ maximal subtrees $T_i$ of $G(\bar Y)$ (they are not
necessarily disjoint).\footnote{We abuse notation and use $C_i$ ($T_i$) to denote both
a maximal clique (maximal subtree) of $G(\bar Y)$ and the set of its nodes.}
Moreover
\begin{enumerate}
\item Two maximal cliques are edge-disjoint, i.e., they do not share a line in $G(\bar Y)$.
	They can overlap, be adjacent, or be connected through a maximal subtree of (internal)
	measured nodes in $G(\bar Y)$.
\item A node $j$ is a boundary measured node if and only if it is in a maximal clique, 	i.e., 
	$j\in C_1 \cup \cdots \cup C_m =: M_\text{bnd}$.
\item A node $j$ is an internal measured node if and only if $j \in M_\text{int} := M\setminus M_\text{bnd}$.
\item If Assumption \ref{Assumption:HiddenNodes} holds then every maximal clique $C_i$ is of size at least 3.
\end{enumerate}
\end{theorem}

\begin{proof}
As mentioned after the definition \eqref{eq:BarY} of Kron-reduced admittance matrix $\bar Y$,
two measured nodes are adjacent in $G(\bar Y)$ if and only if they are adjacent in $G(Y)$
or there is a path consisting of \emph{only} hidden nodes that connect them in $G(Y)$.
Consider a maximal subtree ${ T}_i^\text{hid}(Y)$ of the tree $G(Y)$ consisting of only hidden nodes. 
All the boundary measured nodes that are connected to ${ T}_i^\text{hid}(Y)$ in $G(Y)$ are
therefore adjacent to each other in $G(\bar Y)$, i.e., they form a maximal clique $C_i$ of 
$G(\bar Y)$.  Conversely, if nodes $j,k$ are in a maximal clique $C_i$ with more 
than 3 nodes, then both must be connected to the same maximal subtree ${ T}_i^\text{hid}(Y)$
since $G(Y)$ is a tree.
Hence $j$ is a boundary measured node if and only if $j\in M_\text{bnd}$ and is an internal 
measured node otherwise.  

To show that two maximal cliques $C_1$ and $C_2$ are edge-disjoint, suppose they share a common
edge $(j,k)$ in $G(\bar Y)$.  Since $j$ is in both $C_1$ and $C_2$, it must be connected in $G(Y)$
to two disjoint maximal subtrees ${ T}_1^\text{hid}(Y)$ and ${ T}_2^\text{hid}(Y)$ of hidden nodes.
Similarly for node $k$.  This creates a loop in $G(Y)$, a contradiction.  Hence $C_1$ and $C_2$
are edge-disjoint in $G(\bar Y)$.  Obviously two maximal cliques can overlap, be adjacent, 
or be connected through a maximal tree in $G(\bar Y)$ (or through a sequence of maximal trees 
and maximal cliques).

Finally we show that each $C_i$ has at least 3 nodes if Assumption \ref{Assumption:HiddenNodes} holds.  
Let $m_i$ be the number of nodes in $ C_i$ and $h_i$ be the number 
of hidden nodes in  ${ T}_i^\text{hid}(Y)$.  Then the total number of edges in $G(Y)$ incident on the 
nodes in ${ T}_i^\text{hid}(Y)$ is $(h_i-1) + m_i$ since  ${ T}_i^\text{hid}(Y)$ is a tree.
On the other hand, if Assumption \ref{Assumption:HiddenNodes} holds, then every hidden node 
has degree at least 3 and hence the number of edges incident on the nodes in ${ T}_i^\text{hid}(Y)$ is at least
$(3h_i-m_i)/2 + m_i$.  Therefore
\begin{align*}
(h_i-1) + m_i & \ \geq \ (3h_i-m_i)/2 + m_i
\end{align*}
implying $m_i\geq h_i + 2 \geq 3$ as long as $h_i\geq 1$.

\end{proof}

\begin{remark}
\label{subsec:proof; remark:GraphDecomposition}
Given the Kron-reduced admittance matrix $\bar Y$, 
Theorem \ref{sec:graphBarY; subsec:graphbarY; thm:G(bar Y)} allows us to deduce:
\begin{itemize}
\item The number and identity of measured nodes $M := \{1, \dots, M\}$.  
\item The sets $M_\text{int}$ and $M_\text{bnd}$ of interior measured nodes and boundary measured 
nodes respectively so that $M = M_\text{int} \cup M_\text{bnd}$.
\end{itemize}
While every node in a maximal clique $K_i$ is a boundary measured node, a node in a
maximal subtree $T_i$ may be an internal or boundary measured node.
\qed
\end{remark}


\section{Overall identification algorithm}
\label{sec:algorithm}

In this section we describe our method to identify $Y$ from its Kron reduction $\bar Y$ under
Assumptions \ref{Assumption:y_jk} and \ref{Assumption:HiddenNodes}, as well as 
Assumptions \ref{Assumption:BoundaryHiddenNodes} of Section \ref{sec:IterativeKR}
or Assumption \ref{Assumption:UniformLines} of Section \ref{sec:Assumption3} which are required
for identifying a single maximal clique consisting of degree-1 measured nodes connected by a tree of hidden nodes.  This overall Algorithm 1 consists of five steps.  It
reduces the identification problem to the special case where the network is a single maximal
clique.

\subsection{Step 1: Identification of $Y_{11,11}$, $Y_{11, 12}$ and $Y_{11,21}$}
\label{sec:algorithm; subsec:IdentifyY_1111}

As explained in Remark \ref{subsec:proof; remark:GraphDecomposition}, given the Kron-reduced 
admittance matrix $\Bar Y$, we can partition $M$ into the set $M_\text{int}$ of interior measured 
nodes and the set $M_\text{bnd}$ of boundary measured nodes.  We can order the interior measured nodes in $M_\text{int}$ first followed by boundary measured nodes in $M_\text{bnd}$ so 
that the given Kron-reduced admittance matrix $\bar Y$ takes the form:
\begin{align*}
\bar Y & \ =: \ \begin{bmatrix} \bar Y_{11} & \bar Y_{12} \\ \bar Y_{12}^{\sf T} & \bar Y_{22} \end{bmatrix}
\end{align*}
where $\bar Y_{11}$ corresponds to nodes in $M_\text{int}$ and $\bar Y_{22}$ corresponds
to nodes in $M_\text{bnd}$.
As explained in Section \ref{sec:graphBarY; subsec:graphY}, the admittance matrix $Y$, to be determined,
takes the form 
\begin{align}
\label{sec:algorithm; subsec:IdentifyY_1111; eq:Y.2}
    Y & \ =: \ \left[ \begin{array}{c | c}
    Y_{11} & Y_{12} \\ \hline Y_{12}^{\sf T} & Y_{22}
    \end{array} \right]   
    \ =: \ \left[ \begin{array}{c c | c c}
    Y_{11,11} & Y_{11, 12} & 0 & 0 
    \\ 
      & Y_{11, 22} & Y_{12, 21} & 0 
    \\ \hline
      & & Y_{22, 11} & Y_{22, 12} 
    \\
      & & & Y_{22, 22} 
    \end{array} \right]
\end{align}
where $Y_{11,11}$ corresponds to nodes in $M_\text{int}$ and $Y_{11,22}$ corresponds
to nodes in $M_\text{bnd}$.
Under Assumption \ref{Assumption:y_jk}, $Y_{22}$ in \eqref{sec:algorithm; subsec:IdentifyY_1111; eq:Y.2}
is invertible according to Lemma \ref{lemma:InvertibleSubmatrix}.  Define
\begin{subequations}
\begin{align}
Z_{22} & \ := \ Y_{22}^{-1} \ =: \ 
\begin{bmatrix}
X_{22, 11} & X_{22, 12} \\ X_{22, 21} & X_{22, 22}
\end{bmatrix}
\label{sec:algorithm; eq:Step1.1a}
\end{align}
Substituting into \eqref{sec:algorithm; subsec:IdentifyY_1111; eq:Y.2}, we can write 
$\Bar Y = Y_{11} - Y_{12} Z_{22} Y_{21}$ in terms of $X_{22,11}$: 
\begin{align}
\begin{bmatrix} \bar Y_{11} & \bar Y_{12} \\ \bar Y_{21} & \bar Y_{22} \end{bmatrix} & \ := \ 
\begin{bmatrix} Y_{11,11} & Y_{11, 12} \\
Y_{11,21} & Y_{11, 22} \end{bmatrix} \, - \,
\begin{bmatrix} 0 & 0 \\ 0 & 
Y_{12, 21} X_{22, 11} Y_{21,12} \end{bmatrix}
\label{sec:algorithm; eq:Step1.1b}
\end{align}
Since $\Bar Y$ on 
\label{sec:algorithm; eq:Step1.1}
\end{subequations}
the left-hand side of \eqref{sec:algorithm; eq:Step1.1b} is given, we can identify 
\begin{align}
Y_{11, 11} & \ = \ \bar Y_{11},  &   Y_{11, 12} & \ = \ \bar Y_{12}, & Y_{11,21} & \ = \ \bar Y_{21} 	
\label{eq:Solution.1}
\end{align}
 
To identify $Y_{11, 22}$, $Y_{12, 21}$ and $Y_{22}$ in \eqref{sec:algorithm; subsec:IdentifyY_1111; eq:Y.2}, 
we remove all internal measured nodes from the original network.  
The Kron-reduced admittance matrix $\bar Y'$ of the resulting (sub)network can be computed from the given 
$\bar Y$ in \eqref{sec:algorithm; eq:Step1.1}:
\begin{align}
\bar Y' &\ = \ \bar Y_{22} - \diag\left( \left(\textbf 1 \otimes \textbf I_3 \right)^{\sf T} \bar Y_{22} \right)
\label{eq:solution.3}
\end{align}
where $\textbf 1$ denotes a vector of all 1s of appropriate size.  This renormalizes the diagonal blocks of $\bar Y'$
so that $\bar Y'$ has zero row and column-block sums, by removing from these diagonal blocks the series admittances 
$y_{ij}\in\mathbb C^{3\times 3}$ of lines in $ E_1$ between internal and boundary measured nodes.
After the admittance matrix $Y'$ has been identified from its Kron reduction $\bar Y'$ in
Steps 2 through 4 (Sections \ref{sec:algorithm; subsec:DecomposeMaxCliques} through
\ref{sec:algorithm; subsec:CombineMaxCliques}), internal measured nodes can be put back by 
reversing \eqref{eq:solution.3}, i.e., by computing the admittance matrix $Y$ of the original
network from $Y'$ (Section \ref{sec:algorithm; subsec:ReverseStep1}).

In the graph $G(\bar Y')$, two maximal cliques may be disconnected if they
were connected through a maximal subtree of measured nodes before the operation \eqref{eq:solution.3},
as illustrated in the next example.

\begin{example}
\label{sec:algorithm; subsec:IdentifyY_1111; eg:1}
Figure \ref{fig:RemovingInternalNodes} shows a graph $G(Y)$ whose Kron reduction consists of
two maximal cliques $ C_j(\bar Y)$ and $ C_k(\bar Y)$ connected in $G(\bar Y)$ through a maximal 
subtree of \emph{internal} measured nodes. 
	\begin{figure}[htbp]
	\centering
      \subfigure[ ]{
	\includegraphics[width=0.23\textwidth] {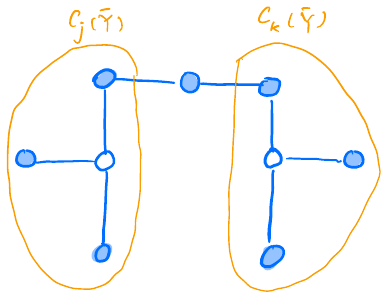} }
    \qquad\quad   
     \subfigure[ ] {
	\includegraphics[width=0.28\textwidth] {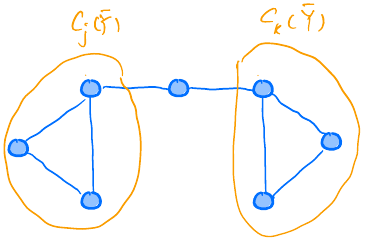} }
    \qquad\quad  
       \subfigure[ ] {
	\includegraphics[width=0.28\textwidth] {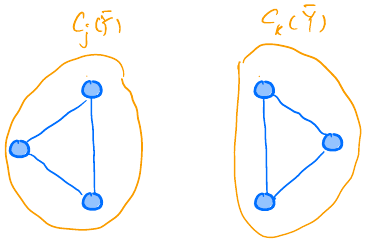} }
	\caption{(a) Original graph $G(Y)$.  (b) Kron reduction $G(\bar Y)$.
	(c) Graph $G(\bar Y')$ after removing internal measured nodes.	
	}
	\label{fig:RemovingInternalNodes}
	\end{figure}
The Kron-reduced matrix $\bar Y$ takes the form ($\times$ denotes nonzero entries):
\begin{align*}
\Bar Y & \ = \ \left[ \begin{array} {ccccc}
\times & \times & & \times &  \\ \cline{2-5}
\times & \multicolumn{1}{|c}{ \cellcolor[rgb]{1,0.5,0.3}{\times} } & \cellcolor[rgb]{1,0.5,0.3} \times & &  \multicolumn{1}{c|} {} \\ 
  & \multicolumn{1}{|c}{ \cellcolor[rgb]{1,0.5,0.3} \times } & \cellcolor[rgb]{1,0.5,0.3} \times & &  \multicolumn{1}{c|} {} \\ 
\times & \multicolumn{1}{|c} {} & & \cellcolor[rgb]{0, 0.6, 1} \times & \multicolumn{1}{c|}{ \cellcolor[rgb]{0, 0.6, 1} \times } \\ 
& \multicolumn{1}{|c}{} & & \cellcolor[rgb]{0, 0.6, 1} \times & \multicolumn{1}{c|} {\cellcolor[rgb]{0, 0.6, 1} \times} \\ \cline{2-5}
	\end{array} \right]
\end{align*}
The submatrix $\bar Y_{22}$ of $\bar Y$ on the left-hand side of \eqref{sec:algorithm; eq:Step1.1b} is
\begin{align*}
\bar Y_{22} & \ =: \ \begin{bmatrix} W_{11} & 0 \\ 0 & W_{22} \end{bmatrix}
\end{align*}
where $W_{11}$, shaded in red, corresponds to the maximal clique 
$ C_j(\bar Y)$ and $W_{22}$ of $\bar Y$, shaded in blue, corresponds to $ C_k(\bar Y)$.  
The first row blocks and column blocks
correspond to the internal measured nodes and their connection to the two maximal cliques. 

After removing the internal measured nodes by performing \eqref{eq:solution.3}, the resulting 
admittance matrix $\bar Y'$ is
\begin{align*}
\bar Y' &\ := \ \bar Y_{22} - \diag\left( \left(\textbf 1 \otimes \textbf I_3 \right)^{\sf T} \bar Y_{22} \right)
	\ =: \ \begin{bmatrix} \bar Y^j & 0 \\ 0 & \bar Y^k \end{bmatrix}
\end{align*}
where $\bar Y^j$ and $\bar Y^k$ are the admittance matrices of the maximal cliques $ C_j(\bar Y')$ 
and $C_k(\bar Y')$ respectively, given by
\begin{align*}
	\bar Y^j & \ := \ W_{11} \ - \ \diag\left( \left(\textbf 1 \otimes \textbf I_3 \right)^{\sf T} W_{11} \right)
\\
	\bar Y^k & \ := \ W_{22} \ - \ \diag\left( \left(\textbf 1 \otimes \textbf I_3 \right)^{\sf T} W_{22} \right)
\end{align*}
i.e., the diagonal blocks of $W_{11}$ and $W_{22}$ are modified so that $\bar Y^j$ and $\bar Y^k$
have zero row and column-block sums.
\qed
\end{example}
Two maximal cliques can also be connected through a line in $G(Y)$ or through a shared measured
node.  These cliques may appear as disconnected, adjacent or overlapping in the Kron reduction
$G(\bar Y')$ after internal measured nodes are removed, as shown in Figure \ref{fig:G(Y)andG(barY')}.
	\begin{figure}[htbp]
	\centering
	\includegraphics[width=0.9\textwidth] {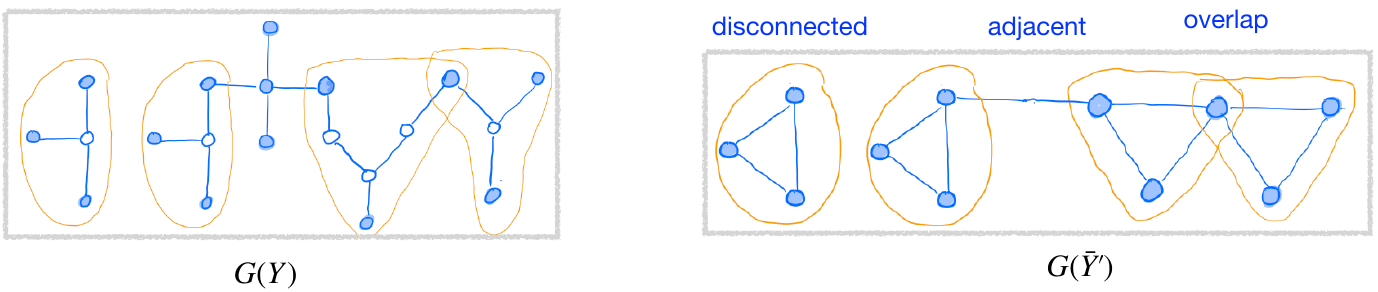} 
	\caption{Original graph $G(Y)$ and Kron reduction $G(\bar Y')$ after internal
	measured nodes are removed.  The resulting admittance matrix $\bar Y'$ takes the form
	in \eqref{fig:G(Y)andG(barY'); eq:barY'} with the 4 maximal cliques ordered left to right.}
	\label{fig:G(Y)andG(barY')}
	\end{figure}
Its admittance matrix $\bar Y'$ takes the form
\begin{align}
\bar Y' & \ = \ \left[ \begin{array} {ccccccc}
\cellcolor[rgb]{0.9,0.9,0.9}{\times} & \cellcolor[rgb]{0.9,0.9,0.9} \times & & & & &   \\ 
\cellcolor[rgb]{0.9,0.9,0.9}{\times} & \cellcolor[rgb]{0.9,0.9,0.9} \times & & & &  & \\ 
& & \cellcolor[rgb]{0.7,0.7,0.7}{\times} & \cellcolor[rgb]{0.7,0.7,0.7} \times & & & \\ 
& & \cellcolor[rgb]{0.7,0.7,0.7} \times  & \cellcolor[rgb]{0.7,0.7,0.7} \times & \times & & \\ 
& & & \times & \cellcolor[rgb]{0, 0.6, 1} \times & \cellcolor[rgb]{0, 0.6, 1} \times & \\ 
& & & & \cellcolor[rgb]{0, 0.6, 1} \times & \cellcolor[rgb]{0, 0.6, 1} \times & \cellcolor[rgb]{0, 0.6, 1} \times \\ 
& & & & & \cellcolor[rgb]{0, 0.6, 1} \times & \cellcolor[rgb]{0, 0.6, 1} \times \\ 
	\end{array} \right]
\label{fig:G(Y)andG(barY'); eq:barY'}
\end{align}

We next compute the Kron-reduced admittance matrix for each maximal clique {in isolation}.

\subsection{Step 2: Maximal-clique decomposition of $Y$}
\label{sec:algorithm; subsec:DecomposeMaxCliques}

We assume all internal measured nodes have been removed in Step 1 through \eqref{eq:solution.3}
and, instead of $\bar Y'$, we denote the admittance matrix of the resulting network by $\bar Y$ in
this and the next two subsections.  

Theorem \ref{sec:graphBarY; subsec:graphbarY; thm:G(bar Y)} in Section \ref{sec:graphBarY; subsec:graphbarY} 
shows that 
the graph $G(\bar Y)$ underlying $\bar Y$ consists of edge-disjoint maximal cliques, denoted by $\{  C_i(\Bar Y) \}$
to emphasize that they are subgraphs of $G(\bar Y)$.
As explained in its proof, the boundary measured nodes in each $ C_i(\Bar Y)$ are connected to each 
other through a tree $ T_i^\text{hid}(Y)$ of hidden nodes in the original graph $ G(Y)$.  
Every boundary measured node in $ C_i(\Bar Y)$ is connected to exactly one hidden node in 
$ T_i^\text{hid}(Y)$ for, otherwise, there exists a loop in $ G(Y)$, contradicting that $ G(Y)$ is a tree.
We will label the boundary measured nodes in $M_\text{bnd}$ so that all nodes in the same maximal clique $C_i(\bar Y)$
are labeled consecutively.  For the example in Figure \ref{fig:G(Y)andG(barY')} the admittance matrix $\bar Y$ takes the
form in \eqref{fig:G(Y)andG(barY'); eq:barY'} with the 4 maximal cliques ordered left to right
(recall that $\bar Y'$ in Figure \ref{fig:G(Y)andG(barY')} and  \eqref{fig:G(Y)andG(barY'); eq:barY'} is relabeled 
as $\bar Y$ in Sections \ref{sec:algorithm; subsec:DecomposeMaxCliques} 
to \ref{sec:algorithm; subsec:CombineMaxCliques}).

We now explain how to iteratively compute the admittance matrix of every maximal clique from the given 
Kron-reduced admittance matrix $\bar Y$.  Let $\tilde Y^0 := \bar Y$.  For iterations $l=1, 2, \dots$, we remove
a single maximal clique from $\tilde Y^{l-1}$.  This yields two graphs, the single maximal clique in isolation that is created in
iteration $l$ whose admittance matrix is denoted $\bar Y^l$ and the remaining subgraph whose admittance
matrix is denoted $\tilde Y^l$. The procedure terminates if $\tilde Y^l$ corresponds to a single maximal clique.  
This procedure will produce the set of all single maximal cliques in isolation, 
$\bar Y^1, \bar Y^2, \dots, \bar Y^{m-1}, \tilde Y^{m-1}$ if the procedure terminates at the end of iteration $m-1$.

The computation of $(\bar Y^l, \tilde Y^l)$ from $\tilde Y^{l-1}$ in each iteration $l=1, 2, \dots$, 
depends on how the maximal clique $G(\bar Y^l)$ was connected to the rest of the subgraph in
$G(\tilde Y^{l-1})$ before it is removed, and there are three cases as shown in 
Figure \ref{fig:G(Y)andG(barY')}.  To simplify description we may abuse notation and use $\bar Y^l$ to 
refer to either the graph or its admittance matrix, and the term ``maximal clique'' to refer to either the
corresponding subgraph of $G(\tilde Y^{l-1})$ before extraction or in isolation after extraction; the 
meaning should be clear from the context.

\vspace{0.1in}\noindent
\textbf{Case 1: Disconnected.} If the maximal clique $\bar Y^l$ created in iteration $l$ was
connected to the original graph $G(Y)$ through a maximal subtree of \emph{internal} measured nodes, then after 
removing the internal measured nodes in Step 1 through \eqref{eq:solution.3}, the maximal clique is disconnected
in the resulting Kron-reduced graph $G(\tilde Y^{l-1})$ (e.g., the leftmost clique in Figure \ref{fig:G(Y)andG(barY')} 
corresponding to the top-left submatrix in \eqref{fig:G(Y)andG(barY'); eq:barY'}).  Hence $\tilde Y^{l-1}$ is block 
diagonal of the form:
\begin{align}
\tilde Y^{l-1} &\ =: \ \begin{bmatrix} \bar Y^l & 0 \\ 0 & \tilde Y^l \end{bmatrix}
\label{eq:barY.case1}
\end{align}
The two graphs created by removing the single maximal clique from $G(\tilde Y^{l-1})$ therefore have as their
admittance matrices the diagonal blocks $(\bar Y^l, \tilde Y^l)$.

\vspace{0.1in}\noindent
\textbf{Case 2: Adjacent through a line in $Y_{11,22}$.} The maximal clique $\bar Y^l$ is
connected by a line in graph $G(\tilde Y^{l-1})$ between two boundary measured nodes
(e.g., the second clique from the left in Figure \ref{fig:G(Y)andG(barY')} corresponding
to the second diagonal block in \eqref{fig:G(Y)andG(barY'); eq:barY'}). Hence $\tilde Y^{l-1}$ 
takes the form:
\begin{align}
\tilde Y^{l-1} & \ = \ \left[ \begin{array} {cccc}
\cellcolor[rgb]{1,0.5,0.3}{\times} & \cellcolor[rgb]{1,0.5,0.3} \times & &  \\ 
\cellcolor[rgb]{1,0.5,0.3} \times  & \cellcolor[rgb]{1,0.5,0.3} \times & \times &  \\ 
& \times & \cellcolor[rgb]{0, 0.6, 1} \times & \cellcolor[rgb]{0, 0.6, 1} \times \\ 
& & \cellcolor[rgb]{0, 0.6, 1} \times & \cellcolor[rgb]{0, 0.6, 1} \times \\ 
	\end{array} \right]
\ =: \ \begin{bmatrix} W_{11} & W_{12} \\ W_{12}^{\sf T} & W_{22} \end{bmatrix}
\label{eq:barY.case2}
\end{align}
where the submatrix $W_{11}$ shaded in red corresponds to the maximal clique $G(\bar Y^l)$
and the submatrix $W_{22}$ shaded in blue corresponds to the remaining graph $G(\tilde Y^l)$.
Their admittance matrices are given by 
\begin{subequations}
\begin{align}
	\bar Y^l & \ := \ W_{11} \ - \ \diag\left( \left(\textbf 1 \otimes \textbf I_3 \right)^{\sf T} W_{11} \right)
\\
	\tilde Y^l & \ := \ W_{22} \ - \ \diag\left( \left(\textbf 1 \otimes \textbf I_3 \right)^{\sf T} W_{22} \right)
\end{align}
i.e., the diagonal blocks $W_{11}$ and $W_{22}$ in \eqref{eq:barY.case2}
are modified so that the matrices $(\bar Y^l, \tilde Y^l)$ have zero row and column-block sums.
\label{eq:barYjbarYk}
\end{subequations}

\vspace{0.1in}\noindent
\textbf{Case 3: Overlap through a shared node.} The maximal clique $G(\bar Y^l)$ is connected 
to the graph $G(\tilde Y^l)$ in $G(\tilde Y^{l-1})$ through a single shared boundary measured node
(e.g., the third clique from the left in Figure \ref{fig:G(Y)andG(barY')} corresponding
to the third diagonal block in \eqref{fig:G(Y)andG(barY'); eq:barY'}). Hence $\tilde Y^{l-1}$ 
takes the form:
\begin{align}
\tilde Y^{l-1} & \ = \ \left[ \begin{array} {cccc}
\cellcolor[rgb]{1,0.5,0.3}{\times} & \cellcolor[rgb]{1,0.5,0.3} \times &  \cellcolor[rgb]{1,0.5,0.3} \times  &  \\ 
\cellcolor[rgb]{1,0.5,0.3} \times  & \cellcolor[rgb]{1,0.5,0.3} \times &  \cellcolor[rgb]{1,0.5,0.3} \times  &  \\ 
 \cellcolor[rgb]{1,0.5,0.3} \times &  \cellcolor[rgb]{1,0.5,0.3} \times  & \cellcolor[rgb]{0.5, 0.6, 0.6} \times & \cellcolor[rgb]{0, 0.6, 1} \times \\ 
& & \cellcolor[rgb]{0, 0.6, 1} \times & \cellcolor[rgb]{0, 0.6, 1} \times \\ 
	\end{array} \right]
\label{eq:barY.case3}
\end{align}
where the submatrix $W_{11}$ shaded in red corresponds to the maximal clique $\bar Y^l$
and the submatrix $W_{22}$ shaded in blue corresponds to the remaining graph $(\tilde Y^l)$. Their admittance 
matrices $(\bar Y^l, \tilde Y^l)$ are given by \eqref{eq:barYjbarYk}.

We next identify each maximal clique in isolation.

\subsection{Step 3: Identification of all maximal cliques in isolation}
\label{sec:algorithm; subsec:IdentifyMaxClique}

Focus now on a single maximal clique $\bar Y^l$ \emph{in isolation} computed in Step 2.  It is the 
Kron reduced admittance matrix of an underlying tree consisting of (degree-1) boundary measured nodes 
(those in $\bar Y^l$) connected through a maximal tree of hidden nodes, without internal measured nodes.  
Step 3 identifies the admittance matrix $Y^l$ of this underlying tree.  This is the core of the 
identification algorithm and will be developed in Sections  \ref{sec:IterativeKR} and \ref{sec:Assumption3}; 
see Algorithm 2 in Section \ref{sec:IDMC; subec:reverseKR}.

Algorithm 2 in Section \ref{sec:IDMC; subec:reverseKR} is repeated to identify the admittance 
matrix $Y^l$ of every maximal clique $\bar Y^i$ in isolation.

\subsection{Step 4: Combining maximal cliques}
\label{sec:algorithm; subsec:CombineMaxCliques}

Denote by $\{Y^1, Y^2, \dots, Y^{m}\}$ the set of admittance matrices of $m$ maximal cliques in isolation
computed in Step 3.  We now explain how to iteratively combine them into the admittance matrix $Y$
of the overall network (without internal measured nodes).

As in Section \ref{sec:algorithm; subsec:DecomposeMaxCliques} we will iteratively add each maximal
clique to create $Y$.   Let $\tilde Y^1 := Y^1$.  For iterations $l = 2, 3, \dots, m-1$, the maximal clique
$G(Y^l)$ is added to $G(\tilde Y^{l-1})$ to produce the combined graph $G(\tilde Y^l)$ whose admittance matrix
is denoted $\tilde Y^l$.  This procedure terminates at the end of iteration $l=m$ to produce the admittance matrix 
$Y := \tilde Y^{m}$ of the underlying tree.  The computation of $\tilde Y^l$ in each iteration $l$ depends on how
the maximal clique $G(\bar Y^l)$ was connected in $G(\bar Y)$ and there are three cases, as explained in 
Section \ref{sec:algorithm; subsec:DecomposeMaxCliques}.  To simplify description, we may sometimes 
refer to the tree $G(Y^l)$ or its admittance matrix $Y^l$ as a maximal clique when we mean $G(\bar Y^l)$ 
and $\bar Y^l$ respectively.  

Consider iteration $l = 2, 3, \dots, m$, where the maximal clique $Y^l$ and $\tilde Y^{l-1}$ are combined to
produce $\tilde Y^l$:

\vspace{0.1in}\noindent
\textbf{Case 1: Disconnected.}  
Recall that all internal measured nodes have been removed in Step 1.
As explained in Section \ref{sec:algorithm; subsec:DecomposeMaxCliques}, the maximal clique $Y^l$ will
be disconnected from $\tilde Y^{l-1}$ if $G(Y^l)$ is connected in the
original graph $G(Y)$ by a maximal subtree of {internal} measured nodes.
Suppose $Y^l$ computed in Step 3 of Section \ref{sec:algorithm; subsec:IdentifyMaxClique} and $\tilde Y^{l-1}$
take the form:
\begin{align}
	Y^l & 
    \ =: \ \left[ \begin{array}{c | c c}
      \textcolor[rgb]{1,0,0}{ Y^l_{11, 22} } & Y^l_{12, 21} & 0 
    \\ \hline
      & Y^l_{22, 11} & Y^l_{22, 12} 
    \\
      & & Y^l_{22, 22} 
    \end{array} \right], &
	\tilde Y^{l-1} & 
    \ =: \ \left[ \begin{array}{c | c c}
      \textcolor[rgb]{0,0,1}{ Y^{l-1}_{11, 22} } & Y^{l-1}_{12, 21} & 0 
    \\ \hline
      & Y^{l-1}_{22, 11} & Y^{l-1}_{22, 12} 
    \\
      & & Y^{l-1}_{22, 22} 
    \end{array} \right]
\label{eq:YjYk}
\end{align}
where \textcolor[rgb]{1,0,0}{ $Y^l_{11, 22}$ }, \textcolor[rgb]{0,0,1}{ $Y^{l-1}_{11, 22}$ } correspond to 
boundary measured nodes and $Y^l_{22}, Y^{l-1}_{22}$ correspond to the hidden nodes that connect 
them in $G(Y)$.  
Then the admittance matrix $\tilde Y^l$ of the combined graph is
	\begin{align*}
    \tilde Y^l & \ := \ \left[ \begin{array}{c c | c c c c}
      {\textcolor[rgb]{1,0,0}{ Y^l_{11, 22} } } & 0 & Y^l_{12, 21} & 0 & 0 & 0
    \\
    0 & \textcolor[rgb]{0,0,1}{ Y^{l-1}_{11, 22} } & 0 & Y^{l-1}_{12, 21} & 0 & 0
    \\ \hline
      & & Y^l_{22, 11} & 0 & Y^l_{22, 12} & 0  \\
      & & 0 & Y^{l-1}_{22, 11}  & 0 & Y^{l-1}_{22, 12}  \\
      & & & & Y^l_{22, 22} & 0 \\
      & & & & 0 & Y^{l-1}_{22, 22}
    \end{array} \right]
\ =: \ \begin{bmatrix} Y_{11} & Y_{12} \\ Y_{12}^{\sf T} & Y_{22}  \end{bmatrix}
	\end{align*} 
where the submatrix $\tilde Y^l_{11}$ of $\tilde Y^l$ has the same structure as that in \eqref{eq:barY.case1}.

\vspace{0.1in}\noindent
\textbf{Case 2: Adjacent through a line in $Y_{11,22}$.}  In this case, the maximal clique $Y^l$ is connected
to $\tilde Y^{l-1}$ by a line in $Y_{11, 22}$ between two boundary measured nodes, as explained in 
Section \ref{sec:algorithm; subsec:DecomposeMaxCliques}.
Suppose $Y^l$ computed in Step 3 of Section \ref{sec:algorithm; subsec:IdentifyMaxClique} and $\tilde Y^{l-1}$
are given by \eqref{eq:YjYk}.
Then the admittance matrix $\tilde Y^l$ of the combined graph is
	\begin{align*}
\tilde Y^l & \ := \ \left[ \begin{array}{c c | c c c c}
    {\textcolor[rgb]{1,0,0}{ \hat Y^l_{11, 22} } } & W_{12} & Y^l_{12, 21} & 0 & 0 & 0
    \\
    W_{12}^{\sf T}  & \textcolor[rgb]{0,0,1}{ \hat Y^{l-1}_{11, 22} } & 0 & Y^{l-1}_{12, 21} & 0 & 0
    \\ \hline
    & & Y^l_{22, 11} & 0 & Y^l_{22, 12} & 0  \\
    & & 0 & Y^{l-1}_{22, 11}  & 0 & Y^{l-1}_{22, 12}  \\
    & & & & Y^l_{22, 22} & 0 \\
    & & & & 0 & Y^{l-1}_{22, 22}
    \end{array} \right]
\ =: \ \begin{bmatrix} Y_{11} & Y_{12} \\ Y_{12}^{\sf T} & Y_{22}  \end{bmatrix}
	\end{align*} 
where the submatrix $\tilde Y^l_{11}$ of $\tilde Y^l$ has the same structure as that in \eqref{eq:barY.case2}.
The submatrices \textcolor[rgb]{1,0,0}{ $\hat Y^l_{11, 22}$ } and \textcolor[rgb]{0,0,1}{ $\hat Y^{l-1}_{11, 22}$ } 
are obtain from \textcolor[rgb]{1,0,0}{ $Y^l_{11, 22}$ } and \textcolor[rgb]{0,0,1}{ $Y^{l-1}_{11, 22}$ } 
and $W_{12}$ (given in \eqref{eq:barY.case2}) by modifying the diagonal entries of 
\textcolor[rgb]{1,0,0}{ $Y^l_{11, 22}$ } and \textcolor[rgb]{0,0,1}{ $Y^{l-1}_{11, 22}$ } so that $\tilde Y^l$ has zero 
row and column-block sums, i.e., 
\begin{align*}
{\textcolor[rgb]{1,0,0}{ \hat Y^l_{11, 22} }} & \ := \
{\textcolor[rgb]{1,0,0}{ Y^l_{11, 22} } }  \, - \, \diag\left( W_{12} \left(\textbf 1 \otimes \textbf I_3 \right) \right),
&
\textcolor[rgb]{0,0,1}{ \hat Y^{l-1}_{11, 22} } & \ := \ 
\textcolor[rgb]{0,0,1}{ Y^{l-1}_{11, 22} }  \, - \, \diag\left( \left(\textbf 1 \otimes \textbf I_3 \right)^{\sf T} W_{12} \right)
\end{align*}

\vspace{0.1in}\noindent
\textbf{Case 3: Overlap through a shared node.}  
In this case, the maximal clique $Y^l$ is connected to $\tilde Y^{l-1}$ through a shared boundary measured node, 
as explained in Section \ref{sec:algorithm; subsec:DecomposeMaxCliques}.
Suppose without loss of generality that the last row and column blocks of $Y^l$ and the first row and column 
blocks of $\tilde Y^{l-1}$ correspond to the shared node.
Suppose $Y^l$ computed in Step 3 of Section \ref{sec:algorithm; subsec:IdentifyMaxClique} and $\tilde Y^{l-1}$
take the form:
\begin{align*}
	Y^l & 
    \ =: \ \left[ \begin{array}{c c | c c}
      \textcolor[rgb]{1,0,0}{ \hat Y^l_{11, 22} } &  \textcolor[rgb]{1,0,0}{ \hat r^{l \sf T}_{11, 22} } 
      &  Y^l_{12, 21} & 0 \\
      \textcolor[rgb]{1,0,0}{ \hat r^{l}_{11, 22} } & \textcolor[rgb]{1,0,0}{ d^l_{11,22} }
      & r^l_{12, 21} & 0
    \\ \hline
      & & Y^l_{22, 11} & Y^l_{22, 12} 
    \\
      & & & Y^l_{22, 22} 
    \end{array} \right]
\\
	\tilde Y^{l-1} & 
    \ =: \ \left[ \begin{array}{c c | c c}
    \textcolor[rgb]{0,0,1}{ d^{l-1}_{11,22} } & \textcolor[rgb]{0,0, 1}{ \hat r^{l-1}_{11, 22} } &  r^{l-1}_{12, 21} & 0 \\
     \textcolor[rgb]{0,0,1}{ \hat r^{l-1 \sf T}_{11, 22} } & \textcolor[rgb]{0,0,1}{ \hat Y^{l-1}_{11, 22} } &Y^{l-1}_{12, 21} & 0
    \\ \hline
     & & Y^{l-1}_{22, 11} & Y^{l-1}_{22, 12} 
    \\
     & & & Y^{l-1}_{22, 22} 
    \end{array} \right]
\end{align*}
where $(\hat r^l_{11,22}, d^l_{11,22}, r^l_{12, 21}, 0)$ and 
$(\hat r^{l-1}_{11,22}, d^{l-1}_{11,22}, r^{l-1}_{12, 21}, 0)$ 
are the row blocks corresponding to the shared node (see an example in \eqref{eq:barY.case3} in Kron reduced form).
%
To construct the admittance matrix $\tilde Y^l$ of the combined network, the row blocks in $Y^l$ and $\tilde Y^{l-1}$ 
corresponding to this shared node is combined into a single row block:
\begin{align*}
\left[ \begin{array}{ c c c | c c c c }
\textcolor[rgb]{1,0,0}{ \hat r^{l}_{11, 22} } & 
\textcolor[rgb]{1,0,0}{ d^l_{11,22} } +  \textcolor[rgb]{0,0,1}{ d^{l-1}_{11,22} } &
\textcolor[rgb]{0,0,1}{ \hat r^{l-1}_{11, 22} } & r^l_{12, 21} & r^{l-1}_{12, 21} & 0 & 0
\end{array} \right]
\end{align*}
The admittance matrix is then
	\begin{align*}
\tilde Y^l & \ := \ \left[ \begin{array}{c c c | c c c c}
    \textcolor[rgb]{1,0,0}{ \hat Y^l_{11, 22} } &  \textcolor[rgb]{1,0,0}{ \hat r^{l \sf T}_{11, 22} } 
 & 0 & Y^l_{12, 21} & 0 & 0 & 0
    \\
\textcolor[rgb]{1,0,0}{ \hat r^{l}_{11, 22} } & 
\textcolor[rgb]{1,0,0}{ d^l_{11,22} } +  \textcolor[rgb]{0,0,1}{ d^{l-1}_{11,22} } &
\textcolor[rgb]{0,0,1}{ \hat r^{l-1}_{11, 22} } & r^l_{12, 21} &  r^{l-1}_{12, 21} & 0 & 0
\\
    0 & \textcolor[rgb]{0,0,1}{ \hat r^{l-1 \sf T}_{11,22} } &  \textcolor[rgb]{0,0,1}{ \hat Y^{l-1}_{11, 22} } &  
    0 & Y^{l-1}_{12, 21} & 0 & 0
    \\ \hline
    & & & Y^l_{22, 11} & 0 & Y^l_{22, 12} & 0  \\
    & & & 0 & Y^{l-1}_{22, 11}  & 0 & Y^{l-1}_{22, 12}  \\
    & & & & & Y^l_{22, 22} & 0 \\
    & & & & & 0 & Y^{l-1}_{22, 22}
    \end{array} \right]
\ =: \ \begin{bmatrix} Y_{11} & Y_{12} \\ Y_{12}^{\sf T} & Y_{22}  \end{bmatrix}
	\end{align*} 
where the submatrix $\tilde Y^l_{11}$ of $\tilde Y^l$ has the same structure as that in \eqref{eq:barY.case3}.

\subsection{Step 5: Putting back internal measured nodes}
\label{sec:algorithm; subsec:ReverseStep1}

Recall that, at the end of Step 1, we have removed all internal measured nodes from the 
Kron-reduced network, resulting in the admittance matrix $\bar Y'$ computed from \eqref{eq:solution.3}.
Step 4 therefore produces the admittance matrix $Y'$ from $\bar Y'$, corresponding to the original 
network with all internal measured nodes and all lines incident on them removed.  We now reverse
\eqref{eq:solution.3} to compute the admittance matrix $Y$ of the original network from $Y'$.

Suppose $Y'$ produced in Step 4 takes the form
\begin{align*}
	Y' &   \ =: \ \left[ \begin{array}{c | c c}
      \textcolor[rgb]{1,0,0}{ Y_{11, 22}' } & Y_{12, 21}' & 0 
    \\ \hline
      & Y_{22, 11}' & Y_{22, 12}'
    \\
      & & Y_{22, 22}'
    \end{array} \right] 
\end{align*}
where \textcolor[rgb]{1,0,0}{ $Y_{11, 22}'$ } corresponds to the 
boundary measured nodes and $Y_{22, 11}'$ and $Y_{22,22}'$ correspond to the hidden nodes
that connect them in $G(Y')$.  Then, given $Y_{11,11}, Y_{11,12}$ from Step 1,
the admittance matrix $Y$ of the overall network is
\begin{align*}
	Y &   \ =: \ \left[ \begin{array}{c c | c c}
    Y_{11,11} & Y_{11, 12} & 0 & 0     \\ 
    &  \textcolor[rgb]{1,0,0}{ Y_{11, 22} } & Y_{12, 21}' & 0 
    \\ \hline
    &  & Y_{22, 11}' & Y_{22, 12}'
    \\
    &  & & Y_{22, 22}'
    \end{array} \right] 
\end{align*}
where \textcolor[rgb]{1,0,0}{ $Y_{11, 22}$} is obtained from \textcolor[rgb]{1,0,0}{ $Y_{11, 22}'$ } 
by modifying the diagonal entries of \textcolor[rgb]{1,0,0}{ $Y_{11, 22}'$ } so that $Y$ has zero row 
and column-block sums, i.e., reversing \eqref{eq:solution.3}:
\begin{align*}
\textcolor[rgb]{1,0,0}{ Y_{11, 22} } & \ := \
\textcolor[rgb]{1,0,0}{ Y_{11, 22}' } \, - \
\diag\left( \left(\textbf 1 \otimes \textbf I_3 \right)^{\sf T} Y_{11,12} \right)
\end{align*}

\subsection{Summary: Algorithm 1}
\label{sec:algorithm; subsec:summary}

By ``identifying node $j$'', we mean identifying $y_{jk}$ in the admittance matrix $Y$
for all lines incident on node $j$.  By ``removing node $j$'', we mean removing the $j$th
row and column-blocks from $Y$ corresponding to node $j$.
\paragraph{Algorithm 1: overall identification algorithm.}
\bee
\item[] Step 1: Identify all internal measured nodes from $\bar Y$.  Obtain Kron reduction
	$\bar Y'$ by removing all internal measured nodes from $\bar Y$:
	\begin{align*}
	Y_{11, 11} & \ = \ \bar Y_{11},  \qquad  Y_{11, 12} = \bar Y_{12}
	\qquad
	\bar Y' \ = \ \bar Y_{22} - \text{diag}\left( \left(\textbf 1 \otimes \textbf I_3 \right)^{\sf T} \bar Y_{22} \right)
\end{align*}

\item[] Step 2: Decompose $G(\bar Y')$ into maximal cliques $\{C_j, j=1,\dots, m\}$ in isolation:
	\begin{align*}
	\bar Y' & \ \longrightarrow \ \{ \bar Y^j, \, j=1,\dots, m\}
	\end{align*}

\item[] Step 3: Identify each maximal clique $G(Y^j)$ in isolation from $G(\bar Y^j) = C_j$ 
	using Algorithm 2 in Section \ref{sec:IDMC; subec:reverseKR}:
	\begin{align*}
	\bar Y^j & \ \longrightarrow \ Y^j, 	\qquad j = 1, \dots, m
	\end{align*}
	
\item[] Step 4: Combine all maximal cliques $\{G(Y^j), j=1,\dots, m\}$ in isolation into the original 
	graph $G(Y')$ without internal measured nodes (reversing Step 2):
	\begin{align*}
	\{ Y^j, \, j=1,\dots, m\} & \ \longrightarrow \ Y' \ \ := \ \ 
	\left[ \begin{array}{c | c c}
      \textcolor[rgb]{1,0,0}{ Y_{11, 22}' } & Y_{12, 21}' & 0 
    \\ \hline
      & Y_{22, 11}' & Y_{22, 12}'
    \\
      & & Y_{22, 22}'
    \end{array} \right] 
	\end{align*}

\item[] Step 5: Obtain the admittance matrix $Y$ of the original graph by putting back all 
	internal measure nodes (reversing Step 1):
	\begin{align*}
	Y &   \ =: \ \left[ \begin{array}{c c | c c}
    Y_{11,11} & Y_{11, 12} & 0 & 0     \\ 
    &  \textcolor[rgb]{1,0,0}{ Y_{11, 22} } & Y_{12, 21}' & 0 
    \\ \hline
    &  & Y_{22, 11}' & Y_{22, 12}'
    \\
    &  & & Y_{22, 22}'
    \end{array} \right] & & \text{ where } &
    \textcolor[rgb]{1,0,0}{ Y_{11, 22} } & \ := \
\textcolor[rgb]{1,0,0}{ Y_{11, 22}' } \, - \
\diag\left( \left(\textbf 1 \otimes \textbf I_3 \right)^{\sf T} Y_{11,12} \right)
\end{align*}
\eee

\section{Identification of maximal clique in isolation}
\label{sec:IterativeKR}

In this section we devise a method to identify a single maximal clique in isolation from its
Kron-reduced admittance matrix computed by Step 2 in 
Section \ref{sec:algorithm; subsec:DecomposeMaxCliques}.
It is a critical step (Step 3 in Section \ref{sec:algorithm; subsec:IdentifyMaxClique}) in the 
overall network identification algorithm.

In Section \ref{sec:IDMC; subsec:BasicIdea} we define the network we
try to identify in this section and describe an iterative procedure that computes its 
Kron reduction
by reducing one hidden node in each step, starting from a tree of measured and hidden nodes and 
terminating in a single maximal clique of only measured nodes.  The key is to track a certain 
permuted admittance matrix $\hat A^l$ of the graph in each step as well as its principal submatrix 
$C^l$ that represents the clique subgraph.  
In Section \ref{sec:IDMC; subec:forwardKR} we derive an invariant structure of this sequence of
matrices $(\hat A^l, C^l)$ as the original graph is being iteratively Kron reduced.  
In Section \ref{sec:IDMC; subec:reverseKR} we show how to reverse each iteration using this invariant 
structure and present the identification algorithm for a single maximal clique in isolation.

\subsection{Basic idea: reversible one-step Kron reduction}
\label{sec:IDMC; subsec:BasicIdea}

\subsubsection{Single maximal clique}
\label{sec:IterativeKR; subsec:1maxclique}

Consider a connected three-phase radial network consisting of $M$ degree-1 (boundary) measured 
nodes connected by $H$ non-leaf (boundary or internal) hidden nodes each with 
degree at least 2 (or at least 3 under Assumption \ref{Assumption:HiddenNodes}). Let $H_b\leq H$
denote the number of boundary hidden nodes.  There are no internal 
measured nodes in this network.  Let $Y$ denote its admittance matrix and $\bar Y$ its Kron reduction
where all the hidden nodes have been Kron reduced.  
Since the Kron reduced network $G(\bar Y)$ is a clique where every measured node is adjacent to every
other measured node, $\bar Y$ is an admittance matrix whose $3\times 3$ entry blocks are all nonzero.
We will often use ``maximal clique'' to refer both to the tree underlying $Y$ as well as the clique underlying 
its Kron reduction $\bar Y$; the meaning should be clear from the context.

Since there are no internal measured nodes the admittance matrix $Y$ in \eqref{eq:Y.2} of the maximal clique 
reduces to the following form
\begin{subequations}
\begin{align}
\label{sec:IterativeKR; subsec:reverse; eq:Y.3a}
    Y & \ =: \ \left[ \begin{array}{c | c}
    Y_{11} & Y_{12} \\ \hline Y_{21} & Y_{22}
    \end{array} \right]   
    \ =: \ \left[ \begin{array}{c | c c}
     Y_{11, 22} & Y_{12, 21} & 0 
    \\ \hline
     & Y_{22, 11} & Y_{22, 12} 
    \\
     & & Y_{22, 22} 
    \end{array} \right]
\end{align}
The given Kron reduction is $\bar Y =: Y/Y_{22}$.
The tree underlying $Y$ in \eqref{sec:IterativeKR; subsec:reverse; eq:Y.3a} is illustrated in 
Figure \ref{fig:IterativeKronReduction1} below.  
Each boundary measured node $i$ has nonzero admittance 
submatrix $y_{ih(i)}\in\mathbb C^{3\times 3}$ for exactly one hidden node $h(i)$ in the tree. Otherwise if $i$ is adjacent
to two hidden nodes, there is a loop in $G$.  Since every boundary measured node is adjacent to a hidden node, 
no boundary measured nodes can be adjacent to each other in $G$; otherwise there is a loop in $G$.  
Therefore $Y_{11,22}$ and $Y_{12,21}$ are of the form
\begin{align}
Y_{11,22} & \ = \ \diag\begin{bmatrix} y_{1h(1)} \\ \vdots \\ y_{Mh(M)}    \end{bmatrix}, 	&
Y_{12, 21} & \ = \ \begin{bmatrix} - e_{h(1)}^{\sf T}\otimes y_{1h(1)} \\ \vdots \\  - e_{h(M)}^{\sf T}\otimes y_{Mh(M)} \end{bmatrix}
\label{sec:IterativeKR; subsec:reverse; eq:Y.3b}
\end{align}
where $e_i\in\{0, 1\}^{H_b}$ is the unit vector with a single 1 in the $i$th entry and 0 elsewhere, and $y_{ij}\in\mathbb C^{3\times 3}$ is the three-phase series admittance of line $(i,j)$.  
\label{sec:IterativeKR; subsec:reverse; eq:Y.3}
\end{subequations}
Here $H_b\leq H$ is the number of boundary hidden nodes.

\begin{example}
\label{sec:IterativeKR; subsec:reverse; eg:Y}
Without loss of generality let measured nodes $1, \dots, m_1$ in $G(A^0)$ be adjacent 
to a common hidden nodes $M+1 = h(1) = \dots = h(m_1)$, measured nodes $m_1+1, \dots, m_2$ 
to hidden node $M+2 = h(m_1+1) = \dots = h(m_2)$, 
$...$, measured nodes $m_{H_b-1}+1, \dots, m_{H_b} := M$ to hidden node 
$M+H_b = h(m_{H_b-1}+1) = \dots = h(M)$.
Collect all $3\times 3$ series admittance matrices adjacent to the same hidden nodes 
in submatrices 
\begin{align*}
y_1 & \ := \ (y_{1h(1)}, \dots, y_{m_1 h(m_1)})	\\
y_2 & \ := \ (y_{(m_1+1)h(m_1+1)}, \dots, y_{m_2 h(m_2)})	\\
& \ \  \ \vdots	\\
y_{H_b} & \ := \ (y_{(m_{H_b-1}+1)h(m_{H_b-1}+1)}, \dots, y_{m_{H_b} h(m_{H_b})} = y_{Mh(M)})
\end{align*}
i.e., $y_1\in\mathbb C^{3m_1\times 3}$ denotes the admittance matrix of lines connecting 
the first $m_1$ measured nodes
to the first hidden node $M+1$, and for $j > 1$, $y_j \in\mathbb C^{3(m_j-m_{j-1})\times 3}$ 
denotes those connecting 
$m_j-m_{j-1}$ measured nodes to the hidden node $M+j$.  
Then $Y$ takes the form \eqref{sec:IterativeKR; subsec:reverse; eq:Y.3a} with
\begin{align*}
Y_{11,22} & \ = \ \begin{bmatrix} \diag\left( y_1\right) & 0 & \hdots & 0 \\ 0 & \diag\left( y_2\right) & \hdots & 0 \\ 
	\vdots & \vdots & \ddots & \vdots \\  0 & 0 & \hdots & \diag\left( y_{H_b} \right) \end{bmatrix}, & 
Y_{12,21} & \ = \ \begin{bmatrix} - y_1 & 0 & \hdots & 0 \\ 0 & - y_2 & \hdots & 0 \\ \vdots & \vdots & \ddots & \vdots \\
	0 & 0 & \hdots & - y_{H_b} 	\end{bmatrix}
\end{align*}

Recall the definition of $Z_{22}$ in \eqref{sec:algorithm; eq:Step1.1a}:
\begin{align}
Z_{22} & \ := \ \left( Y_{22}\right)^{-1} \ =: \ 
\begin{bmatrix}
X_{22, 11} & X_{22, 12} \\ X_{22, 21} & X_{22, 22}
\end{bmatrix}
\label{sec:IterativeKR; subsec:reverse; eg:Y; eq:Z_22}
\end{align}
Lemma \ref{lemma:InvertibleSubmatrix} implies that $Y_{22}$ and $Z_{22}$ are symmetric matrices.
Substituting into \eqref{sec:IterativeKR; subsec:reverse; eq:Y.3a}, we can write 
$Y/Y_{22} = Y_{11} - Y_{12} Z_{22} Y_{21}$ in terms of $X_{22,11}$: 
\begin{align*}
\bar Y & \ := \ Y/Y_{22} \ = \ Y_{11, 22}  \, - \, Y_{12, 21} X_{22, 11} Y_{21,12} 
\end{align*}
with $Y_{21,12} = \left( Y_{12, 21} \right)^{\sf T}$.
From \eqref{sec:IterativeKR; subsec:reverse; eq:Y.3a}\eqref{sec:IterativeKR; subsec:reverse; eg:Y; eq:Z_22}
and \eqref{eq:A-1.b} in Section \ref{sec:Assumption3; subsec:SchurComplement}, 
$X_{22,11} = \left( Y_{22}/Y_{22,22} \right)^{-1}$ which exists under Assumption \ref{Assumption:y_jk} 
according to Lemma \ref{lemma:InvertibleSubmatrix}.  
Partition the $3H_b \times 3H_b$ matrix $X_{22,11}$ into $H_b^2$ blocks each of $3\times 3$:
\begin{align*}
X_{22, 11} & \ =: \ \begin{bmatrix}
    \beta_{11} & \beta_{12} & \cdots & \beta_{1H_b} \\  \beta_{21} & \beta_{22} & \cdots & \beta_{2H_b} \\     
    \vdots & \vdots & \vdots & \vdots    \\  \beta_{H_b 1} & \beta_{H_b 2} & \cdots & \beta_{H_b H_b}
    \end{bmatrix}
\end{align*}
where $\beta_{kl} = \beta_{lk}^{\sf T} \in\mathbb C^{3\times 3}$.  
Hence
\begin{align*}
Y_{12, 21} X_{22, 11} Y_{21,12} & \ = \ \begin{bmatrix}
     y_1 \beta_{11}  y_1^{\sf T} &  y_1 \beta_{12}  y_2^{\sf T} & \cdots &  y_1 \beta_{1H_b}  y_{H_b}^{\sf T} \\     
     y_2 \beta_{21}  y_1^{\sf T} &  y_2 \beta_{22}  y_2^{\sf T} & \cdots &  y_2 \beta_{1H_b}  y_{H_b}^{\sf T} \\     
    \vdots & \vdots & \vdots & \vdots    \\  
 y_{H_b} \beta_{H_b 1}  y_1^{\sf T} &  y_{H_b} \beta_{H_b 2}  y_2^{\sf T} & \cdots &  y_{H_b} \beta_{H_b H_b}  y_{H_b}^{\sf T}
    \end{bmatrix}
\end{align*}
Therefore the Kron-reduced admittance matrix 
$\bar Y := Y_{11, 22}  \, - \, Y_{12, 21} X_{22, 11} Y_{21,12}$
becomes ($\bar Y$ is symmetric)
\begin{align*}
\bar Y & \ = \ \begin{bmatrix} \diag\left( y_1\right) & 0 & \hdots & 0 \\ & \diag\left( y_2\right) & \hdots & 0 \\ 
	 &  & \ddots & \vdots \\   &  &  & \diag\left( y_{H_b} \right) \end{bmatrix} \ - \
	\begin{bmatrix}
     y_1 \beta_{11}  y_1^{\sf T} &  y_1 \beta_{12}  y_2^{\sf T} & \cdots &  y_1 \beta_{1H_b}  y_{H_b}^{\sf T} \\     
     &  y_2 \beta_{22}  y_2^{\sf T} & \cdots &  y_2 \beta_{1H_b}  y_{H_b}^{\sf T} \\     
     &  & \ddots & \vdots    \\  &  &  &  y_{H_b} \beta_{H_b H_b}  y_{H_b}^{\sf T}
    \end{bmatrix}
\end{align*}
\qed
\end{example}

In this section we design an algorithm that uses the structure in \eqref{sec:IterativeKR; subsec:reverse; eq:Y.3}
to recover $Y$ for a single maximal clique by iteratively reversing Kron reduction, starting from $\bar Y$.
We start by decomposing the forward Kron reduction into a sequence of iterations that maintain an invariant 
structure.

\subsubsection{Iterative Kron reduction}
\label{sec:IterativeKR; subsec:1maxclique; subsubsec:IKR}

It is more convenient to describe iterative Kron reduction in terms of an arbitrary $3n\times 3n$ complex matrix $A^0$
on a graph $ G^0 := ( N^0,  E^0)$ where its $3\times 3$ $(i,j)$th blocks $A^0[i,j]$ are given by:
\begin{align*}
A^0[i,j] & \ = \ \left\{ \begin{array}{lcl}
    - y_{jk} & & (i,j)\in  E^0 \\ 	\sum_{k: (i,k)\in E^0} y_{ik} & & i = j \\	 0 & & \text{otherwise}
    \end{array} \right.
\end{align*}
We refer to $A^0$ as the admittance matrix of the graph $G^0$, or equivalently $G^0 = G(A^0)$.  
Suppose the graph and its admittance matrix $(G^0, A^0)$ satisfy Assumption \ref{Assumption:y_jk}.
For example, if $A^0 := Y$ as in our case then $n=M+H$.

Let $A^0 =: \begin{bmatrix} A_{11} & A_{12} \\ A_{12}^{\sf T} & A_{22} \end{bmatrix}$
with a $3k\times 3k$ nonsingular submatrix $A_{22}$, $1\leq k < n$ (in our case, $k=H$).
To simplify exposition, we will refer to nodes in $A_{22}$ to be Kron reduced as ``hidden nodes'' 
and nodes in $A_{11}$ as ``measured nodes''.  
We can compute the Schur complement $A^0/A_{22}$ of $A_{22}$ of the admittance matrix $A^0$ by eliminating  
hidden nodes on the graph $ G^0$ one by one through Kron reduction.  
Following \cite{DorflerBullo2013}, we define
\begin{align}
A^1 & \ := \ A^0/A^0[n,n], & \cdots & &
A^k & \ := \ A^{k-1}/A^{k-1}[n-k+1, n-k+1] \ = \ A^0/A_{22}
\label{subsec:IterativeKR; eq:1}
\end{align}
i.e., $A^{l+1}$ is the admittance matrix for the graph after the last node
in $A^l$ has been Kron reduced, and hence $A^0/A_{22} = A^k$.  Conversely
a sequence of matrices $A^0$, $A^1, \dots, A^k$ computed according to \eqref{subsec:IterativeKR; eq:1}
defines a sequence of graphs $G^0, G^1, \dots, G^k$ with $G^l = (N^l, E^l) := G(A^l)$ defined by
($N^0 := \{1, \dots, n\}$)
\begin{align*}
N^l & \ := \ N^0 \setminus \left\{ n, n-1, \dots, n-l+1 \right\}, &
E^l & \ := \ \left\{(i,j) : A^l[i,j] \neq 0 \right\}, & l & = 0, 1, \dots, k
\end{align*}
We refer to $G^l$ as the graph underlying the matrix $A^l$.  Sometimes we refer to the graph by $A^l$
instead of $G^l := G(A^l)$ when the meaning clear should be clear from the context.

A more explicit form of the iteration \eqref{subsec:IterativeKR; eq:1} is given in terms of the computation
of each $3\times 3$ $(i,j)$th block $A^{l+1}[i,j]$: for $l=0, \dots, k-1$,
\begin{align}
A^{l+1}[i,j] & \ = \ A^l[i,j] \, - \, A^l[i,n-l] \left( A^l[n-l, n-l] \right)^{-1} A^l[j,n-l], & i, j & = 1, \dots, n-l-1
\label{lemma:Kron-reduced-y; eq:3}
\end{align}
Starting from $A^0$, \eqref{lemma:Kron-reduced-y; eq:3} iteratively computes the Kron-reduced admittance 
matrix $A/^0A_{22} = A^k$.   
The iterative computation is useful for proving properties that are preserved under Kron reduction, as we will see.
The following properties follow directly from \eqref{lemma:Kron-reduced-y; eq:3}.  
\begin{lemma}
\label{lemma:IterativeKronReduction.1}
Suppose the $3\times 3$ principal submatrices $A^l[n-l, n-l]$ in 
\eqref{subsec:IterativeKR; eq:1} are nonsingular for $l=0, 1, \dots, k-1$.  Then
\bee
\item $A^{l+1}[i,j] = A^l[i,j]$ unless both nodes $i$ and $j$ are adjacent to node $n-l$ in $A^l$, i.e., 
	$A^l[i,n-l]\neq 0\in\mathbb C^{3\times 3}$ and 	$A^l[j,n-l]\neq 0\in\mathbb C^{3\times 3}$.
\item $A^{l+1}[i,j]=0$, i.e., nodes $i$ and $j$ are not adjacent in $A^{l+1}$, if
	\bi 
	\item $i$ and $j$ are not adjacent in $A^l$; and
	\item $i$ and $j$ are not both adjacent to $n-l$.
	\ei
\eee
\end{lemma}
The converses of the lemma hold under certain conditions (see Theorem \ref{lemma:Kron-reduced-y}). 
\begin{example}[Iterative Kron reduction]
\label{eg:IterativeKronReduction.1}
Computing the Kron reduction $A^0/A_{22}$ by iteratively computing $A^l$ using \eqref{lemma:Kron-reduced-y; eq:3} and the
underlying graphs $G^l$ are illustrated in Figure \ref{fig:IterativeKronReduction1}.
	\begin{figure}[htbp]
	\centering
	\includegraphics[width=0.95\textwidth] {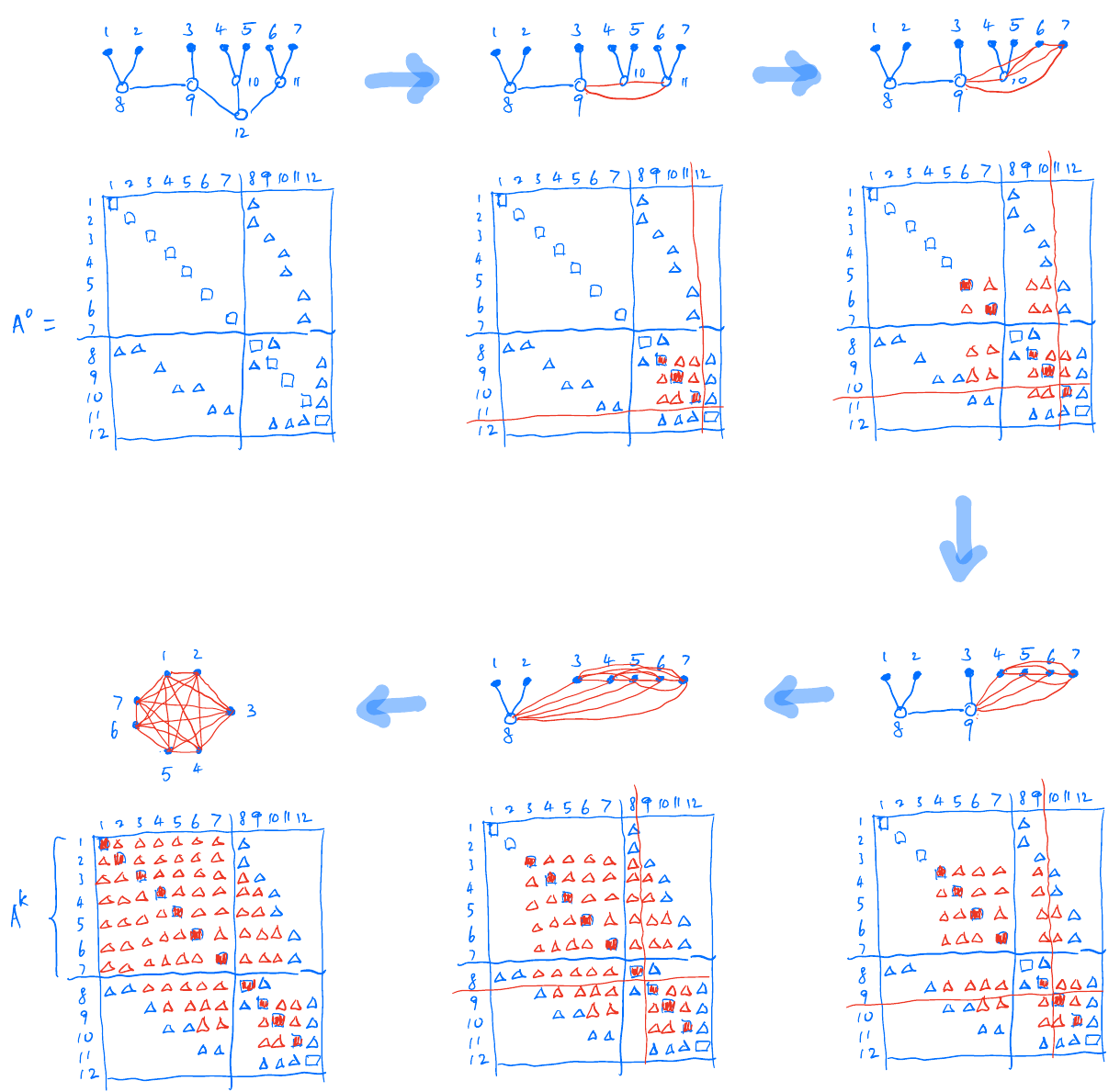}
	\caption{Example \ref{eg:IterativeKronReduction.1}: Iterative Kron reduction.  The graph $G_0$ underlying $A^0$ is a tree with shaded and unshaded nodes.
	The unshaded nodes correspond to $A_{22}$ and are to be Kron reduced.  In each iteration, only entries (corresponding
	to a maximal clique in $A^{l+1}$) marked
	by red triangles or red squares are updated; other entries remain the same as their values in $A^0$.}
	\label{fig:IterativeKronReduction1}
	\end{figure}
Suppose Lemma \ref{lemma:IterativeKronReduction.1} and its converse both hold.  Then each iteration of 
\eqref{lemma:Kron-reduced-y; eq:3} has the following properties:
\bee
\item It eliminates one unshaded node $n-l$ and connects all (and only) nodes that are adjacent to node 
	$n-l$ in graph $A^l$ into a maximal clique in $A^{l+1}$, e.g., nodes $9, 10, 11$ in $A^1$ or nodes 
	$6,7,9,10$ in $A^2$ in Figure \ref{fig:IterativeKronReduction1}. 
	The maximal clique forms a principal submatrix of $A^{l+1}$ in which all entries are nonzero.
\item Lemma \ref{lemma:IterativeKronReduction.1} implies that the entry $A^{l+1}[i,j]$ remains the same as $A^l[i,j]$ 
	if either node $i$ or node $j$ is not adjacent to node $n-l$ in $A^l$.
	 The entries $A^{l+1}[i,j]$ are updated only for nodes $i,j$ that are in the maximal clique in $A^{l+1}$, i.e., 
	 only for nodes $i,j$ that are either in $A^l$ or when both $i$ and $j$ are neighbors of node $n-l$; 
	 they are indicated by red 	 triangles or squares in Figure \ref{fig:IterativeKronReduction1}. 
\item For $l=2$, e.g., $A^2[i,j] = A^0[i,j]$ for $i,j=1, 2, 3, 4, 5, 8, 9$, except 
the diagonal entry $A^2[9,9]$ (in red), as Lemma \ref{lemma:IterativeKronReduction.1} indicates.
 \eee
 \qed
\end{example}

\subsubsection{Maximal clique $C^l$ of re-labeled matrix $\hat A^l$}
\label{sec:IterativeKR; subsec:1maxclique; subsubsec:MaxClique}

As Figure \ref{fig:IterativeKronReduction1} illustrates, given a tree $A^0$, the forward iterative 
Kron reduction grows an initial maximal clique
$A^0[n,n]\in\mathbb C^{3\times 3}$ consisting of a single node $n$ into a single maximal clique 
$A^k\in\mathbb C^{3(n-k)\times 3(n-k)}$ consisting of $n-k$ nodes, while eliminating $k$ hidden nodes
from $A^0$ in the process, one hidden node in each step.

The {basic idea} of our identification method is to derive an invariant
structure that is preserved under one-step Kron reduction and that is reversible.  
To this end, it is more convenient to focus, not on the sequence 
$A^l := A^{l-1}/A^l[n-l-1, n-l-1]$, but a permuted sequence $\hat A^l$ over which the invariant structure can 
be propagated and from which the sequence $A^l$ in \eqref{subsec:IterativeKR; eq:1} can be extracted
(see Section \ref{sec:IDMC; subec:forwardKR}).
We will use $C^l$ to denote the maximal clique in the permuted matrix 
$\hat A^l$ in iteration $l$. We will abuse notation to use the term ``maximal clique'' and the symbol $C^l$ 
to refer to either the principal submatrix of $\hat A^l$, or the subgraph of $G(\hat A^l)$, or the nodes in the subgraph
corresponding to the maximal clique; the meaning should be clear from the context.\footnote{The principal 
submatrix $C^l$ of $A^l$ differs from an admittance matrix 
in our model only in that its row and column sums are nonzero due to the connectivity between the 
maximal clique and other nodes in $G^l$.  If the diagonal entries of $C^l$ are normalized so that
$C^l$ has zero row and column sums (see \eqref{eq:solution.3}), the resulting matrix $(C^l)'$ is
the admittance matrix of the maximal clique \emph{in isolation} $G((C^l)')$.
}

The permuted sequence $\hat A^l$ results from re-labeling nodes in each step of the iterative Kron reduction.
Specifically, given the permuted matrix $\hat A^l$ in each iteration $l$, after taking the Schur complement 
$\hat A^l/\hat A^l[n-l, n-l]$ to reduce node $n-l$, we will re-label nodes so that the next permuted matrix
$\hat A^{l+1}$ has all (hidden and measured) nodes in the maximal clique indexed consecutively with 
the largest indices (as well as another convenient structure).
This corresponds to multiplying the matrix $\hat A^l/\hat A^l[n-l, n-l]$ on the left and on 
the right by appropriate permutation matrices and its transpose respectively to obtain $\hat A^{l+1}$.
We explain in Section \ref{sec:IDMC; subec:forwardKR} 
the invariant structure of the permuted sequence $(\hat A^l, C^l)$ and how to compute $A^l$ from $\hat A^l$.
Before we do that we first describe in the following example what re-labeling means and how to recover the 
Kron reduction of a matrix from the Kron reduction of the permuted version of the matrix.
\begin{example}[Re-labeling nodes]
\label{sec:IDMC; subsec:BasicIdea; eg:relabeling}
Consider an arbitrary matrix $A\in\mathbb C^{m\times n}$ and the permutation matrices $P_m\in\{0, 1\}^{m\times m}$,
$P_n\in\{0, 1\}^{n\times n}$ ($\textbf I_k$ denotes the identify matrix of size $k$):
\begin{align*}
A & \ := \ \begin{bmatrix} {\color{red}{C_{11}}} & A_{12} & {\color{red}{C_{13}}} \\
	A_{21} & A_{22} & A_{23} \\
	{\color{red}{C_{31} }} & A_{32} & {\color{red}{C_{33}}}
	\end{bmatrix}, & 
P_m & \ := \ \begin{bmatrix} 0 & \textbf I_{m_2} & 0 \\ \textbf I_{k_1} & 0 & 0 \\ 0 & 0 & \textbf I_{k_3} \end{bmatrix}, &
P_n & \ := \ \begin{bmatrix} 0 & \textbf I_{n_2} & 0 \\ \textbf I_{k_1} & 0 & 0 \\ 0 & 0 & \textbf I_{k_3} \end{bmatrix}
\end{align*}
where $C_{11}\in\mathbb C^{k_1\times k_1}$,  $A_{22}\in\mathbb C^{m_2\times n_2}$, 
$C_{33}\in\mathbb C^{k_3\times k_3}$ so that $k_1+m_2+k_3 = m$ and $k_1+n_2+k_3 = n$;
the other submatrices $A_{jk}, C_{jk}$ and the zero matrices are of appropriate sizes.  
Then 
\begin{align*}
P_m A& \ := \ \begin{bmatrix}	A_{21} & A_{22} & A_{23} \\  {\color{red}{C_{11}}} & A_{12} & {\color{red}{C_{13}}} \\
	{\color{red}{C_{31} }} & A_{32} & {\color{red}{C_{33}}}
	\end{bmatrix}, & 
\hat A \ := \ P_m A P_n^{\sf T} & \ := \ \begin{bmatrix}	
	A_{22} & A_{21} & A_{23} \\  A_{12} & {\color{red}{C_{11}}} & {\color{red}{C_{13}}} \\
	A_{32} & {\color{red}{C_{31} }} & {\color{red}{C_{33}}}
	\end{bmatrix}
\end{align*}
i.e., the permuted matrix $\hat A$ collects the submatrices $(C_{11}, C_{13}, C_{31}, C_{33})$ that
were originally spread across $A$ into a contiguous block at the lower right corner.

We may also encounter matrices whose maximal cliques are scattered into more than four pieces, as in the following
example:
\begin{align*}
A & := \begin{bmatrix} {\color{red}{C_{11}}} & A_{12} & {\color{red}{C_{13}}} & A_{14} &  {\color{red}{C_{15}}} \\
	A_{21} & A_{22} & A_{23} & A_{24} & A_{25} \\ {\color{red}{C_{31} }} & A_{32} & {\color{red}{C_{33}}} &A_{34} &  {\color{red}{C_{35}}} \\
	A_{41} & A_{42} & A_{43} & A_{44} & A_{45} \\ {\color{red}{C_{51} }} & A_{52} & {\color{red}{C_{53}}} & A_{54} &  {\color{red}{C_{55}}}
	\end{bmatrix}, & 
P_m & := \begin{bmatrix} 0 & \textbf I_{m_2} & 0 & 0 & 0 \\ 0 & 0 & 0 & \textbf I_{m_4} & 0 \\
	\textbf I_{k_1} & 0 & 0 & 0 & 0 \\ 0 & 0 & \textbf I_{k_3} & 0 & 0 \\ 0 & 0 & 0 & 0 & \textbf I_{k_3} \end{bmatrix}, &
P_n & := \begin{bmatrix} 0 & \textbf I_{n_2} & 0 & 0 & 0 \\ 0 & 0 & 0 & \textbf I_{n_4} & 0 \\
	\textbf I_{k_1} & 0 & 0 & 0 & 0 \\ 0 & 0 & \textbf I_{k_3} & 0 & 0 \\ 0 & 0 & 0 & 0 & \textbf I_{k_3} \end{bmatrix}
\end{align*}
where $A_{44}\in\mathbb C^{m_4\times n_4}$, 
$C_{55}\in\mathbb C^{k_5\times k_5}$ so that $k_1+m_2+k_3 + m_4 + k_5 = m$ and $k_1+n_2+k_3 + n_4 + k_5 = n$, giving
\begin{align*}
P_m A & := \begin{bmatrix} 
	A_{21} & A_{22} & A_{23} & A_{24} & A_{25} \\ 	A_{41} & A_{42} & A_{43} & A_{44} & A_{45} \\ 
	{\color{red}{C_{11}}} & A_{12} & {\color{red}{C_{13}}} & A_{14} & {\color{red}{C_{15}}} \\
	{\color{red}{C_{31} }} & A_{32} & {\color{red}{C_{33}}} & A_{34} & {\color{red}{C_{35}}} \\
	{\color{red}{C_{51} }} & A_{52} & {\color{red}{C_{53}}} & A_{54} & {\color{red}{C_{55}}}
	\end{bmatrix}, & 
\hat A \ := \ P_m A P_n^{\sf T} & := \begin{bmatrix} 
	A_{22} & A_{24} & A_{21} & A_{23} & A_{25} \\ 	A_{42} & A_{44} & A_{41} & A_{43} & A_{45} \\ 
	A_{12} & A_{14} &  {\color{red}{C_{11}}} & {\color{red}{C_{13}}} & {\color{red}{C_{15}}} \\
	A_{32} & A_{34} & {\color{red}{C_{31} }} & {\color{red}{C_{33}}} & {\color{red}{C_{35}}} \\
	A_{52} & A_{54} & {\color{red}{C_{51} }} & {\color{red}{C_{53}}} & {\color{red}{C_{55}}}
	\end{bmatrix}
\end{align*}
\qed
\end{example}

Starting from $A^0$, we may re-label nodes after each iteration in the forward iterative Kron reduction 
to obtain the permuted matrix $\hat A^l$ (see Section \ref{sec:IDMC; subec:forwardKR}) before taking the 
Schur complement in 
the next iteration.   The Kron reduction $A^0/A_{22}$ is generally not equal to the matrix $\hat A^k$ at the 
end of the forward iterations.  The Kron reduction $A^0/A_{22}$ can, however, be recovered from $\hat A^k$
since the re-labeling in each iteration $l$ does not re-label node $n-l$ that will be Kron reduced in that iteration.  
Consider an arbitrary square matrix $A\in\mathbb C^{3(m_1+m_2)\times 3(m_1+m_2)}$ partitioned as
$A =: \begin{bmatrix} A_{11} & A_{12} \\ A_{12}^{\sf T} & A_{22} \end{bmatrix}$ where $A_{11}$ is
$m_1\times m_1$ and $A_{22}$ is $m_2\times m_2$.  The permutation matrix that re-labels nodes 
in $A_{11}$ takes the form
\begin{subequations}
\begin{align}
P & \ = \ \begin{bmatrix} P_{m_1} & 0 \\ 0 & \textbf I_{m_2} \end{bmatrix}
\label{sec:IDMC; subsec:BasicIdea; eq:P.1a}
\end{align}
and the permuted matrix is
\begin{align*}
\hat A \ := \ P A P^{\sf T} & \ = \ \begin{bmatrix} P_{m_1} A_{11} P_{m_1}^{\sf T} & P_{m_1} A_{12} \\
	A_{21} P_{m_1}^{\sf T} & A_{22} \end{bmatrix}
\end{align*}
The Kron reduction of the permuted matrix is
\begin{align*}
\hat A/A_{22} \ := \ P A P^{\sf T}/A_{22} & \ = \ P_{m_1} \left( A_{11} - A_{12} A_{22}^{-1} A_{21} \right) P_{m_1}^{\sf T} 
\end{align*}
Since the square of any permutation matrix is an identity matrix, the Kron reduction $A/A_{22}$ of the 
original matrix can be recovered as
\begin{align}
A/A_{22} & \ = \ P_{m_1} \left( \hat A/A_{22} \right) P_{m_1}^{\sf T}
\ = \ P_{m_1} \left( PAP^{\sf T}/A_{22} \right) P_{m_1}^{\sf T}
\end{align}
where the permutation matrix $P_{m_1}$ is given in \eqref{sec:IDMC; subsec:BasicIdea; eq:P.1a}.
\label{sec:IDMC; subsec:BasicIdea; eq:P.1}
\end{subequations}

\subsection{Forward Kron reduction: growing $C^l$}
\label{sec:IDMC; subec:forwardKR}

We now design an alternative iterative Kron reduction that is equivalent to the computation in \eqref{subsec:IterativeKR; eq:1}.
The alternative procedure grows the maximal clique from $C^0 := A[n,n]$ corresponding to the single node $n$ to 
$C^k := A/A_{22}$ corresponding to the Kron reduced network after removing $k$ hidden nodes.
It has the advantage that each step is easy to reverse, as 
we will explain in Section \ref{sec:IDMC; subec:reverseKR}.

Consider the admittance matrix $A^0$ of a single maximal clique of the form in \eqref{sec:IterativeKR; subsec:reverse; eq:Y.3}.  
Then initially the maximal clique $C^0 := A^0[n,n]\in\mathbb C^{3\times 3}$ is in the lower-right corner of $A^0$.  As we take Schur complements, 
the components of $C^l$
may be spread across $A^l$ (see Figure \ref{fig:IterativeKronReduction1}).  To facilitate reversing each iteration of
Kron reduction we will work with a sequence of permuted matrices $\hat A^l$ over which a convenient structure 
can be propagated and from which the sequence $A^l := A^{l-1}/A^l[n-l-1, n-l-1]$ in \eqref{subsec:IterativeKR; eq:1} can
be extracted.  The procedure is summarized in Figure \ref{fig:ReversibleKR} and given next.
	\begin{figure}[htbp]
	\centering
	\includegraphics[width=0.75\textwidth] {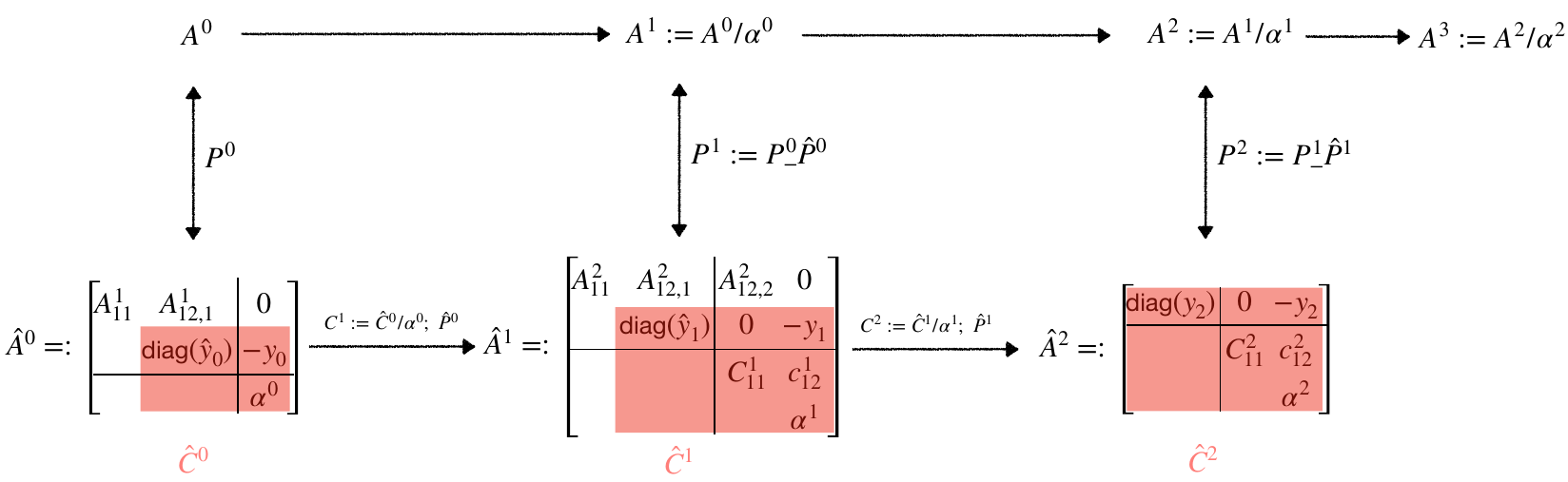}
	\caption{Reversible iterative forward Kron reduction with three hidden nodes.
	The focus of iterative Kron reduction and its reverse process will be on the
	sequence of permuted matrices and their maximal clqiues 
	$(\hat A^0, C^l), \dots, (\hat A^k, C^k)$.}
	\label{fig:ReversibleKR}
	\end{figure}

\paragraph{Initialization.}
Let the set of neighbors of the hidden node $n$ in the graph $G(A^0)$ be
\begin{align*}
N_{n}  & \ := \ \left\{j\neq n :  A^0[j, n]\neq 0\in\mathbb C^{3\times 3} \right\}
\end{align*}
This set may contain both measured and hidden nodes (e.g., if $n$ is a boundary hidden node).  
Let $n_0 := \left| N_{n} \right|$ be the number of these neighbors.
Order them in a way that maintains their relative order in $A^0$:
\begin{align*}
\underbrace{ i_1 < \dots < i_{m_0} }_{\text{measured nodes}} \ < \ \underbrace{ i_{m_0+1} < \dots < i_{n_0} }_{\text{hidden nodes}}
\end{align*}
Let $y_0\in\mathbb C^{3n_l \times 3}$ denote the admittance submatrices of lines connecting node $n$, to be Kron reduced in 
iteration $l=0$, to these neighbors: 
\begin{align*}
y_0[j,1] & \ := \  A^0[i_j,n], 	\qquad  j = 1, \dots, n_0
\end{align*}
Re-label nodes in $A^0$ except node $n$ so that $y_0$ immediately precedes $C^0$ in the lower-right corner
of the permuted matrix $\hat A^0$.  Let the permutation matrix be denoted by $P^0$, given below
(see Section \ref{sec:IterativeKR; subsec:1maxclique; subsubsec:MaxClique}).
Then $\hat A^0$ and its maximal clique $C^0$ take the
form (the structure of $\hat A^0$ will be proved in Theorem \ref{sec:IDMC; subec:forwardKR; thm:1stepKR} below):
\begin{subequations}
	\begin{align}
	\hat A^0 & \ := \ P^0 A^0 \left(\hat P^0\right)^{\sf T} \ =: \ \left[ \begin{array} {c c | c}
	A^{1}_{11} & A^{1}_{12, 1} & 0 \\ 
	& \cellcolor[rgb]{1,0.5,0.3}{ \diag(\hat y_0)} & \cellcolor[rgb]{1,0.5,0.3}{- y_0} \\ \hline
	&  \cellcolor[rgb]{1,0.5,0.3}{} &   \cellcolor[rgb]{1,0.5,0.3}{\alpha^0} \\  \end{array}  \right] \ =: \ 
	\left[ \begin{array} {c | c c} A^{1}_{11} & A^{1}_{12, 1} & 0 \\ \hline
	& \multicolumn{2}{c}{\multirow{2}{*}{{\color{red}{$\hat C^0$}}}} \\ 
	& \multicolumn{2}{c}{}
	\end{array} \right]
\label{sec:IDMC; subec:forwardKR; eq:InvStructure.0c}
\\
C^0 & \ := \ \alpha^0 \ := \ A^0[n,n], 	\qquad \ \ 
	{\color{red}{\hat C^0}} \ := \ \left[ \begin{array} {c | c}
	\diag(\hat y_0) & {- y_0} \\ \hline -y_0^{\sf T} & {\alpha^0} \end{array}  \right]
\label{sec:IDMC; subec:forwardKR; eq:InvStructure.0a} 
\\
	P^0 & \ = \ \begin{bmatrix} P^0_- & 0 \\ 0 & \textbf I_3 \end{bmatrix}
\label{sec:IDMC; subec:forwardKR; eq:InvStructure.0b}
\end{align}
where $P^0_-\in\{0, 1\}^{3(n-1)\times 3(n-1)}$ and (since $\hat A^0$ has zero row-block sums)
\begin{align}
\hat y_0 & \ := \ y_0 \, - \, \left(A^{1}_{12,1}\right)^{\sf T} \left( \textbf 1\otimes \textbf I_3 \right)
\label{sec:IDMC; subec:forwardKR; eq:InvStructure.0d}
\end{align}
From \eqref{sec:IDMC; subec:forwardKR; eq:InvStructure.0c} we have $A^0 = P^0 \hat A^0 (P^0)^{\sf T}$
since the square of a permutation matrix is an identity matrix.
\label{sec:IDMC; subec:forwardKR; eq:InvStructure.0}
\end{subequations}

Given $A^0$ the initialization step thus produces $(\hat A^0, C^0, P^0)$ where $C^0$ is the maximal clique of 
$\hat A^0$ and $A^0$ can be obtained from $\hat A^0$ through permutation matrix $P^0$.

\paragraph{Iteration.}  For $l=0, \dots, k-1$, given $(\hat A^l, C^l, P^l)$ that takes the form:
\begin{subequations}
\begin{align}
	\hat A^l & \ =: \ \left[ \begin{array} {c | c} A^l_{11} & A^l_{12} \\ \hline \left( A^l_{12} \right)^{\sf T} & C^l \end{array} \right]
	\ = \  	\left[ \begin{array} {c c | c c}
	A^{l+1}_{11} & A^{l+1}_{12, 1} & A^{l+1}_{12, 2} & 0 \\ 
	& \cellcolor[rgb]{1,0.5,0.3}{\diag(\hat y_l)} & \cellcolor[rgb]{1,0.5,0.3}{0} & \cellcolor[rgb]{1,0.5,0.3}{- y_l} \\ \hline
	&\cellcolor[rgb]{1,0.5,0.3}{} &  \cellcolor[rgb]{1,0.5,0.3}{C^l_{11}}  &  \cellcolor[rgb]{1,0.5,0.3}{c^l_{12}} \\ 
	& \cellcolor[rgb]{1,0.5,0.3}{} &  \cellcolor[rgb]{1,0.5,0.3}{} &  \cellcolor[rgb]{1,0.5,0.3}{\alpha^l } \\ 
	\end{array}  \right]
\label{sec:IDMC; subec:forwardKR; eq:InvStructure.la}
\\
C^l & \ =: \ \begin{bmatrix} C^l_{11} & c^l_{12} \\ \left( c^l_{12} \right)^{\sf T} & \alpha^l \end{bmatrix},  \qquad\ \
{\color{red}{\hat C^l}} \ := \ \begin{bmatrix} \diag(\hat y_l) & 0 & -y_l \\ & C^l_{11} & c^l_{12} \\ & & \alpha^l \end{bmatrix}
\label{sec:IDMC; subec:forwardKR; eq:InvStructure.lb}
\\
P^l & \ = \ \begin{bmatrix} P^l_- & 0 \\ 0 & \textbf I_3 \end{bmatrix}
\label{sec:IDMC; subec:forwardKR; eq:InvStructure.lc}
\end{align}
where (since $\hat A^l$ has zero row-block sums)
\begin{align}
\hat y_l & \ := \ y_l \, - \, \left(A^{l+1}_{12,1}\right)^{\sf T} \left( \textbf 1\otimes \textbf I_3 \right)
\label{sec:IDMC; subec:forwardKR; eq:InvStructure.ld}
\end{align}
and $\textbf 1$ is the vector of 1s of size $n_l$,
\label{sec:IDMC; subec:forwardKR; eq:InvStructure.l}
\end{subequations}
such that $A^l = P^l \hat A^l (P^l)^{\sf T}$, we compute $(\hat A^{l+1}, C^{l+1}, P^{l+1})$ so that:
\bi
\item The permuted matrix $\hat A^{l+1}$ has the same structure as that of $\hat A^l$. 
\item $C^{l+1}$ is the maximal clique of $\hat A^{l+1}$.
\item The Schur complement $A^{l+1} = A^{l}/\alpha^{l} = P^{l+1} \hat A^{l+1} (P^{l+1})^{\sf T}$ 
	(re-labeling nodes in $\hat A^{l+1}$).
\ei
As we will see, the one-step Kron reduction, and its reversal, boils down to the propagation of the maximal clique
(from \eqref{sec:IDMC; subec:forwardKR; eq:InvStructure.lb}):
\begin{align}
C^{l+1} & \ := \ \hat C^l/\alpha^l \ = \ \begin{bmatrix} \diag(\hat y_l) & 0 \\ 0 & C^l_{11} \end{bmatrix} 
	- \begin{bmatrix} -y_l \\ c^l_{12} \end{bmatrix} (\alpha^l)^{-1} \begin{bmatrix} -y_l^{\sf T} & \left( c^l_{12}\right)^{\sf T} \end{bmatrix} 
\label{sec:IDMC; subec:forwardKR; eq:Cl+1}
\end{align}
Specifically we compute $(\hat A^{l+1}, C^{l+1}, P^{l+1})$ and $A^{l+1}$ from $(\hat A^l, C^l, P^l)$ in 
\eqref{sec:IDMC; subec:forwardKR; eq:InvStructure.l} in four steps.  
\begin{subequations}
\bee
\item  From \eqref{sec:IDMC; subec:forwardKR; eq:InvStructure.la} and Lemma \ref{lemma:IterativeKronReduction.1},  
	the (unpermuted) Kron reduction $\hat A^l/\alpha^l$ is equal to:
\begin{align*}
	\hat A^{l}/\alpha^l & \ = \ 	
\left[ \begin{array} {c | c} \hat A^{l+1}_{11} & \hat A^{l+1}_{12} \\ \hline & \hat C^l/\alpha^l \end{array} \right] \ =: \
\left[ \begin{array} {c | c} \hat A^{l+1}_{11} & \hat A^{l+1}_{12} \\ \hline & C^{l+1} \end{array} \right]
\end{align*}
where $\hat A^{l+1}_{12} := \begin{bmatrix} A^{l+1}_{12,1} & A^{l+1}_{12, 2} \end{bmatrix}$ and
$C^{l+1}$ is given by \eqref{sec:IDMC; subec:forwardKR; eq:Cl+1}.

\item 
It can be checked that the last node $n-(l+1)$ remains the node to be Kron reduced in the next iteration
because the way the nodes in $y_l$ have been ordered.
Let the set of neighbors of node $n-(l+1)$ in the graph $G(\hat A^l/\alpha^l)$ that are not already in the maximal clique $C^{l+1}$ be
\begin{align*}
N_{n-(l+1)}  & \ := \ \left\{j \not\in C^{l+1} : (\hat A^l/\alpha^l)[j, n-(l+1)]\neq 0\in\mathbb C^{3\times 3} \right\}
\end{align*}
This set may contain both measured and hidden nodes.  Let $n_{l+1} := \left| N_{n-(l+1)} \right|$ be the number of these neighbors.
Order them in a way that maintains their relative order in $A^0$:
\begin{align*}
\underbrace{ i_1 < \dots < i_{m_{l+1}} }_{\text{measured nodes}} \ < \ \underbrace{ i_{m_{l+1}+1} < \dots < i_{n_{l+1}} }_{\text{hidden nodes}}
\end{align*}
Let $y_{l+1} \in\mathbb C^{3n_{l+1} \times 3}$ denote the admittance submatrices of lines connecting 
node $n-(l+1)$, to be Kron reduced in iteration $l+1$, to these neighbors: 
\begin{align*}
y_{l+1}[j,1] & \ := \ (\hat A^l/\alpha^l) [i_j,n-(l+1)], 	\qquad  j = 1, \dots, n_{l+1}
\end{align*}
Partition the maximal clique $C^{l+1}$ of $\hat A^l/\alpha^l$ into the last row and column block and other submatrices:
\begin{align}
C^{l+1} & \ =: \ \begin{bmatrix} C^{l+1}_{11} & c^{l+1}_{12} \\ \left( c^{l+1}_{12} \right)^{\sf T} & \alpha^{l+1} \end{bmatrix}
\label{sec:IDMC; subec:forwardKR; eq:InvStructure.l+1a}
\end{align}
where $\alpha^{l+1} := (A^l/\alpha^l)[n-(l+1), n-(l+1)]\in\mathbb C^{3\times 3}$.  

\item Re-label nodes in $\hat A^l/\alpha^l$, other than node $n-(l+1)$, so that the permuted matrix $\hat A^{l+1}$ has the
same structure as that in \eqref{sec:IDMC; subec:forwardKR; eq:InvStructure.la}.  Let the permutation matrix for re-labeling 
be $\hat P^l$ so that 
	\begin{align}
	\hat A^{l+1} & \ := \ \hat P^l \left( \hat A^l/\alpha^l \right) \left( \hat P^l \right)^{\sf T} \ =: \
	\left[ \begin{array} {c | c} A^{l+1}_{11} & A^{l+1}_{12} \\ \hline \left( A^{l+1}_{12} \right)^{\sf T} & C^{l+1} \end{array} \right]
	\ =: \  	\left[ \begin{array} {c c | c c}
	A^{l+2}_{11} & A^{l+2}_{12, 1} & A^{l+2}_{12, 2} & 0 \\ 
	& \cellcolor[rgb]{1,0.5,0.3}{\diag(\hat y_{l+1})} & \cellcolor[rgb]{1,0.5,0.3}{0} & \cellcolor[rgb]{1,0.5,0.3}{- y_{l+1}} \\ \hline
	&\cellcolor[rgb]{1,0.5,0.3}{} &  \cellcolor[rgb]{1,0.5,0.3}{C^{l+1}_{11}}  &  \cellcolor[rgb]{1,0.5,0.3}{c^{l+1}_{12}} \\ 
	& \cellcolor[rgb]{1,0.5,0.3}{} &  \cellcolor[rgb]{1,0.5,0.3}{} &  \cellcolor[rgb]{1,0.5,0.3}{\alpha^{l+1} } \\ 
	\end{array}  \right] 
\label{sec:IDMC; subec:forwardKR; eq:InvStructure.l+1b}
\end{align}
where (to maintain zero row-block sums of $\hat A^{l+1}$)
\begin{align}
\hat y_{l+1} & \ := \ y_{l+1} \, - \, \left(A^{l+2}_{12,1}\right)^{\sf T} \left( \textbf 1\otimes \textbf I_3 \right)
\label{sec:IDMC; subec:forwardKR; eq:InvStructure.l+1c}
\end{align}
and $\textbf 1$ is the vector of 1s of size $n_{l+1}$.

\item Define the permutation matrix
	\begin{align}
	P^{l+1} & \ := \ P^l_-\, \hat P^l 
	\label{sec:IDMC; subec:forwardKR; eq:InvStructure.l+1d}
	\end{align}
	where $P^l_-$ is defined in \eqref{sec:IDMC; subec:forwardKR; eq:InvStructure.lc} and $\hat P^l$ is defined
	in \eqref{sec:IDMC; subec:forwardKR; eq:InvStructure.l+1b}.  Then from 
	\eqref{sec:IDMC; subsec:BasicIdea; eq:P.1} and the fact that the square of a permutation matrix is the 
	identity matrix, the one-step Kron reduction of $A^{l+1} := A^{l}/\alpha^{l}$ can be recovered as
\begin{align}
A^{l+1} & \ = \ P^{l+1} \hat A^{l+1} (P^{l+1})^{\sf T}
\label{sec:IDMC; subec:forwardKR; eq:InvStructure.l+1e}
\end{align}
\eee
\label{sec:IDMC; subec:forwardKR; eq:InvStructure.l+1}
\end{subequations}

\paragraph{Return.}  A clique $\hat A^k = C^k$ of size $n-k$ and the Kron reduction
	$A^0/A_{22} = A^k = P^k \hat A^k (P^k)^{\sf T}$.

The important feature of this procedure is that the structure of the permuted matrix $\hat A^l$ in 
\eqref{sec:IDMC; subec:forwardKR; eq:InvStructure.la} 
is preserved under one-step Kron reduction where the maximal clique $C^l$ is a contiguous
block in the lower-right corner of $\hat A$.
This structure together with the way its maximal clique $C^l$ propagates through
\eqref{sec:IDMC; subec:forwardKR; eq:Cl+1} make it possible to reverse the Kron reduction as we explain
in Section \ref{sec:IDMC; subec:reverseKR}.
We now justify \eqref{sec:IDMC; subec:forwardKR; eq:InvStructure.l} and prove the correctness of the forward iterative Kron
reduction procedure above. 
Consider $A^0 =: \begin{bmatrix} A_{11} & A_{12} \\ A_{12}^{\sf T} & A_{22} \end{bmatrix}\in\mathbb C^{3n\times 3n}$
with a $3k\times 3k$ nonsingular submatrix $A_{22}$, $1\leq k < n$.   

\begin{theorem}[Invariant structure of $\hat A^l$]
\label{sec:IDMC; subec:forwardKR; thm:1stepKR}
Suppose $A^0$ is the admittance matrix of a single maximal clique of the form in 
\eqref{sec:IterativeKR; subsec:reverse; eq:Y.3} and it satisfies Assumption \ref{Assumption:y_jk}.   The procedure
above computes the Kron reduction $A^0/A_{22}$, i.e., for $l=0, \dots, k-1$,
\bee
\item The permuted matrix $\hat A^l$ has the form in \eqref{sec:IDMC; subec:forwardKR; eq:InvStructure.la}.  In particular, the entries of $A^{l+1}_{11}, A^{l+1}_{12,1}, A^{l+1}_{12, 2}$ 
in $\hat A^l$ are equal to the corresponding entries in $A^0$.

\item The matrix $A^{l+1}$ computed in \eqref{sec:IDMC; subec:forwardKR; eq:InvStructure.l+1e} is equal to $A^{l} / \alpha^{l}$.

\item At $l=k-1$, $A^{k}_{11} = 0$, $A^{k}_{12}=0$, $\hat A^k = C^k$, and $A^k = A^0/A_{22}$.
\eee
Moreover
\bee
\item[4.] In \eqref{sec:IDMC; subec:forwardKR; eq:InvStructure.ld}, if node $i$ in $y_l$ is a measured node (i.e., in $A_{11}$
	with an appropriate label in light of all the permutations by iteration $l$), then $\hat y_l[i,1] = y_l[i,1]$. 
\item[5.] If Assumption \ref{Assumption:HiddenNodes} holds then for each $l=0, \dots, k-1$, the number $n_l$
	of neighbors of node $n-l$ not in $C^l$ is at least 2.
\eee
\end{theorem}

\begin{proof}
\bee
\item We can always permute $A^l$ so that the permuted matrix $\hat A^l$ has its maximal clique $C^l$
	in its lower-right corner preceded by $y_l$.  We therefore have to prove the following features of $\hat A^l$ in
	\eqref{sec:IDMC; subec:forwardKR; eq:InvStructure.la}:
\bee
\item The zero matrix at the upper-right corner of $\hat A^l$.
\item The diagonal matrix $\diag(\hat y_l)$ with $\hat y_l$ given by \eqref{sec:IDMC; subec:forwardKR; eq:InvStructure.ld}
	and the zero matrix in $\hat C^l$ in \eqref{sec:IDMC; subec:forwardKR; eq:InvStructure.lb}.
\item The entries of $A^{l+1}_{11}$ and $A^{l+1}_{12} := \begin{bmatrix} A^{l+1}_{12,1} & A^{l+1}_{12, 2}  \end{bmatrix}$ in
	$\hat A^l$ in \eqref{sec:IDMC; subec:forwardKR; eq:InvStructure.la} being equal to the corresponding 
	entries in $A^{l}_{11}$ and $A^{l}_{12}$ respectively, i.e., these entries have not be changed by iterations 
	$0, 1, \dots, l$.
\eee
Assertion (a) follows from the construction of $y_{l}$ in each iteration $l$.  
For assertion (b), suppose $\diag(\hat y_l)$ in $\hat C^l$ should be a matrix $a_l$ which is not diagonal, in particular 
$a_l[i,j]\neq 0\in\mathbb C^{3\times 3}$ for some nodes $i, j$ in $y_l$.  This means that nodes $i$ and $j$ are adjacent in 
$G(\hat A^{l-1}/\alpha^{l-1})$, the graph of $\hat A^l$ before it is permuted by $\hat P^{l-1}$.
By Lemma \ref{lemma:IterativeKronReduction.1}, there is a unique path in the tree $G(\hat A^0)$ connecting 
nodes $i$ and $j$ (their labels may be changed in $G(\hat A^0)$ due to permutations along the way).
But both $i$ and $j$ are in $y_l$ and hence are neighbors of node $n-l$ in $G(\hat A^l/\alpha^l)$ by the construction
of $N_{n-l}$ and $y_l$.
Hence there is a unique path in $G(\hat A^0)$ connecting nodes $i$ and $n-l$ and another unique path in $G(\hat A^0)$ 
connecting $j$ and $n-l$ (possibly with different labels).  This creates a loop, contradicting
that $G(\hat A^0)$ is a tree.  The same argument proves the zero matrix in $\hat C^l$.

We prove assertion (c) by induction on $l$. For $l=0$, applying \eqref{lemma:Kron-reduced-y; eq:3} to $A^0$ in 
\eqref{sec:IDMC; subec:forwardKR; eq:InvStructure.0c} we have: for $i,j \in \{1, \dots, n-1\}$,
	\begin{align*}
	\left( \hat A^0/\alpha^0\right) [i,j] & \ := \ \left\{ \begin{array}{ll}
		\hat A^0[i,j] - \hat A^0[i,n] \! \left( \alpha^0 \right)^{-1}\! \hat A^0[j,n] & \ \text{if } i,j \in  \{n\} \cup N_n \\
			\hat A^0[i,j]	& \ \text{otherwise}
			\end{array} \right.
	\end{align*}
 	i.e., only entries for node $n$ and its neighbors need to be updated.  Since nodes in $A^1_{11}$ are not in $N_n$,
	$A^1_{11}$ and $A^1_{12,1}$ are equal to the corresponding submatrices of $\hat A^0$.  This proves assertion (c)
	for the base case $l=0$. Suppose assertion (c) holds for $l$.  For $l+1$, we have from
	\eqref{sec:IDMC; subec:forwardKR; eq:InvStructure.la},
	\begin{align*}
	\left( \hat A^{l}/\alpha^l \right) [i,j] & \ := \ \left\{ \begin{array}{ll}
		\hat A^l[i,j] - \hat A^l[i,n-l] \! \left( \alpha^l \right)^{-1}\! \hat A^l[j,n-l] & \ \text{if } i,j \in  C^l \cup N_{n-l} \\
			\hat A^l[i,j]	& \ \text{otherwise}
			\end{array} \right.
	\end{align*}
	Since nodes in $A^{l+1}_{11}$ are not in $C^l \cup N_{n-l}$,
	assertion (c) holds for $l+1$.  This completes the proof of part 1. 

\item Since $\hat A^0 = P^0 A^0 (P^0)^{\sf T}$ and $\hat A^1 = \hat P^0 (\hat A^0/\alpha^0) (\hat P^0)^{\sf T}$, 
	\eqref{sec:IDMC; subsec:BasicIdea; eq:P.1} and the fact that the square of a permutation matrix is the identity matrix
	imply that $A^1 = (P^0_- \hat P^0) \hat A^1 (P^0_- \hat P^0)^{\sf T}$.  Hence the assertion holds for $l=0$.  Suppose it
	holds for $l$.  For $l+1$, since $\hat A^{l+1} = \hat P^l \left( \hat A^l / \alpha^l \right) (\hat P^l)^{\sf T}$, the same reasoning
	shows that $A^{l+1} = P^{l+1} \hat A^{l+1} (P^{l+1})^{\sf T}$.  Therefore \eqref{sec:IDMC; subec:forwardKR; eq:InvStructure.l+1e} 
	holds for all $l$.

\item Part 3 follows from part 2 because the number of nodes in $A_{22}$ is $k$.

\item Since $\hat A^l$ has zero row and column-block sums, $\hat y_l$ is given by \eqref{sec:IDMC; subec:forwardKR; eq:InvStructure.ld}.
Suppose node $i$ in $y_l$ is a measured node (i.e. in $A_{11}$), a principal submatrix of $A^0$.  Then since $A^0$ is of the form in
\eqref{sec:IterativeKR; subsec:reverse; eq:Y.3}, node $i$ is of degree 1 and $i$ can only be adjacent to node $n-l$ in graph 
$G(\hat A^{l-1}/\alpha^{l-1})$.  This implies that the $i$th column block of $A^{l+1}_{12, 1}$ is a zero matrix, i.e., 
$A^{l+1}_{12,1}[j,i] = 0 \in\mathbb C^{3\times 3}$ for $j=1, \dots, n_l$.  Therefore $\hat y_l = y_l$ from
\eqref{sec:IDMC; subec:forwardKR; eq:InvStructure.ld}.

\item In the original graph $G^0 := G(A^0)$,  the node $n-l$ (with an appropriate label) is adjacent to exactly one of the nodes in the 
	maximal clique $C^l$, for otherwise, there will be a loop in the tree $G^0$.  Since Assumption \ref{Assumption:HiddenNodes} ensures 
	that node $n-l$ (to be Kron reduced in iteration $l$) has at least 3 neighbors in $G^0$, node $n-l$ must have at least 2 neighbors 
	that are not in $C^l$.  These neighbors cannot be nodes that have already been Kron reduced before iteration $l$ starts
	because those nodes are also connected (not necessarily adjacent) in $G^0$ to nodes in $C^l$.  
	These neighbors therefore must be in graph $G(\hat A^{l-1}/\alpha^{l-1})$ (with appropriate labels in light of
	 permutations up to iteration $l$), i.e., $n_l \geq 2$. 

\eee
This completes the proof of the theorem. 
\end{proof}

\begin{remark}
\bee
\item The entries of $\hat A^l$ corresponding to nodes not in the maximal clique $C^l$ remain unchanged until iteration
	$l$, i.e., $\hat A^l[i,j] = A^0[i',j']$ unless both $i,j \in C^l$ where $i', j'$ are the labels in $A^0$ corresponding to nodes
	$i, j$ in graph $G(\hat A^l)$ respectively.
	Take Example \ref{eg:IterativeKronReduction.1} and 
	Figure \ref{fig:IterativeKronReduction1}.  Consider the graph and its admittance matrix
$(G^3, A^3)$ after hidden nodes $12, 11, 10$ have been Kron reduced.  The maximal clique $C^3$ consists of nodes
$\{4,5,6,7, 9\}$. As indicated in Figure \ref{fig:IterativeKronReduction1}, $A^3[i,j] = A^0[i,j]$ for $i,j=1, 2, 3, 8, 9$, except 
the diagonal entry $A^3[9,9]$.

\item In each iteration $l-1$, there will always be at least two neighbors of $n-l$ in 
	$G(\hat A^{l-1}/\alpha^{l-1})$ that are not in the maximal clique $C^{l}$ but will be added to the maximal 
	clique $C^{l+1}$.  The size of $C^l$ therefore grows by at least 1 in each iteration and the algorithm
	terminates when the size of $C^k$ equals $n-k$.

\item Starting from $A^0$, to compute its Kron reduction $A^k = A^0/A_{22}$ iteratively, we only need
	to compute the permuted sequence $\hat A^l$, $l=0, \dots, k$, and then obtain $A^k$ from
	$\hat A^k$ at the end of the procedure.  It is not necessary to compute the intermediate unpermuted
	matrices $A^l$ for $l<k$.
\eee
\qed
\end{remark}

We now use Theorem \ref{sec:IDMC; subec:forwardKR; thm:1stepKR} to devise a method to compute 
$Y$ from its Kron reduction $\bar Y$ by reversing each step in the forward iterative Kron reduction
described in this subsection.

\subsection{Reverse Kron reduction: shrinking $C^l$ (Algorithm 2)}
\label{sec:IDMC; subec:reverseKR}

Given $A^0/A_{22}$, it is generally not possible to reverse the Kron
reduction to recover the original admittance matrix $A^0$.  This turns out to be possible, using the invariance
structure of $A^l$ in Theorem \ref{sec:IDMC; subec:forwardKR; thm:1stepKR}, when the network is a tree, 
as we now explain.

The basic idea is as follows.  Given the Kron reduction $A^k := A^0/A_{22}$, we will reverse each step in the
process illustrated in Figure \ref{fig:ReversibleKR}, focusing on computing the permuted sequence $\hat A^l$ in the reverse direction.
Each reverse iteration $l$ will involve three steps:
\bi
\item From $\hat A^{l+1}$ with the structure in \eqref{sec:IDMC; subec:forwardKR; eq:InvStructure.la}, 
	identify a set of ``sibling'' nodes in $C^{l+1}$ that are adjacent to a unique ``parent'' hidden node
	(this step requires Assumption \ref{Assumption:BoundaryHiddenNodes} below).  
	These sibling nodes define the set $N_{n-l}$ of nodes in $y_l$, to be determined, and their 
	parent node will be added in iteration $l$ to the set of identified nodes.
\item Reverse the Kron reduction of the maximal clique, i.e., compute $\hat C^l$ from $C^{l+1}$.  This allows
	us to construct $\hat A^l$ which will have the structure in \eqref{sec:IDMC; subec:forwardKR; eq:InvStructure.la}.
\item Permute $\hat A^l$ to obtain $A^l$ (this step can also be done only on $\hat A^0$ at the end).
\ei

To describe these steps in detail, it is important to keep track of what the algorithm knows at the beginning of each 
\emph{reverse iteration} and what it computes in that iteration.  Since the algorithm does not know the number $k := H$ of 
hidden nodes (nor the number $H_b\leq H$ of boundary hidden nodes) until the algorithm terminates, we will denote
the sequence $A^k, A^{k-1}, \dots, A^0$ by $\tilde A^0, \tilde A^1, \dots, \tilde A^k$ such that
\begin{align*}
\tilde A^0 & \ := \ A^k \ := \ Y/Y_{22}, 		
\quad \dots \quad \tilde A^{k-(l+1)} \ := \ A^{l+1}, \quad \tilde A^{k-l} \ := \ A^l, \quad\dots\quad 
\tilde A^k \ := \ A^0 \ := \ Y
\end{align*}
i.e., the identification algorithm starts in iteration $\tilde l = 0$ with the Kron reduced admittance matrix 
$\tilde A^{\tilde l} = A^k$ which is a single maximal clique.  It terminates at the end of iteration $\tilde l = k-1$
when the maximal clique $C^0$ of $\tilde A^{\tilde l + 1} := A^0$
is a $3\times 3$ matrix $Y[n,n]$.  When it is clear from the context that we do not need to know the value of $k$,
instead of index $\tilde l$ that progresses forward $\tilde l = 0, 1, \dots$, we will often use 
the index $l$ that progresses backward $l := k-\tilde l = k, k-1, \dots$, and describe our reverse algorithm in 
terms of $A^l$ instead of $\tilde A^{\tilde l} = \tilde A^{k-l}$, and their permuted matrices $\hat A^l$ and 
maximal cliques $C^l$.

The following assumption is
important for constructing the permuted matrix $\hat A^{l}$ from $\hat A^{l+1}$.

\begin{MyAssumption}{3}[Parent node]
\label{Assumption:BoundaryHiddenNodes}
Given any $\hat A^{l+1}$ and its maximal clique $C^{l+1}$, it is possible to determine the identity of a
set of all ``sibling'' nodes in $C^{l+1}$ that are adjacent to a common ``parent'' node in the original graph 
$G(A^0)$ but not in the graph $G(\hat A^{l+1})$.
\qed
\end{MyAssumption}

\begin{remark}[Assumption \ref{Assumption:BoundaryHiddenNodes}]
\bee
\item
The parent node not in $G(\hat A^{l+1})$ will be labeled by $(n-k)+\tilde l = n-l$ and added to $G(\hat A^l)$
in iteration $\tilde l$.   In the forward direction, $y_l$ is constructed from the graph $G(\hat A^{l-1}/\alpha^{l-1})$
(or $G(\hat A^l)$).  In the reverse direction, 
Assumption \ref{Assumption:BoundaryHiddenNodes} allows us to determine the identity of the nodes in
(but not the value of)
$y_l$ from the graph $G(\hat A^{l+1})$ and construct \eqref{sec:IDMC; subec:forwardKR; eq:Cl+1}.
We will provide a sufficient condition in Section \ref{sec:Assumption3} for Assumption 
\ref{Assumption:BoundaryHiddenNodes} to hold and Algorithm 3 to identify these sibling nodes.

\item
Even though the algorithm does not know $k$ nor $n$, it knows the number $n-k$ of measured nodes.
\qed
\eee
\end{remark}

The procedure to identify $A^0$ from its Kron reduction $A^0/A_{22}$ is as follows.

\paragraph{Algorithm 2: reverse iterative Kron reduction.}
\bee
\item[] \textbf{Given}:  $\tilde A^0 := A^0/A_{22} =: A^k \in\mathbb C^{3(n-k)\times 3(n-k)}$.  
\item[] \textbf{Initialize}: Let $\tilde l:=0$ and $l := k-\tilde l = k$.  Let
	\begin{align*}
	\hat A^k & \ := \ \tilde A^0 \ = \ A^k, \qquad C^k \ := \ \hat A^k, \qquad \tilde P^k \ := \ \textbf I_{3(n-k)\times 3(n-k)}
	\end{align*}
	
\item[] \textbf{Iterate} for $l=k-1, k-2, \dots$, until the maximal clique $C^{k-\tilde l} =: C^{l} \in\mathbb C^{3\times 3}$:

Given $(\hat A^{l+1}, C^{l+1}, \tilde P^{l+1}) := (\tilde A^{k-\tilde l+1}, C^{k-\tilde l+1}, \tilde P^{k-\tilde l+1})$ with
\begin{align}
	\hat A^{l+1} & \ =: \ 
	\left[ \begin{array} {c | c} A^{l+1}_{11} & A^{l+1}_{12} \\ \hline \left( A^{l+1}_{12} \right)^{\sf T} & C^{l+1} \end{array} \right]
	\ =: \  	\left[ \begin{array} {c c | c c}
	A^{l+2}_{11} & A^{l+2}_{12, 1} & A^{l+2}_{12, 2} & 0 \\ 
	& {\diag(\hat y_{l+1})} & {0} & {- y_{l+1}} \\ \hline
	& &  {C^{l+1}_{11}}  &  {c^{l+1}_{12}} \\ 	&   &   &  {\alpha^{l+1} } \\ 
	\end{array}  \right] 
\label{sec:IDMC; subec:reverseKR; eq:Al+1}
\end{align}
we compute $(\hat A^{l}, C^{l}, \tilde P^{l}) := (\tilde A^{k-\tilde l}, C^{k-\tilde l}, \tilde P^{k-\tilde l})$ as follows.

\bee
\item[1.] Identify all the ``sibling'' nodes in $C^{l+1}$ that are adjacent to a common ``parent'' 
	hidden node in $G(\hat A^0)$ but not in $G(\hat A^{l+1})$ 
guaranteed by Assumption \ref{Assumption:BoundaryHiddenNodes}.\footnote{
See Algorithm 3 in Section \ref{sec:Assumption3; subsec:CommonBHN} under
Assumption \ref{Assumption:UniformLines}.}  These nodes in $C^{l+1}$ define the set $N_{n-l}$ of nodes in $y_l$. 
The parent node is labeled $(n-k)+\tilde l = n-l$ in $\hat A^l$ 
(to be determined).   

\item[2.] Solve \eqref{sec:IDMC; subec:forwardKR; eq:Cl+1}, reproduced here:
\begin{align}
C^{l+1} & \ := \ \hat C^l/\alpha^l \ = \ \begin{bmatrix} \diag(\hat y_l) & 0 \\ 0 & C^l_{11} \end{bmatrix} 
	- \begin{bmatrix} -y_l \\ c^l_{12} \end{bmatrix} (\alpha^l)^{-1} \begin{bmatrix} -y_l^{\sf T} & \left( c^l_{12}\right)^{\sf T} \end{bmatrix} 
\label{sec:IDMC; subec:reverseKR; eq:Cl+1}
\end{align}
for $(y_l, \hat y_l)$ and $(C^l_{11}, c^l_{12}, \alpha^l)$.  

\item[3.] Substitute $(y_l, \hat y_l)$ and $(C^l_{11}, c^l_{12}, \alpha^l)$ as well as $(A^{l+1}_{11}, A^{l+1}_{12})$ from
	\eqref{sec:IDMC; subec:reverseKR; eq:Al+1} into 
\eqref{sec:IDMC; subec:forwardKR; eq:InvStructure.la}\eqref{sec:IDMC; subec:forwardKR; eq:InvStructure.lb} to
obtain $(\hat A^l, C^l)$.  Note that $\hat A^l$ has the same structure as $\hat A^{l+1}$ in
\eqref{sec:IDMC; subec:reverseKR; eq:Al+1} by construction, i.e., the invariance structure is preserved
in both the forward and the reverse directions.

\item[4.] The matrix $y_l\in\mathbb C^{3n_l\times 3n_l}$ may include both measured and hidden nodes and therefore 
the order in which the measured nodes appear in $\hat A^l$ may not agree with that in the given $A^k = \tilde A^0$. 
Re-label the nodes in $\hat A^l$ so that they agree and let the permutation matrix be $\tilde P^l$.  Then set
\begin{align*}
A^l & \ := \ \tilde P^l (\hat A^l) (\tilde P^l)^{\sf T}
\end{align*}
\eee

\item[] \textbf{Return}: $A^{\tilde k} = A^0$.
\eee
\begin{remark}
As for the forward Kron reduction, starting from the Kron reduction $A^k$, to compute $A^0$ iteratively,
we only need to compute the permuted sequence $\hat A^l$, $l=k, \dots, 0$, and then obtain $A^0$ from
$\hat A^0$ at the end of the procedure.  It is not necessary to compute the intermediate unpermuted
matrices $A^l$ for $l>0$, i.e., Step 4 in Algorithm 2 can be taken after the iteration has terminated.
\qed
\end{remark}

We now solve \eqref{sec:IDMC; subec:reverseKR; eq:Cl+1} first for $(y_l, \hat y_l)$ and then for $C^l$, or
equivalently $(C^l_{11}, c^l_{12}, \alpha^l)$.

\bee
\item
\noindent\emph{Computing $(y_l, \hat y_l)$ from \eqref{sec:IDMC; subec:reverseKR; eq:Cl+1}.}
Recall that $y_l, \hat y_l\in\mathbb C^{3n_l \times 3}$ and that Assumption \ref{Assumption:HiddenNodes} ensures
that $n_l\geq 2$ (Theorem \ref{sec:IDMC; subec:forwardKR; thm:1stepKR}).   
We have from 
\eqref{sec:IDMC; subec:forwardKR; eq:InvStructure.ld}
\begin{align*}
\hat y_l & \ := \ y_l \, - \, \underbrace{ \left(A^{l+1}_{12,1}\right)^{\sf T} \left( \textbf 1\otimes \textbf I_3 \right) }_{d_l}
\ =: \ y_l \, - \, d_l
\end{align*}
where $A^{l+2}_{12,1}$ is known and given in \eqref{sec:IDMC; subec:reverseKR; eq:Al+1}
and its row-block sums are abbreviated as $d_l \in\mathbb C^{3n_l \times 3}$.
Consider the first $n_l$ row blocks and $n_l$ column blocks of \eqref{sec:IDMC; subec:reverseKR; eq:Cl+1} and letting
\begin{align}
C^{l+1} & \ =: \ \begin{bmatrix} M_{11} & M_{12} \\ (M_{12})^{\sf T} & M_{22} \end{bmatrix}
\label{sec:IDMC; subec:reverseKR; eq:Cl+1partition}
\end{align}
where $M_{11}\in\mathbb C^{n_l\times n_l}$, we have:
\begin{align}
M_{11} & \ = \ \diag\left(y_l - d_l \right) \ - \ y_l \left( \alpha^l \right)^{-1} y_l^{\sf T}
\label{eq:barY_k1}
\end{align}
To simplify notation let $\tilde d_j := d_l[j]\in\mathbb C^{3\times 3}$ and $\tilde y_j := y_l[j] \in\mathbb C^{3\times 3}$
denote the $j$th block of $d_l$ and $y_l$ respectively with respect to the fixed $l$, and drop the superscript of 
$\alpha^l$ to write $\alpha$.
Since $n_l\geq 2$ we have four $3\times 3$ blocks at the top-left of \eqref{eq:barY_k1} and they are 
\begin{subequations}
\begin{align}
M_{11} [1,1] & \ = \ ( \tilde y_1 - \tilde d_1 ) - \tilde y_1\, \alpha^{-1} \tilde y_1	\\
M_{11} [2,2] & \ = \ ( \tilde y_2 - \tilde d_2 ) - \tilde y_2\, \alpha^{-1} \tilde y_2	\\
M_{11} [1,2] & \ = \  - \tilde y_1\, \alpha^{-1} \tilde y_2
\label{Step3; eq:1c}
\end{align}
and $M_{11} [2,1] = - \tilde y_2 \alpha^{-1} \tilde y_1 = 
	- \left(  - \tilde y_1 \alpha^{-1} \tilde y_2 \right)^{\sf T} =M_{11} [1,2]^{\sf T}$
by symmetry of $\tilde y_j$ and $\alpha$.
\label{Step3; eq:1}
\end{subequations}
To solve for $\tilde y_1, \tilde y_2, \alpha$ from \eqref{Step3; eq:1}, simplify the notation by setting
\begin{align}
a_1 & \ := \ M_{11}[1,1] +\tilde d_1, & a_2 & \ := \ M_{11}[2,2] + \tilde d_2, & a_3 \ := \ M_{11}[1,2]  
\label{sec:IDMC; subec:reverseKR; eq:a123}
\end{align}
Under Assumption \ref{Assumption:y_jk}, Lemma \ref{lemma:InvertibleSubmatrix} implies that $(\tilde y_1, \tilde y_2)$ 
and hence $a_3 := M_{11}[1,2]$ in \eqref{Step3; eq:1c} are invertible.
Moreover it can be shown, using \eqref{sec:IDMC; subec:forwardKR; eq:InvStructure.la}
and \eqref{sec:IDMC; subec:forwardKR; eq:Cl+1}, that $(a_1, a_2)$ are invertible as well.\footnote{Indeed $a_j$ is a diagonal block of the admittance matrix of the
clique $C^{l+1}$ in isolation, i.e., if the clique is the entire network with zero row and column-block sums.}
Then we have from \eqref{Step3; eq:1}
\begin{align*}
\tilde y_2 \tilde y_1^{-1} a_1 & \ = \ \tilde y_2 - \tilde y_2 \alpha^{-1} \tilde y_1, &  
\tilde y_1 \tilde y_2^{-1} a_2 & \ = \ \tilde y_1 - \tilde y_1\alpha^{-1} \tilde y_2
\end{align*} 
Since $\tilde y_1, \tilde y_2, \alpha^{-1}$ are symmetric, substitute $a_3 = - \tilde y_1 \alpha^{-1} \tilde y_2$ to 
eliminate $\alpha$:
\begin{align*}
\tilde y_1^{-1} a_1 & \ = \ \textbf I_3 + \tilde y_2^{-1} a_3, &  
\tilde y_2^{-1} a_2 & \ = \ \textbf I_3 + \tilde y_1^{-1} a_3
\end{align*} 
and hence
\begin{subequations}
\begin{align}
\tilde y_2^{-1} & \ = \ \tilde y_1^{-1}a_1 a_3^{-1} - a_3^{-1}, & \tilde y_2^{-1} & \ = \ \tilde y_1^{-1}a_3 a_2^{-1} + a_2^{-1}
\label{Step3; eq:2a}
\\
\tilde y_1^{-1} & \ = \ \tilde y_2^{-1} a_3 a_1^{-1} + a_1^{-1}, & \tilde y_1^{-1} & \ = \ \tilde y_2^{-1} a_2 a_3^{-1} - a_3^{-1}
\label{Step3; eq:2b}
\end{align}
\end{subequations}
Eliminating $\tilde y_2^{-1}$ from \eqref{Step3; eq:2a} and $\tilde y_1^{-1}$ from \eqref{Step3; eq:2b}, we 
have 
\begin{subequations}
\begin{align}
y_l[1] & \ =: \ \tilde y_1 \ = \ \left( a_1 a_3^{-1} - a_3 a_2^{-1} \right) \left( a_2^{-1} + a_3^{-1} \right)^{-1}
\label{sec:algorithm; subsec:IdentifyMaxClique; eq:tilde.a}
\\
y_l[2] & \ =: \ \tilde y_2 \ = \ \left( a_2 a_3^{-1} - a_3 a_1^{-1} \right) \left( a_1^{-1} + a_3^{-1} \right)^{-1}
\label{sec:algorithm; subsec:IdentifyMaxClique; eq:tilde.b}
\end{align}
where $a_1, a_2, a_3$ are defined in \eqref{sec:IDMC; subec:reverseKR; eq:a123}.
If $n_l\geq 3$, we have from the first row block of
\eqref{eq:barY_k1}, 
\begin{align}
y_l[j] & \ = \ - M_{11}[1,j] \left( y_l[1] \right)^{-1} \alpha^l, \qquad j=3, \dots, n_l	\\
\hat y_l & \ = \ y_l \, - \, d_l \ = \ y_l \, - \, \left(A^{l+1}_{12,1}\right)^{\sf T} \left( \textbf 1\otimes \textbf I_3 \right)
\end{align} 
where $\alpha^l$ is given by \eqref{sec:IDMC; subec:reverseKR; eq:Cl.a} below.
\label{sec:algorithm; subsec:IdentifyMaxClique; eq:tilde}
\end{subequations}
This completes the computation of $(y_l, \hat y_l)$ from \eqref{sec:IDMC; subec:reverseKR; eq:Cl+1}.

\item 
\noindent\emph{Computing $C^l$ from \eqref{sec:IDMC; subec:reverseKR; eq:Cl+1}.}
We compute $(C^l_{11}, c^l_{12}, \alpha^l)$.  For $\alpha^l$ we have from \eqref{Step3; eq:1c} 
\begin{subequations}
\begin{align}
\alpha^l & \ =: \ \alpha \ \ = \ - \tilde y_2 a_3^{-1} \tilde y_1 
\label{sec:IDMC; subec:reverseKR; eq:Cl.a}
\end{align}
where $a_3$ is defined in \eqref{sec:IDMC; subec:reverseKR; eq:a123} and $\tilde y_1, \tilde y_2$ are
given in 
\eqref{sec:algorithm; subsec:IdentifyMaxClique; eq:tilde.a}\eqref{sec:algorithm; subsec:IdentifyMaxClique; eq:tilde.b}.
For $c^l_{12}$, let $r := C^{l+1}[1, n_l+1:n-(l+1)]$ denote the first row block of $C^{l+1}$ from column block $n_l+1$ 
to the last column block $n-(l+1)$.  Then we have from \eqref{sec:IDMC; subec:reverseKR; eq:Cl+1} 
\begin{align*}
r & \ = \ y_l[1] \left( \alpha^l \right)^{-1} \left( c^l_{12} \right)^{\sf T}
\end{align*}
Since $r$ is known and $y_l[1], \alpha^l$ given in \eqref{sec:IDMC; subec:reverseKR; eq:Cl.a}
are invertible and symmetric, we have
\begin{align}
c^l_{12} & \ = \ r^{\sf T} \left( y_l[1] \right)^{-1} \alpha^l
\end{align}
Finally for $C^l_{11}$ we have, from
\eqref{sec:IDMC; subec:reverseKR; eq:Cl+1} and \eqref{sec:IDMC; subec:reverseKR; eq:Cl+1partition},
$M_{22} = C^l_{11} - c^l_{12} \left( \alpha^l \right)^{-1} \left( c^l_{12} \right)^{\sf T}$ and hence
\begin{align}
C^l_{11} & \ = \ M_{22}\, + \, c^l_{12} \left( \alpha^l \right)^{-1} \left( c^l_{12} \right)^{\sf T}
\end{align}
\label{sec:IDMC; subec:reverseKR; eq:Cl}
\end{subequations}
This completes the computation of $C^l$ from \eqref{sec:IDMC; subec:reverseKR; eq:Cl+1}.
 \eee

In summary, given the Kron reduction $\tilde A^0 := A^0/A_{22}$ of a single maximal clique, the
algorithm to identify the admittance matrix $A^0$ reverses the iterative Kron reduction procedure
step by step.  It computes $(\hat A^l, C^l, \tilde P^l)$ from $(\hat A^{l+1}, C^{l+1}, \tilde P^{l+1})$
until the single maximal clique $C^l$ shrinks to a $3\times 3$ matrix.

\section{Identification of sibling nodes}
\label{sec:Assumption3}

The identification algorithm in Section \ref{sec:IDMC; subec:reverseKR} for a single maximal clique
in isolation relies on Assumption \ref{Assumption:BoundaryHiddenNodes} that, in each iteration of the
reverse Kron reduction, given $\hat A^{l+1}$ and its maximal clique $C^{l+1}$, the algorithm can determine 
the identity of a set of sibling nodes in $C^{l+1}$ with a common parent node in the original graph 
$G(A^0)$ but not in the graph $G(\hat A^{l+1})$.  This allows the Algorithm 2 to determine from
$G(\hat A^{l+1})$ the identity of nodes in $y_l$  and construct \eqref{sec:IDMC; subec:reverseKR; eq:Cl+1}.

In this section we derive a method to fulfill Assumption \ref{Assumption:BoundaryHiddenNodes}
when the three-phase distribution lines are uniform (Assumption \ref{Assumption:UniformLines} below), 
as we now explain.

\subsection{Uniform lines}
\label{sec:Assumption3; subsec:UniformLine}

We assume that all lines are of the same type
specified by an impedance matrix $y^{-1}$ per unit length.  These lines differ only in their lengths. 
We call $y$ the \emph{unit admittance} and assume it satisfies
Assumption \ref{Assumption:y_jk}.
Consider an arbitrary $3n\times 3n$ complex symmetric matrix $A$ on a graph $ G^0 := ( N^0,  E^0)$
whose $3\times 3$ $(i,j)$th blocks $A[i,j]$ are given by:
\begin{align}
A[i,j] & \ = \ \left\{ \begin{array}{lcl}
    - y_{ij} & & (i,j)\in  E^0 \\ 	\sum_{m: (i,m)\in E^0} y_{im} & & i = j \\	 0 & & \text{otherwise}
    \end{array} \right.
\label{sec:Assumption3; subsec:UniformLine; eq:A0}
\end{align}
where $y_{ij} := \mu_{ij} y$ with $y\in\mathbb C^{3\times 3}$ 
and $\mu_{ij} = \mu_{ji} > 0$.  Let $A =: \begin{bmatrix} A_{11} & A_{12} \\ A_{12}^{\sf T} & A_{22} \end{bmatrix}$
with a $3k\times 3k$ nonsingular submatrix $A_{22}$, $1\leq k < n$.   

We show that the assumption of uniform lines is preserved under Schur complement, i.e., the effective line 
admittances of an arbitrary Kron-reduced admittance matrix $A/A_{22}$ are also specified by the unit 
admittance $y$.  Suppose Re$(y)\succ 0$ and $\mu_{ij}>0$ for all $(i,j)\in E^0$.
Then Lemma \ref{lemma:InvertibleSubmatrix} implies that $y^{-1}$ exists, is symmetric,
and Re$(y^{-1})\succ 0$.  Kron reduction preserves this structure as well.
\begin{theorem}
\label{lemma:Kron-reduced-y}
Suppose $y$ satisfies Assumption \ref{Assumption:y_jk} and $\mu_{ij} = \mu_{ji}>0$ 
for all $(i,j)\in E^0$ in the complex
symmetric matrix $A$ defined in \eqref{sec:Assumption3; subsec:UniformLine; eq:A0}.
\begin{enumerate}
\item 
The $3\times 3$  $(i,j)$th blocks $(A/A_{22})[i,j]$ of the Schur complement $A/A_{22}$ of 
$A_{22}$ of $A$ are given by
\begin{align}
(A/A_{22})[i,j] & \ = \ \left\{ \begin{array}{lcl}
    - \tilde \mu_{ij} \, y & & i \leadsto j \\
    \left( \sum_{m:i\leadsto m} \tilde \mu_{im} \right) y & & i = j \\
    0 & & \text{otherwise}
    \end{array} \right.
\label{lemma:Kron-reduced-y; eq:1}
\end{align}
for some $\tilde \mu_{ij} = \tilde \mu_{ji}>0$.  Here $i\leadsto j$ if and only if
there is a path in the underlying graph $G^0 := G(A)$ connecting nodes $i$ and $j$.

\item If Assumption \ref{Assumption:y_jk} is satisfied then
$(A/A_{22})^{-1}$ exists and is symmetric, and both Re$(A/A_{22})\succ 0$ and
Re$(A/A_{22})^{-1}\succ 0$.
\end{enumerate}
\end{theorem}
\begin{proof}
Since Re$(y)\succ 0$, Lemma \ref{lemma:InvertibleSubmatrix} implies that $y$ is nonsingular.
We follow the approach of \cite{DorflerBullo2013} to prove the lemma by induction on the nodes 
to be iteratively Kron reduced.  Recall the admittance matrices 
$A^0 := A, A^1, \dots, A^k = A/A_{22}$ and their underlying graphs $G^l = (N^l, E^l) := G(A^l)$, 
$l=0, \dots, k$, defined in Section \ref{sec:IterativeKR; subsec:1maxclique; subsubsec:IKR}. 
For $0<l<k$, let the induction hypothesis be
\begin{align}
A^l[i,j] & \ = \ \left\{ \begin{array}{lcl}
    - \mu^l_{ij} \, y & & (i,j)\in  E^l \\
    \left(\sum_{m:(i,m)\in E^l} \mu^l_{im} \right) y & & i = j \\
    0 & & \text{otherwise}
    \end{array} \right.
\label{lemma:Kron-reduced-y; eq:2}
\end{align}
for some $\mu^l_{ij} = \mu^l_{ji} >0$.  
Clearly $A^0$ satisfies \eqref{lemma:Kron-reduced-y; eq:2}.  Suppose $A^l$ satisfies 
\eqref{lemma:Kron-reduced-y; eq:2}.  We now prove using the one-step Kron reduction 
\eqref{lemma:Kron-reduced-y; eq:3} that $A^{l+1} := A^l / A^l[n-l, n-l]$ satisfies \eqref{lemma:Kron-reduced-y; eq:2}.   
Let 
\begin{align*}
\sigma^l_{i} & \ := \ \sum_{m:(i, m)\in E^l} \mu^l_{im}
\end{align*} 
so that $A^l[i, i] = \sigma^l_{i} y$.

From \eqref{lemma:Kron-reduced-y; eq:3}, $(i,j)\in E^{l+1}$ if and only if $(i,j)\in E^l$ or  
both $(i,n-l)\in E^l$ and $(j,n-l)\in E^l$.  Consider three cases by substituting the induction 
hypothesis \eqref{lemma:Kron-reduced-y; eq:2} into \eqref{lemma:Kron-reduced-y; eq:3}:
for $i,j\in N^{l+1} := \{1, \dots, n-l-1\}$:

\noindent\emph{Case 1: $i\neq j$ and $(i,j)\in E^{l+1}$.}
\begin{subequations}
\begin{enumerate}
\item If $(i,j)\in E^{l}$ but either $(i,n-l)\not\in E^l$ or $(j,n-l)\not\in E^l$
then we have $A^{l+1}[i,j] = A^l[i,j] = - \mu^{l+1}_{ij}\, y$
where 
\begin{align}
\mu^{l+1}_{ij} & \ := \ \mu^{l+1}_{ji} \ := \ \mu^l_{ij} > 0
\label{sec:Assumption3; subsec:UniformLine; eq:mul+1.1c}
\end{align}

\item 
If $(i,j)\not\in E^{l}$ but both $(i,n-l)\in E^l$ and $(j,n-l)\in E^l$
then, since $y$ is nonsingular,
\begin{align*}
A^{l+1}[i,j] & \ = \ - \mu^l_{i(n-l)}\, y 
    \left( \sigma^l_{n-l}\, y \right)^{-1} 
    \mu^l_{j(n-l)}\, y  \ =: \ - \mu^{l+1}_{ij}\, y
\end{align*}
where
\begin{align}
\mu^{l+1}_{ij} \ := \ \mu^{l+1}_{ji} & \ := \  \mu^l_{i(n-l)} \, \mu^l_{j(n-l)}
    \left( \sigma^l_{n-l} \right)^{-1} \ > \ 0
\label{sec:Assumption3; subsec:UniformLine; eq:mul+1.1a}
\end{align}

\item 
If $(i,j)\in E^{l}$, $(i,n-l)\in E^l$ and $(j,n-l)\in E^l$
then   
\begin{align*}
A^{l+1}[i,j] & \ := \ - \mu^l_{ij}\, y \, - \, \mu^l_{i(n-l)}\, y 
    \left( \sigma^l_{n-l}\, y \right)^{-1} \mu^l_{j(n-l)}\, y  \ =: \ - \mu^{l+1}_{ij}\, y
\end{align*}
where 
\begin{align}
\mu^{l+1}_{ij} \ := \ \mu^{l+1}_{ji} & \ := \ \mu^l_{ij} \, + \, 
    \mu^l_{i(n-l)} \, \mu^l_{j(n-l)} \left( \sigma^l_{n-l} \right)^{-1} 
  \ > \ 0
\label{sec:Assumption3; subsec:UniformLine; eq:mul+1.1b}
\end{align}
\eee
Hence $A^{l+1}[i,j] = - \mu^{l+1}_{ij} \, y$ for some $\mu^{l+1}_{ij} = \mu^{l+1}_{ji} > 0$
when $(i,j)\in E^{l+1}$.
\label{sec:Assumption3; subsec:UniformLine; eq:mul+1.1}
\end{subequations}

\noindent\emph{Case 2: $i=j$.}  We have to prove that $A^{l+1}$ has zero row-block
sums.
\bee
\item If $i=j$ but $(i,n-l)\not\in E^l$ then \eqref{lemma:Kron-reduced-y; eq:3} yields
$A^{l+1}[i,i] = A^l[i,i]$ and hence 
\begin{align*}
A^{l+1}[i,i] & \ = \ \left( \sum_{m: (i,m)\in  E^{l}} \mu^{l}_{im} \right) y 
	\ = \ \left( \sum_{m: (i,m)\in  E^{l+1}} \mu^{l+1}_{im} \right) y 
\end{align*}
where $\mu^{l+1}_{im}$ are defined in Case 1 above.  The last equality claims that
$\{m : (i,m)\in E^l \} = \{ m : (i,m)\in E^{l+1} \}$ which holds because $(i,n-l)\not\in E^l$.

\item If $i=j$ and $(i,n-l)\in E^l$ then \eqref{lemma:Kron-reduced-y; eq:3} yields
\begin{align}
A^{l+1}[i,i] & \ := \ \sigma^l_i y \, - \, \mu^l_{i(n-l)}\, y 
    \left( \sigma^l_{n-l} \, y \right)^{-1}  \mu^l_{i(n-l)}\, y
\nonumber \\
& \ = \ \left( \sum_{\substack{m\neq n-l \\ (i,m)\in E^l}} \mu^l_{im} \, + \, \mu^l_{i(n-l)} \left( 1 \, - \, 
    \frac{\mu^l_{i(n-l)}} { \sum_{m: (m,n-l)\in E^l} \mu^l_{m(n-l)} } \right) \right) y
\nonumber \\
& \ = \ \left( \sum_{\substack{m\neq n-l \\ (i,m)\in E^l}} \mu^l_{im} \, + \, 
	\sum_{\substack{m\neq i \\ (m,n-l)\in E^l}} \mu^l_{i(n-l)} \mu^l_{m(n-l)} 
	\left( \sigma^l_{n-l}\right)^{-1}    \right) y
\label{sec:Assumption3; subsec:UniformLine; eq:mul+1.2}
\end{align}
For $m\neq i$, $(i,m)\in E^{l+1}$ if and only if $(i,m)\in E^l$ or $(i,n-l)\in E^l, (m,n-l)\in E^l$.
Therefore for $m\not\in \{i, n-l\}$, since $(i,n-l)\in E^l$, the expression in the summations in 
\eqref{sec:Assumption3; subsec:UniformLine; eq:mul+1.2} reduces to 3 cases corresponding
to the 3 cases in \eqref{sec:Assumption3; subsec:UniformLine; eq:mul+1.1}:
\bi
\item[(a)] $\mu^l_{im}$ if $(i,m)\in E^l$ and $(m,n-l)\not\in E^l$.
\item[(b)] $\mu^l_{i(n-l)} \mu^l_{m(n-l)} \left( \sigma^l_{n-l}\right)^{-1}$ if $(i,m)\not\in E^l$ and 
	$(m,n-l)\in E^l$.
\item[(c)] $\mu^l_{im} + \mu^l_{i(n-l)} \mu^l_{m(n-l)} \left( \sigma^l_{n-l}\right)^{-1}$ if 
	$(i,m)\in E^l$ and $(m,n-l)\in E^l$.
\ei
From \eqref{sec:Assumption3; subsec:UniformLine; eq:mul+1.1}, these expressions all equal
$\mu^{l+1}_{im} = \mu^{l+1}_{mi}$.  Hence
\begin{align*}
A^{l+1}[i,i] & \ = \ \left(\sum_{m: (i,m)\in E^{l+1}} \mu^{l+1}_{im} \right) y
\end{align*}
\eee

\noindent\emph{Case 3: $i\neq j$ and $(i,j)\not\in E^{l+1}$.}
Then \eqref{lemma:Kron-reduced-y; eq:3} implies that $A^{l+1}[i,j] = 0$.
This completes the induction and the proof of part 1.

Part 2 follows from Lemma \ref{lemma:InvertibleSubmatrix}.
\end{proof}

Theorem \ref{lemma:Kron-reduced-y} motivates the following assumption on line admittance matrices
that says that the series impedance of a line $(i,j)$ depends only on its length $\lambda_{ij}>0$.
\begin{MyAssumption}{4}[Uniform lines] 
\label{Assumption:UniformLines}
For all lines $(i,j)\in E$, $y_{i,j} =: \lambda_{i,j}^{-1} y$ where $y\in\mathbb C^{3\times 3}$ 
satisfies Assumption \ref{Assumption:y_jk} and $\lambda_{ij} = \lambda_{ji}>0$.
\qed
\end{MyAssumption}


\subsection{Sibling nodes (Algorithm 3)}
\label{sec:Assumption3; subsec:CommonBHN}

Consider $\hat A^{l+1}$ in iteration $l$ of the reverse Kron reduction
(from \eqref{sec:IDMC; subec:reverseKR; eq:Al+1}):
\begin{align}
	\hat A^{l+1} &  \ =: \
	\left[ \begin{array} {c | c} A^{l+1}_{11} & A^{l+1}_{12} \\ \hline \left( A^{l+1}_{12} \right)^{\sf T} & C^{l+1} \end{array} \right]
\label{sec:Assumption3; subsec:CommonBHN; eq:Cl+1}
\end{align}
Assumption \ref{Assumption:BoundaryHiddenNodes} allows the algorithm to determine 
the identity of a set of all sibling nodes in $C^{l+1}$ that are adjacent to a parent hidden 
node that has not been identified.  This new 
hidden node will be labeled as node $n-l$ and added to $\hat A^l$ in 
\eqref{sec:IDMC; subec:forwardKR; eq:InvStructure.la} and its neighbors define the 
identity of nodes in $y_l$.  
We now describe a method to make this determination under Assumption \ref{Assumption:UniformLines}, in two steps.  First we prove the method for the base case of $C^k = Y/Y_{22}$ in the next theorem.  The theorem extends 
\cite[Proposition 3]{YuanLow2023} from single-phase networks to three-phase networks.  Then we use 
the theorem to fulfill Assumption  \ref{Assumption:BoundaryHiddenNodes}  for the general case of $C^{l+1}$.
 
\begin{theorem}
\label{thm:SameHiddenNode}
Suppose $\hat A^0 := A^0 := Y$ is an admittance matrix of the form in 
\eqref{sec:IterativeKR; subsec:reverse; eq:Y.3} and it satisfies Assumptions \ref{Assumption:y_jk}, \ref{Assumption:HiddenNodes} and \ref{Assumption:UniformLines}.
Given the maximal clique $C^k = \hat A^k := Y/Y_{22}$, two nodes $i$ and $j$ in 
$C^k$ are adjacent to the same hidden node $h(i)$ in $\hat A^0$ (but not in $G(\hat A^k)$ if 
and only if the $3\times 3$ matrices $C^k[i, m]$ and $C^k[j,m]$ satisfy
\begin{align*}
C^k[i, m] & \ = \ \gamma(i,j)\, C^k[j,m], \qquad \forall m\neq i, j, \ m \in C^k
\end{align*}
for some nonsingular matrix $\gamma(i,j)\in\mathbb C^{3\times 3}$ that does not depend 
on $m\neq i, j$.
\end{theorem}

\begin{proof}
When $\hat A^0 = Y$, $n-k = M$ is the number of measured nodes in $G(\hat A^0)$, 
$H_b$ the number of boundary hidden nodes, and $k = H$ the total number of 
(boundary and internal) hidden nodes.  Recall from \eqref{sec:IterativeKR; subsec:reverse; eq:Y.3} 
that each measured node in $G(\hat A^0)$ is adjacent to only a single hidden node $h(i)$ since 
$G(\hat A^0)$ is a tree.  Specifically
\begin{subequations}
\begin{align}
\label{sec:Assumption3; eq:Y.3a}
    \hat A^0 & \ =: \ \left[ \begin{array}{c | c}
    A_{11}^0 & A^0_{12} \\ \hline A^0_{21} & A^0_{22}
    \end{array} \right]   
    \ =: \ \left[ \begin{array}{c | c c}
     A^0_{11, 22} & A^0_{12, 21} & 0 
    \\ \hline
     & A^0_{22, 11} & A^0_{22, 12} 
    \\
     & & A^0_{22, 22} 
    \end{array} \right]
\end{align}
where 
\begin{align}
A^0_{11,22} & \ = \ \diag\begin{bmatrix} y_{1h(1)} \\ \vdots \\ y_{Mh(M)}    \end{bmatrix}, 	&
A^0_{12, 21} & \ = \ \begin{bmatrix} - e_{h(1)}^{\sf T}\otimes y_{1h(1)} \\ \vdots \\  - e_{h(M)}^{\sf T}\otimes y_{Mh(M)} \end{bmatrix}
\label{sec:Assumption3; eq:Y.3b}
\end{align}
where $e_i\in\{0, 1\}^{H_b}$ is the unit vector with a single 1 in the $i$th entry and 0 elsewhere, and $y_{ij}\in\mathbb C^{3\times 3}$ is the three-phase series admittance of line $(i,j)$.  
\label{sec:Assumption3; eq:Y.3}
\end{subequations}
(See Example \ref{sec:IterativeKR; subsec:reverse; eg:Y}.)
Under Assumption \ref{Assumption:y_jk}, $A^0_{22}$ in \eqref{sec:Assumption3; eq:Y.3a}
is invertible according to Lemma \ref{lemma:InvertibleSubmatrix}.  

Recall the definition of $Z_{22}$ in \eqref{sec:algorithm; eq:Step1.1a}:
\begin{align}
Z_{22} & \ := \ \left( A^0_{22}\right)^{-1} \ =: \ 
\begin{bmatrix}
X_{22, 11} & X_{22, 12} \\ X_{22, 21} & X_{22, 22}
\end{bmatrix}
\label{sec:Assumption3; eq:Z_22}
\end{align}
Lemma \ref{lemma:InvertibleSubmatrix} also implies that 
$A^0_{22}$ and $Z_{22}$ are symmetric matrices.
Substituting into \eqref{sec:Assumption3; eq:Y.3a}, we can write 
$\hat A^0/A^0_{22} = A^0_{11} - A^0_{12} Z_{22} A^0_{21}$ in terms of $X_{22,11}$: 
\begin{align*}
\bar A^0 & \ := \ \hat A^0/A^0_{22} \ = \ A^0_{11, 22}  \, - \, A^0_{12, 21} X_{22, 11} A^0_{21,12} 
\end{align*}
with $A^0_{21,12} = \left( A^0_{12, 21} \right)^{\sf T}$.
From \eqref{sec:Assumption3; eq:Y.3a}\eqref{sec:Assumption3; eq:Z_22} and \eqref{eq:A-1.b} in
Section \ref{sec:Assumption3; subsec:SchurComplement}, $X_{22,11} = \left( A^0_{22}/A^0_{22,22} \right)^{-1}$
which exists under Assumption \ref{Assumption:y_jk} according to Lemma \ref{lemma:InvertibleSubmatrix}.  
Partition the $3H_b \times 3H_b$ matrix $X_{22,11}$ into $H_b^2$ blocks each of $3\times 3$:
\begin{align*}
X_{22, 11} & \ =: \ \begin{bmatrix}
    \beta_{11} & \beta_{12} & \cdots & \beta_{1H_b} \\  \beta_{21} & \beta_{22} & \cdots & \beta_{2H_b} \\     
    \vdots & \vdots & \vdots & \vdots    \\  \beta_{H_b 1} & \beta_{H_b 2} & \cdots & \beta_{H_b H_b}
    \end{bmatrix}
\end{align*}
where $\beta_{kl} = \beta_{lk}^{\sf T} \in\mathbb C^{3\times 3}$.  
Substituting into \eqref{sec:Assumption3; eq:Y.3b}, the $(i,j)$th block of $A^0_{12, 21} X_{22, 11} A^0_{21,12}$ is,
using $A^{\sf T}\otimes B^{\sf T} = (A\otimes B)^{\sf T}$, 
\begin{align*}
A^0_{12, 21} X_{22, 11} A^0_{21,12}[i,j] & \ = \ 
\left(e_{h(i)}^{\sf T}\otimes y_{ih(i)}\right) X_{22, 11} \left( e_{h(j)}\otimes y_{jh(j)} \right)
\end{align*}
Fix any $i\neq j$.  The $i$th and $j$th row blocks of $\hat A^0$ are respectively 
\begin{align*}
& y_{ih(i)} \begin{bmatrix} \beta_{h(i)h(1)}\, y_{1h(1)} & \cdots & \beta_{h(i)h(M)}\, y_{M h(M)} \end{bmatrix}
\\
& y_{jh(j)} \begin{bmatrix} 
\beta_{h(j)h(1)}\, y_{1h(1)} & \cdots & \beta_{h(j)h(M)}\, y_{M h(M)} \end{bmatrix}
\end{align*}
Therefore, for $m\neq i, j$, $m=1, \dots, M$, we have 
\begin{align}
\Bar Y[i, m] & \ = \ y_{ih(i)}\, \beta_{h(i)h(m)} \, y_{mh(m)}, & 
\Bar Y[j,m] & \ = \ y_{jh(j)}\, \beta_{h(j)h(m)} \, y_{mh(m)}    & &
\label{sec:HiddenNodes; subsec:proof; eq:BarY2jk}
\end{align}
If $i$ and $j$ are adjacent to the same hidden node $h(i)=h(j)$, then
$\Bar Y[i, m] = \gamma(i,j)\, \Bar Y[j,m]$ with $\gamma(i,j) := y_{ih(i)}\, y_{jh(j)}^{-1}$ for $m\neq i, j$, 
where $y_{jh(j)}\in\mathbb C^{3\times 3}$ is invertible by Assumption \ref{Assumption:y_jk}.
Clearly $\gamma(i,j)$ is invertible.

Conversely suppose $h(i)\neq h(j)$ but $\Bar Y[i, m] = \gamma'(i,j)\, \Bar Y[j,m]$ for all $m\neq i, j$, 
for some 
invertible $\gamma'(i,j)\in\mathbb C^{3\times 3}$.
From \eqref{sec:HiddenNodes; subsec:proof; eq:BarY2jk} and invertibility of $y_{mh(m)}$, this is equivalent to
\begin{align}
\gamma \beta_{h(i)h(m)} & \ = \ \beta_{h(j)h(m)}, \qquad m\neq i, j
\label{sec:HiddenNodes; subsec:proof; eq:betaij}
\end{align}
where the $3\times 3$ matrix $\gamma := \gamma(i,j) := \left( \gamma'(i,j)\, y_{jh(j)} \right)^{-1} y_{ih(i)}$.
We will use \eqref{sec:HiddenNodes; subsec:proof; eq:betaij} to derive a contradiction.

Assume without loss of generality that $h(i) = 1$ and $h(j) = 2$ (corresponding 
to nodes $M+1$ and $M+2$ respectively when $\hat A^0 = Y$).  By Lemma \ref{lemma:InvertibleSubmatrix}, 
$Y$ is symmetric and hence $\bar A^0$, $Z_{22} := \left( A^0_{22}\right)^{-1}$, and $X_{22,11}$
are symmetric and invertible.   Then \eqref{sec:HiddenNodes; subsec:proof; eq:betaij} 
means that the matrix $X_{22,11} = \left( A^0_{22}/A^0_{22,22} \right)^{-1}$ is of the form
\begin{align}
X_{22,11} & \ = \ \left[ \begin{array} {c c | c c c}
\beta_{11} & \beta_{12} & \beta_{13} & \cdots & \beta_{1H_b} \\
\beta_{12} & \beta_{22} & \gamma \beta_{13} & \cdots & \gamma\beta_{1H_b} \\ \hline
\beta_{13} & \gamma \beta_{13} & \beta_{33} & \cdots & \beta_{3H_b} \\
\vdots & \vdots & \vdots & \ddots & \vdots \\
\beta_{1H_b} & \gamma \beta_{1H_b} & \beta_{3H_b} & \cdots & \beta_{H_bH_b} 
\end{array} \right]  
\ =: \ \begin{bmatrix} B_{11} & B_{12} \\ B_{12}^{\sf T} & B_{22} \end{bmatrix}
\label{eq:X_2211.2}
\end{align}
where $B_{11}\in\mathbb C^{6\times 6}$ and $B_{22}\in\mathbb C^{3(H_b-2) \times 3(H_b-2)}$.
Hence $B_{12}\in\mathbb C^{6 \times 3(H_b-2)}$ is singular.  Moreover, since $\gamma\in\mathbb C^{3\times 3}$ has 
three independent rows $\gamma_1^{\sf T}, \gamma_2^{\sf T}, \gamma_3^{\sf T}$, the null space of $B_{12}^{\sf T}$
has dimension (nullity) at least 3 because the null space contains at least the column vectors
$(\gamma_i, -e_i)\in\mathbb C^6$, $i=1, 2, 3$.
By the rank-nullity theorem, rank$(B_{12}) \leq 3$.

Suppose more than one measured node are adjacent to each of the hidden nodes 
$h(i)$ and $h(j)$.  Specifically let $h(i) = h(i')$ and $h(j) = h(j')$ for distinct $i'$, $j'$.  
Then letting $m = i'$ and then $m=j'$ in \eqref{sec:HiddenNodes; subsec:proof; eq:betaij} implies that 
\begin{align*}
\gamma \beta_{h(i)h(i)} & \ = \ \gamma \beta_{h(i)h(i')} \ = \ \beta_{h(j)h(i')} 
	\ = \ \beta_{h(j)h(i)}	\\
\beta_{h(j)h(j)} & \ = \ \beta_{h(j)h(j')} \ = \ \gamma \beta_{h(i)h(j')}
	\ = \ \gamma\beta_{h(i)h(j)}
\end{align*}
This implies that $\beta_{21} = \beta_{12} = \gamma\beta_{11}$ and $\beta_{22} = \gamma \beta_{12}$
in \eqref{eq:X_2211.2}, i.e., \eqref{sec:HiddenNodes; subsec:proof; eq:betaij} holds for all $m$.
Hence, as for $B_{12}$, the first two row blocks of $X_{22,11}$ also are of rank no more than 3.
This contradicts the invertibility of $X_{22,11}$ and hence $h(i)=h(j)$ if more than one measured
node are adjacent to each of $h(i)$ and $h(j)$.

Suppose then at least one of $h(i)$, $h(j)$ is adjacent to only a single measured node 
(i.e., $i$ or $j$) 
so that $\gamma \beta_{h(i)h(m)} = \beta_{h(j)h(m)}$ in \eqref{sec:HiddenNodes; subsec:proof; eq:betaij}
may not hold for $m=i$ or $m=j$ or both.  From \eqref{sec:Assumption3; eq:Y.3a} and \eqref{eq:A-1.b} we have
\begin{align}
X_{22,11}^{-1} \ = \ A^0_{22}/A^0_{22,22} & \ =: \ 
\begin{bmatrix} A_{11} & A_{12} \\ A_{12}^{\sf T} & A_{22} \end{bmatrix}
\label{eq:X_2211^-1}
\end{align}
where $A_{11}\in\mathbb C^{6\times 6}$ and $A_{22}\in\mathbb C^{3(H_b-2)\times 3(H_b-2)}$.
Hence $X_{22,11}^{-1}$, being the Schur complement of $A^0_{22, 22}$ of $A^0_{22}$, can be interpreted as
the Kron-reduced admittance matrix for the graph connecting boundary hidden nodes when all internal
hidden nodes are Kron reduced, but with nonzero
shunt admittances from their connectivity to boundary measured nodes.  Applying \eqref{eq:A-1.b}
to \eqref{eq:X_2211.2} and \eqref{eq:X_2211^-1} we have
\begin{align}
A_{12} & \ = \ - (X_{22,11}/B_{22})^{-1} B_{12} B_{22}^{-1} 
    \ = \ - A_{11} B_{12} B_{22}^{-1} 
\label{eq:A_{12}}
\end{align}
where $B_{22}^{-1}$ exists by Lemma \ref{lemma:InvertibleSubmatrix}.  Hence rank$(A_{12})\leq \text{rank}(B_{12}) \leq 3$.

Assumption \ref{Assumption:UniformLines} is now needed to establish a contradiction.
Under Assumption \ref{Assumption:UniformLines}, the $(l,m)$th blocks $\bar A^0[l,m]$ of the Kron-reduced 
admittance matrix $\bar A^0$ are given by:
\begin{subequations}
\begin{align}
\bar A^0[l,m] & \ = \ \left\{ \begin{array}{lcl}
    - \tilde \lambda_{lm}^{-1} \, y & & l \leadsto m \\
    \left( \sum_{m':l\leadsto m'} \tilde \lambda_{lm'}^{-1} \right) y & & l = m \\
    0 & & \text{otherwise}
    \end{array} \right.
\label{eq:UniformLine.1}
\end{align}
for some $\tilde \lambda_{lm} = \tilde \lambda_{ml}>0$.  Here $l\leadsto m$ if and only if
$(l,m)\in E$ or there is a path in the graph $G(\hat A^0)$ connecting nodes $l$ and $m$.
As mentioned above, $X_{22,11}^{-1} = A^0_{22}/A^0_{22,22}$ exists
under Assumption \ref{Assumption:y_jk} according to Lemma \ref{lemma:InvertibleSubmatrix}.  
Moreover the line admittances in \eqref{eq:UniformLine.1} are of full rank:
\begin{align}
\text{rank}\left( ( A^0_{22}/ A^0_{22,22})[l,m] \right) & \ = \ 3, \quad\quad
\text{if } l\leadsto m \text{ of } l=m
\label{eq:UniformLine.2}
\end{align}
This means that every nonzero
\label{eq:UniformLine.12}
\end{subequations}
$(l,m)$th block $X_{22,11}^{-1}[l,m]\in\mathbb C^{3\times 3}$ in \eqref{eq:X_2211^-1} is of full rank, 
and hence either rank$(A_{12}) = 0$ or rank$(A_{12}) = 3$.
We claim that there is a contradiction in each case, proving  $h(i)=h(j)$.
\begin{enumerate}
\item \emph{rank$(A_{12}) = 0$}:
In this case $A_{12} = 0$ the zero matrix.  This means that $A_{11}$ is full rank since $A^0_{22}/A^0_{22,22}$ is nonsingular under Assumption \ref{Assumption:y_jk}.  Hence, in $G(\hat A^0)$, nodes $h(i)$ and $h(j)$ are not adjacent to other (boundary or internal) hidden nodes except to each other and measured nodes.  If one of 
$h(i)$ and $h(j)$ is adjacent to only a single measured node in $G^(\hat A^0)$, say $i$, then $h(i)$ violates the requirement in Assumption \ref{Assumption:HiddenNodes} that each hidden node is adjacent to at least 3 other 
nodes in $G(\hat A^0)$.

\item \emph{rank$(A_{12}) = 3$}:
We derive the structure of $A_{12}$ using 
\eqref{eq:X_2211.2}\eqref{eq:X_2211^-1}\eqref{eq:A_{12}} and \eqref{eq:UniformLine.12}.
Under Assumption \ref{Assumption:UniformLines},  \eqref{eq:UniformLine.1} says that the $6\times 6$ submatrix $A_{11}$ in \eqref{eq:X_2211^-1} is of the form:
\begin{align*}
A_{11} & \ =: \ \begin{bmatrix} \tilde\mu_{11}\, y & -\tilde\mu_{12}\, y \\
-\tilde\mu_{12}\, y & \tilde\mu_{22}\, y
\end{bmatrix}
\end{align*}
for some positive $\tilde\mu_{11},\tilde\mu_{12},\tilde\mu_{22}$.
Denote the row blocks of $B_{12}$ in \eqref{eq:X_2211.2} by:
\begin{align*}
B_{12} & \ =: \ \begin{bmatrix} b^{\sf T} \\ \gamma b^{\sf T} \end{bmatrix}
\end{align*}
where $b^{\sf T}\in\mathbb C^{3\times 3(H_b-2)}$ and $\gamma\in\mathbb C^{3\times 3}$.
Substituting into \eqref{eq:A_{12}} we have
\begin{align}
A_{12} & \ = \ -  \begin{bmatrix} \tilde\mu_{11}\, y & -\tilde\mu_{12}\, y \\
-\tilde\mu_{12}\, y & \tilde\mu_{22}\, y
\end{bmatrix}
\begin{bmatrix} \textbf I_3 & 0 \\ 0 & \gamma \end{bmatrix}
\begin{bmatrix} b^{\sf T} B_{22}^{-1} \\ b^{\sf T} B_{22}^{-1}
\end{bmatrix} 
\ = \ \begin{bmatrix} \tilde y_1 \\ \tilde y_2 \end{bmatrix} \left( b^{\sf T} B_{22}^{-1} \right)
\label{thm:SameHiddenNode; eq:A_12}
\end{align}
where $\textbf I_3$ is the identity matrix of size 3, 
$\tilde y_1 := - \tilde\mu_{11}\, y + \tilde\mu_{12}\, y \gamma$, and 
$\tilde y_2 := \tilde\mu_{12}\, y - \tilde\mu_{22}\, y \gamma$.
We know from \eqref{eq:UniformLine.12} that each of the two $3\times 3(H_b-2)$ row blocks of $A_{12}$ 
is of full rank.\footnote{The only way $A_{12}$ can be of rank 3 is if $\tilde y_1 = \rho\tilde y_2$ for some nonzero $3\times 3$ matrix $\rho$.}
Moreover \eqref{thm:SameHiddenNode; eq:A_12} implies that
these two row blocks have zero and nonzero $3\times 3$ blocks at exactly the same locations since each is a multiple of $b^{\sf T}B_{22}^{-1}$.  This means, from \eqref{eq:X_2211^-1}, that the nodes $h(i)$ and $h(j)$ are adjacent to the same set of boundary hidden nodes in the graph $G(A^0_{22}/A^0_{22,22})$ other than, possibly, to each other.
Since rank$(A_{12})>0$, $h(i)$ and $h(j)$ are both adjacent to at least one common boundary hidden node, 
say node $l$, in $G(A^0_{22}/A^0_{22,22})$ other than each other.  There are two possibilities: either $h(i)$ and
$h(j)$ are both adjacent to the boundary hidden node $l$ in the original graph $G(\hat A^0)$, or neither 
$h(i)$ nor $h(j)$ is adjacent to $l$ in $G(\hat A^0)$ but they are connected to $l$ through the tree of 
\emph{internal} hidden nodes in $G(\hat A^0)$.
In both cases $h(i)$ and $h(j)$ are adjacent to exactly one (boundary or internal) hidden node in $G(\hat A^0)$ 
and are not adjacent to each other, for otherwise there is a loop in graph $G(\hat A^0)$. 
Again, if one of $h(i)$ and $h(j)$ is adjacent to only a single measured node, say $i$, then $h(i)$ violates 
Assumption \ref{Assumption:HiddenNodes} that each hidden node is adjacent to at least 3 other nodes in 
$G(\hat A^0)$.
\end{enumerate}
This completes the proof of Theorem \ref{thm:SameHiddenNode}.
\end{proof}

To fulfill Assumption \ref{Assumption:BoundaryHiddenNodes} of Section \ref{sec:IDMC; subec:reverseKR},
we apply Theorem \ref{thm:SameHiddenNode} in each step of the reverse iterative Kron reduction.  
The theorem applies directly in the base case when the maximal 
clique $(\hat A^k, C^k)$ is given.  For each iteration $l\leq k-1$, given $(\hat A^{l+1}, C^{l+1})$ in
\eqref{sec:Assumption3; subsec:CommonBHN; eq:Cl+1}, as long as 
$C^{l+1}\not\in \mathbb C^{3\times 3}$, we can remove from the graph $G(\hat A^{l+1})$ 
those nodes that are not in the clique $C^{l+1}$ and the resulting admittance matrix is obtained by
normalizing $C^{l+1}$ so that it has zero row and column-block sums (cf. \eqref{eq:solution.3}), 
i.e., 
\begin{align*}
C^{l+1'} &\ = \ C^{l+1} - \diag\left( \left(\textbf 1 \otimes \textbf I_3 \right)^{\sf T} C^{l+1} \right)
\end{align*}
Then $C^{l+1'}$ is a maximal clique to which we can apply Theorem \ref{thm:SameHiddenNode} 
to identify a set of all sibling nodes in $C^{l+1'}$ that are adjacent to a common parent node in 
$G(\hat A^0)$ (but not in $G(\hat A^{l+1})$).  We will label this hidden node by $n-l$ and add it
to $G(\hat A^l)$.
Theorem \ref{sec:IDMC; subec:forwardKR; thm:1stepKR} guarantees that there are at least 2 such
sibling nodes in each iteration $l$ that make up $y_l$.  
If there are more than one group of sibling nodes in $C^{l+1'}$, each with its own parent  
node, then pick any one of these groups of sibling nodes and label their parent node as node $n-l$.
This fulfills Step 1 in each iteration $l$ of Algorithm 2 for reverse iterative Kron reduction.
Specifically, under Assumption \ref{Assumption:UniformLines}, the following algorithm can be executed at the beginning
of each iteration $l$.
\vspace{-0.1in}
\paragraph{Algorithm 3: identification of sibling nodes.}
\bee
\item[] \textbf{Given}:  $(\hat A^{l+1}, C^{l+1})$.  
\item[] \textbf{Initialization}: $C^{l+1'} : = C^{l+1} - \diag\left( \left(\textbf 1 \otimes \textbf I_3 \right)^{\sf T} C^{l+1} \right)$.

\item[] \textbf{Construct} the set $N_{n-l}$ of neighbors of a new hidden node $n-l$: $i,j\in N_{n-l}$ if and
		only if the $3\times 3$ matrices $C^{l+1'}[i, m]$ and $C^{l+1'}[j,m]$ satisfy
\begin{align*}
C^{l+1'}[i, m] & \ = \ \gamma(i,j)\, C^{l+1'}[j,m], \qquad \forall m\neq i, j, \ m \in C^{l+1'}
\end{align*}
for some nonsingular matrix $\gamma(i,j)\in\mathbb C^{3\times 3}$ that does not depend on $m\neq i, j$.
	
\item[] \textbf{Return}: $N_{n-l}$ for identity of nodes in $y_l$; new hidden node $n-l$ in $\hat A^l$.
\eee

The property in Theorem \ref{thm:SameHiddenNode} is similar to the sibling grouping property in
\cite[Lemma 4]{ChoiWillsky2011} where $C^k[i, m]C^k[j,m]^{-1} = \gamma(i,j)$ in Theorem \ref{thm:SameHiddenNode} plays the role of the difference $\Phi_{ijm} := d_{im} - d_{jm}$ of 
information distances.  The recursive grouping
procedure of \cite{ChoiWillsky2011} computes a latent tree graphical model (probability distribution with a
conditional product structure) by iteratively determining, starting from leaf nodes, their unique parent nodes 
based on information distances and computing the information distances from each
new parent node to all nodes that have been identified in each iteration.  
This procedure might be adaptable to our 
context using the matrix $C^k$, if the entries $C^{k-1}[h, j]$ between a new 
parent node $h$ and all nodes $j\in C^k$ can be computed in each iteration.

\subsection{Appendix: Schur complement}
\label{sec:Assumption3; subsec:SchurComplement}

We review some standard properties of Schur complement that are used in the identification method.
Consider a complex matrix $A = \begin{bmatrix} A_{11} & A_{12} \\ A_{21} & A_{22} \end{bmatrix}$
where $A_{11}$ and $A_{22}$ are square submatrices.  If $A_{11}$ is nonsingular then the
\emph{Schur complement of $A_{11}$} of $A$ is the following matrix:
\begin{align*}
A/A_{11} & \ := \ A_{22} - A_{21} A_{11}^{-1} A_{12}
\end{align*}
When $A_{11}$ is nonsingular, $A$ is nonsingular if and only if its Schur complement $A/A_{11}$
is nonsingular, in which case $A^{-1}$ can be expressed in terms of the inverses of $A_{11}$ and $A/A_{11}$:
\begin{subequations}
\begin{align}
A^{-1} & \ = \ \begin{bmatrix}
	B_{11} & - A_{11}^{-1} A_{12} (A/A_{11})^{-1}  \\
	- (A/A_{11})^{-1} A_{21} A_{11}^{-1}  	& (A/A_{11})^{-1}
\label{eq:A-1.a}
\end{bmatrix}
\end{align}
where $B_{11} := A_{11}^{-1} + A_{11}^{-1}A_{12} (A/A_{11})^{-1} A_{21} A_{11}^{-1}$.
If $A_{22}$ is nonsingular then the
\emph{Schur complement of $A_{22}$} of $A$ is
\begin{align*}
A/A_{22} & \ := \ A_{11} - A_{12} A_{22}^{-1} A_{21}
\end{align*}
When $A_{22}$ is nonsingular, $A$ is nonsingular if and only if its Schur complement $A/A_{22}$
is nonsingular, in which case $A^{-1}$ can be expressed in terms of the inverses of $A_{22}$ and $A/A_{22}$:
\begin{align}
A^{-1} & \ = \ \begin{bmatrix}
	(A/A_{22})^{-1}  & - (A/A_{22})^{-1} A_{12} A_{22}^{-1}  \\
	- A_{22}^{-1} A_{21} (A/A_{22})^{-1}  	& B_{22}
\label{eq:A-1.b}
\end{bmatrix}
\end{align}
where $B_{22} := A_{22}^{-1} + A_{22}^{-1} A_{21} (A/A_{22})^{-1} A_{12} A_{22}^{-1}$.
\label{eq:A-1}
\end{subequations}

\newpage

\bibliographystyle{unsrt}
\bibliography{PowerRef-201202}

\end{document}